\documentclass[12pt]{article}

\usepackage{setspace}
\usepackage{amsmath,amsfonts}
\usepackage{theorem}

\usepackage{epsfig}
\usepackage{xcolor}

\usepackage{lscape}

\usepackage{multirow}

\usepackage{rotating}

\usepackage{tikz}

\usepackage{makecell}
\newcolumntype{x}[1]{>{\centering\arraybackslash}p{#1}}
\usepackage[hyphens]{url}
\usepackage[breaklinks]{hyperref}
\hypersetup{
    colorlinks=true,       % false: boxed links; true: colored links
    linkcolor=red,          % color of internal links
    citecolor=blue,        % color of links to bibliography
    filecolor=magenta,      % color of file links
    urlcolor=cyan           % color of external links
}

\usepackage[sort&compress]{natbib}
\bibpunct{(}{)}{,}{a}{,}{,}

\theoremstyle{plain} { \theorembodyfont{\rmfamily}

}
\newtheorem{proposition}{Proposition}[section]
\newtheorem{lemma}{Lemma}[section]

\newtheorem{theorem}{Theorem}[section]

\newcommand{\qed}{\hfill Q.E.D}
\def\Halmos{\mbox{\quad$\square$}}
\newenvironment{proof}{\vspace{1ex}\noindent{\bf Proof}\hspace{0.5em}}
    {\Halmos\vspace{1ex}}

\newcommand{\cF}{\mathcal{F}}
\newcommand{\cT}{\mathcal{T}}
\newcommand{\cP}{\mathcal{P}}

\newcommand{\bZ}{\mathbb{Z}}
\newcommand{\bR}{\mathbb{R}}
\newcommand{\bE}{\mathbb{E}}
\newcommand{\bN}{\mathbb{N}}
\newcommand{\prob}{\mathbb{P}}

\newcommand{\vecx}{\mathbf{x}}
\newcommand{\vecy}{\mathbf{y}}

\begin{document}
\title{When to Quit Gambling, if You Must!\thanks{We thank Nick Barberis for a long list of constructive comments on a previous version of the paper that have led to a much improved version.}
}

\author{Sang Hu
\thanks{School of Data Science, The Chinese University of Hong Kong, Shenzhen, China 518172. Email: \texttt{husang@cuhk.edu.cn}.
This author would like to acknowledge the funding of National Natural Science of China (Grant No. 11901494).}
\and Jan Ob{\l}{\'o}j
\thanks{Mathematical Institute, the Oxford-Man Institute of Quantitative Finance and St John's College, University of Oxford, Oxford, UK. Email: \texttt{Jan.Obloj@maths.ox.ac.uk}.
Part of this research was completed whilst this author was visiting CUHK and he is grateful for the support from the host. He also gratefully acknowledges support from ERC Starting Grant {\sc RobustFinMath} 335421.}
\and Xun Yu Zhou
\thanks{Department of Industrial Engineering and Operations Research, Columbia University, New York, New York 10027. Email: \texttt{xz2574@columbia.edu}.
This author gratefully acknowledges financial supports through start-up grants at both University of Oxford and Columbia University, and through Oxford--Nie Lab for Financial Big Data and the Nie Center for Intelligent Asset Management.}
}

\maketitle

\begin{abstract}
We develop an approach to solve \citet{Barberis2012:Casino}'s casino gambling model in which a gambler whose preferences are specified by the cumulative prospect theory (CPT) must decide when to stop gambling by a prescribed deadline. We assume that the gambler can assist their decision using an independent randomization, and explain why it is a reasonable assumption. The problem is inherently time-inconsistent due to the probability weighting in CPT, and we study both precommitted and na\"ive stopping strategies. We turn the original problem into a computationally tractable mathematical program, based on  which we derive an optimal precommitted rule which is randomized and Markovian. The analytical treatment enables us to make several predictions regarding a gambler's behavior, including that with randomization they may enter the casino even when allowed to play only once, that whether they will play longer once they are granted more bets depends on whether they are in a gain or at a loss, and that it is prevalent that a naivit\'e never stops loss.

%{\color{orange} Numerical examples show that loss-exit type of stopping strategy exhibits more randomization than gain-exit type, and randomization are mostly taken just before the terminal time.}
\bigskip

\end{abstract}

\section{Introduction}\label{se:Introduction}

\citet{Barberis2012:Casino} proposes a casino gambling model in the framework of
\cite{TverskyKahneman1992:CPT}'s cumulative prospect theory (CPT) to study the optimal timing to  quit gambling and leave the casino. The author has derived two key economic insights: (1) a CPT gambler may be willing to enter the casino even though its bets offer neither positive expected values nor skewness because, by implementing an appropriate stopping strategy, he would be able to build a positively
skewed final winning amount that would be favored by the underlying probability weighting in CPT; (2) there is an inherent time-inconsistency due to the dynamically changing strength of probability weighting on a same event: the gambler may deviate completely from his initial stopping strategy as he gambles along, and his eventual stopping behavior depends on whether he is aware of this time-inconsistency and whether he is able to commit his original plan.\footnote{\citet{Barberis2012:Casino} discusses three types of gamblers, following the original classification of \cite{Strotz1955}: a {\it na\"ive} gambler who is unaware of the time-inconsistency and changes his strategy all the time; a {\it precommitted} gambler  who is aware of time-inconsistency and can commit to his initial plan; and a {\it sophisticated} gambler who is aware of time-inconsistency yet unable to commit, and at each time takes the future selves' disobedience into account when devising an optimal strategy.}

It is, however, not an objective of \citet{Barberis2012:Casino} to develop a { general} approach to {\it solve} the casino model he puts forward. \citet{Barberis2012:Casino}  acknowledges that
the nonlinear probability weighting involved in CPT
makes it ``very difficult" to solve the problem analytically, and ``the problem has no known analytical
solution for general $T$" (p. 42), where $T$ is an exogenously given number of bets the gambler can maximally have.
Instead, \citet{Barberis2012:Casino} uses an {\it exhaustive search} to find a solution; namely, he enumerates all the possible {\it Markovian} stopping strategies, calculates the CPT value of each of them and finds the one that achieves the highest CPT value as the optimal strategy. As one would expect, this approach works only for smaller $T$, as the number of admissible Markovian strategies is {\it exponential} in $T^2$.\footnote{The number of nodes is $\frac{T(T+1)}{2}$  in a binomial tree of horizon $T$, and at each node there is a binary choice of \{stop, continue\}. Hence the total number of strategies is $2^{\frac{T(T+1)}{2}}$.}
 \citet{Barberis2012:Casino} solves the problem with $T=5$.\footnote{We ran exhaustive search on a desktop with Intel Core i5-4590/CPU 3.30GHz/RAM 8.00GB for different $T$'s while keeping the other parameters same as \citet{Barberis2012:Casino}'s. The running times for $T = 5, 6, 7 $ were 39 seconds, 771 seconds and 27 hours, respectively. We were unable to obtain the solution for $T=8$ due to out of storage, with the running time estimated to be 300 days.}

Naturally, to better understand the implications of a model it is important to have a {\it systematic approach}
to solve it, not necessarily in an analytically closed form, but in a computationally efficient way.\footnote{Consider, e.g., the simplex method for linear programs or the dynamic programming formulation for optimal control problems.} Not only can we then obtain optimal solutions for arbitrary values of parameters, but we may gain (likely more profound) economic insights from the model by post-optimality analyses such as comparative statics.
The main technical hurdle to solve the casino model is probability weighting, as pointed out by \citet{Barberis2012:Casino}. The two main approaches in the classical optimal stopping theory -- dynamic programming (variational inequalities) and martingale method -- both fail under probability weighting: the former does because of the time-inconsistency, and the latter does because of the absence of a ``tower property'' with respect to the {\it weighted} probability.

\citet{HeHuOblojZhou2014:OptimalCasinoBettingwhyluckycoinandgoodmemoryareimportant,HeHuOblojZhou2017:Twoexplicitembedding,HeHuOblojZhou2015:OptimalCasinoBettingstrategiesofprecommittedandnaivegamblers} are probably the first series of papers that aim at an {\it analytical treatment} of the casino model, albeit in the {\it infinite} time horizon.\footnote{Here by ``analytical treatment" we mean an optimization analysis not based on heuristics or on brute force such as an exhaustive search.}  The main idea of these papers consists of two deeply intertwined steps: (1) search the optimal {\it probability  distribution} of the final winning/losing amount upon leaving the casino instead of the optimal {\it time} to leave; (2) once the optimal distribution is found, recover the optimal time that generates it.\footnote{This idea was first put forth by \cite{XuZhou2012:OptimalStoppingunderProbabilityDistortion} for a continuous-time optimal stopping model featuring probability weighting. There is considerable difficulty to adapt this idea to the discrete-time setting.} Both steps call for a complete characterization of the set of all the admissible distributions, and the second step is the discrete-time version of the eminent {\it Skorokhod embedding theorem} which in the casino setting is solved in \citet{HeHuOblojZhou2017:Twoexplicitembedding}. The main thrust to make this idea work is to permit {\it randomization}, namely the gambler can flip an independent,  possibly biased,  coin to assist his decision each step of the way.
The probabilities of the head of the coin are {\it endogenous} and dynamically changing; thus they are {\it part} of the final solution.\footnote{Mathematically, randomization {\it convexifies} the aforementioned set of admissible distributions. Hence, \cite{HeHuOblojZhou2014:OptimalCasinoBettingwhyluckycoinandgoodmemoryareimportant,
HeHuOblojZhou2015:OptimalCasinoBettingstrategiesofprecommittedandnaivegamblers} use
randomization  as a {\it technical} tool to make the Skorokhod embedding work, but fall short of explaining, economically,  why people would randomize and how exactly they do it. The present paper offers  discussions on these issues; see Subsection \ref{randomize}.}
 The randomization of decisions is a key feature when studying agents with CPT preferences, as discussed independently by \cite{HendersonHobsonTse14}.
\citet{HeHuOblojZhou2014:OptimalCasinoBettingwhyluckycoinandgoodmemoryareimportant, HeHuOblojZhou2015:OptimalCasinoBettingstrategiesofprecommittedandnaivegamblers} also allow path-dependent strategies, that is, the stopping decision is made based on the whole betting history instead of just the current winning/losing amount. They further show that allowing path-dependent strategies or randomized ones {\it strictly} improve the optimal CPT values. Based on  these analyses, \citet{HeHuOblojZhou2015:OptimalCasinoBettingstrategiesofprecommittedandnaivegamblers}
turn the casino model into an infinite dimensional mathematical program that can be solved fairly efficiently. Most of the gambler's behaviors -- those of a precommitter and of a naivet\'e -- implied from the solutions reconcile qualitatively with \citet{Barberis2012:Casino}'s results; but there are also new findings. For example,
it is revealed that, for most empirically relevant CPT  parameter estimates, a precommitted gambler lets gain run while stops loss, but
 a na\"ive one almost surely does not stop at {\it any} loss level.

As noted, \cite{HeHuOblojZhou2014:OptimalCasinoBettingwhyluckycoinandgoodmemoryareimportant,
HeHuOblojZhou2015:OptimalCasinoBettingstrategiesofprecommittedandnaivegamblers}
deal with the infinite horizon gambling model.
There are important reasons to study the {\it finite} horizon model under CPT preferences, despite the existing results for the infinite horizon counterpart. Conceptually, the finite horizon problem approximates the reality much better, as a gambler clearly will not be able to play arbitrarily and indefinitely long. Also, the original work of \citet{Barberis2012:Casino} considers $T=5$ and hence we need to solve the finite horizon model in order to be able to make a direct comparison. It is worth noting that solutions to the finite horizon case can {\it not} be recovered from those of the infinite horizon case by a simple truncation: if
 $\tau$ is optimal for the latter, then, typically, $\tau\land T$ will not be optimal for the former.

Methodologically, the finite horizon case is significantly more complex. It is well acknowledged that optimal stopping in a finite horizon is {\it fundamentally} more difficult than its infinite horizon counterpart, mainly because value function of the former has both time and spatial variables while the latter has only spatial variables. %(and, hence, in the continuous-time setting, the former involves a PDE while the latter an ODE).
In the infinite time horizon setting in which the accumulated winning/losing amount
is modelled by a symmetric random walk $S$, \citet{HeHuOblojZhou2017:Twoexplicitembedding} show that for any centered probability measure $\mu$ on the set of integers $\mathbb{Z}$, there exists a {\it randomized} stopping time $\tau$ such that $ S_\tau$'s distribution is $\sim \mu$.\footnote{In the terminology of  Skorokhod embedding theorem, we say $\tau$ {\it embeds} $\mu$ in $S$.} As discussed previously, this is the key theoretical underpinning for the new approach. Unfortunately, this result is no longer true if the stopping time is  constrained  by a pre-specified deadline. Indeed, {\it additional} conditions are required for measures that can be embedded by
uniformly bounded stopping times. One of the contributions of this paper is to identify explicitly these conditions, which in turn enables us to
reformulate the original casino model into a mathematical program whose number of constraints is of the order of $T$ and, hence, can be efficiently solved.

Once we have an algorithm to solve the gambling model for {\it any} parameter values, we will then be able to first compare our results with those of \citet{Barberis2012:Casino}'s. In particular, we compute for exactly the same case that is solved and discussed in \citet{Barberis2012:Casino} with  $T=5$. The respective stopping strategies for a precommitter are identical except in two time--state instances in which our decisions are to stop with very small probabilities (0.00864 and 0.0368 respectively) whereas Barberis' are just to continue. Qualitatively, both strategies are of the so-called {\it loss-exit} type, namely, they continue in gains but stop after having accumulated sufficient amounts of losses. With randomization, our optimal CPT value  improves, if slightly, over Barberis'. %\footnote{
%{\hl We also revisited the $T=6$ example from
%\citet{Barberis2012:Casino} considered in
%\citet{HeHuOblojZhou2014:OptimalCasinoBettingwhyluckycoinandgoodmemoryareimportant}. Therein, a randomized strategy, found by trial and error, led to the value function $V=0.250702$ compared with $V=0.25044$ for the best non-randomized strategy in \cite{Barberis2012:Casino}. Using our algorithm, we see that the best randomized strategy actually gives $V=0.257$. It is still a loss-exit type and it stops at node $(0,0)$ with probability $0.201$, node $(5,1)$ with $0.436$, node $(5,3)$ with $0.0292$ and node $(5,5)$ with $0.0113$.
%}
Likewise,  the respective na\"ive strategies are the same save for one time--state instance in which ours is  to stop with a probability of 0.179 while Barberis' is to continue. Our solution, however, enables us to look beyond the relatively short horizon  of $T=5$. Indeed, we carry out numerical experiments for different values of $T$ up to $T=20$, and discover that the interplay between the utility function, probability weighting and loss aversion dictates various gambling behaviors.

Note that our analytical treatment relies on the introduction of randomization in our model, as randomization convexifies the optimization problem. \citet{Barberis2012:Casino} does not allow randomization, for which our approach would fail. However, our solution would provide a well-founded relaxation {\it heuristic} for solving a casino model without randomization: we first relax the problem by introducing randomization, and then, for each time-state pair, round up or round down the probability of stopping to 1 (which means stop) or to 0 (which means continue).

Our approach makes it possible to analyze and understand the impacts of some key attributes of the model, which we believe is the most important contribution of this paper. For example, \citet{Barberis2012:Casino} argues that %, especially that of the length of the horizon, $T$.
a gambler may be willing to enter a casino because, by implementing a loss-exit strategy, he may be able to generate a positively skewed
probability distribution of the final accumulated gain/loss which has a positive
 CPT preference value. However, he will  need to spend time building  such
a skewed distribution, which requires a sufficiently large $T$. We show, however, that for the same gambler who would have demanded a long horizon for agreeing to enter the casino, will enter even if he is allowed to play only {\it once} (i.e., $T=1$), provided that he can flip a coin. The reason for this is that, with
randomization, the gambler can {\it design} a coin {\it right away} with the desired skewed distribution,
saving all the time otherwise needed to reach that distribution. Another
insight is about the value of time: how much is time on your hands worth?
Specifically, we examine the question of what a gambler would do should he be allowed to stay one more period than previously
agreed. Would he always take advantage of this extended time horizon and {\it actually} play
the additional round? It turns out that there is no uniform answer to the question -- it depends crucially on
whether the gambler is currently in a gain or at a loss.

We also study the behaviors of a na\"ive gambler with various parameter specifications
and a longer time horizon ($T=20$). We find that, unless  he does not enter the casino, his behavior is
consistently of gain-exit type, i.e., he stops gain but lets loss run, reminiscent of the {\it disposition effect} in security trading \citep{OdeanT:98de}.
In particular, he never stops loss and gambles ``until the bitter end".
This gamble-until-the-bitter-end behavior is derived by
\citet{EbertStrack2012:UntilTheBitterEnd} in a
model in which a na\"ive  gambler can construct {\it arbitrarily} small random payoffs. Because he prefers ``skewness in the small", he  never stops gambling.
\cite{HendersonHobsonTse14}, employing the approach developed  in \citet{XuZhou2012:OptimalStoppingunderProbabilityDistortion}, investigate a stylized continuous-time model and show that a na\"ive gambler may stop  with a positive probability if she is allowed to randomize, which complements and counters the findings in \citet{EbertStrack2012:UntilTheBitterEnd}.
Both \citet{EbertStrack2012:UntilTheBitterEnd} and \cite{HendersonHobsonTse14} rely on the crucial feature of their models that allows the gambler to construct arbitrarily small random payoffs. This feature is absent
in our discrete-time model, in which the gambler cannot construct strategies with arbitrarily small random payoffs due to  the minimal stake size fixed to be \$1.
Hence, their results are not applicable to our setting. Our finding therefore suggests that the gamble-until-the-bitter-end phenomenon
is probably more prevalent of a naivit\'e's behavior.

The paper proceeds  as follows. In Section \ref{se:TheModel}, we  formulate a  casino gambling model under CPT as an optimal stopping problem and discuss why we allow randomization in our model.
In Section \ref{se:RandomizedRootstoppingtime}, we develop the key step in our approach to solve the gambling model: characterizing the set of probability distributions of all possible accumulated
winning/losing amounts upon leaving the casino. In Section \ref{se:finiteprogram}, we present a mathematical program that is equivalent to the casino model, and then
report the results of a numerical example which is studied in \citet{Barberis2012:Casino}. We discuss about various implications and predictions of our model in Section \ref{se:Discussion}. Finally, we conclude the paper by Section \ref{se:concludingremark}. Proofs are placed in Appendices.

\section{The Model}\label{se:TheModel}
In this section we first highlight the key ingredients of \citet{TverskyKahneman1992:CPT}'s CPT, then formulate the casino gambling model in a finite time horizon as an optimal stopping problem, and finally discuss about the reasons why we make randomization available in our model.

\subsection{Cumulative prospect theory}
%We first briefly review the cumulative prospect theory of \citet{TverskyKahneman1992:CPT}. Readers who are familiar with the cumulative prospect theory may skip this part of section.
%In classical economics theories, one evaluates uncertain payoffs according to the expected utility framework. However, expected utility theory (EUT) has been challenged by many famous paradoxes. In response to this, some non-expected utility theories such as behavioral ones have been proposed. In the followings,
%We briefly highlight the key ingredients of  the {\it cumulative prospect theory} (CPT) proposed by \citet{TverskyKahneman1992:CPT}.
%There are important features differentiating CPT from the traditional EUT. First,
In CPT, a {\it utility} (or {\it value})  {\it function} $u(\cdot)$ depends on  a {\it reference point} $k$ in wealth that divides {\it gains} and {\it losses}. %The values above $k$ are called gains and those below $k$ are losses.
An agent derives the utility from gains and losses, rather than from the absolute amount of wealth itself. The utility function is
      \begin{align*}
      u(x) = \begin{cases}
      u_+(x-k) ,  & x \geq k,\\
      -\lambda u_-(k-x) , &  x < k,
      \end{cases}
      \end{align*}
where $u_+(\cdot)$ and $u_-(\cdot)$ are both {\it concave} functions and $\lambda > 1$. This renders an overall {\it S-shaped} utility function $u(\cdot)$ that is concave (risk-averse) in the gain region $x \geq k$ and convex (risk-loving) in the loss region $x < k$. Moreover,
$\lambda >1$ yields that, for the same magnitude of a gain and a loss, the agent is more sensitive to the latter, a notion termed {\it loss aversion}. \citet{TverskyKahneman1992:CPT} propose the following parametric form of $u(\cdot)$:
\begin{align}\label{eq:utility}
        u(x) = \begin{cases}
        (x-k)^{\alpha_+} , & x \geq k,\\
        -\lambda (k-x)^{\alpha_-} , & x < k,
        \end{cases}
      \end{align}
where $0 < \alpha_\pm \leq 1$ and $\lambda > 1$; see the left panel of Figure \ref{fig:distortion_utility} for an illustration of this type of functions.% signifies that $u(\cdot)$ is concave in the gain domain and convex in the loss one, and $\lambda \geq 1$ is the degree of loss aversion.

In CPT  there are also {\it probability weighting}  (or {\it distortion}) {\it functions} $w_+(\cdot)$ and $w_-(\cdot)$ applied to gains and losses respectively. %An effect of the probability distortion is that small probabilities of both extremely good and bad events are overweighted.
An {\it inverse S-shaped} weighting function is first concave  and then convex in the domain of probabilities. Such a weighting function overweights both tails of a probability distribution, reflecting  the exaggeration of extremely small probabilities of extremely large gains and losses. \citet{TverskyKahneman1992:CPT} suggest a parametric  form of a weighting function $w(\cdot)$:
      \begin{align}\label{eq:distortion}
        w(p) = \frac{p^\delta}{(p^\delta+(1-p)^\delta)^{\frac{1}{\delta}}},
      \end{align}
where $0 <\delta \leq 1$; see the right panel of Figure \ref{fig:distortion_utility} for an illustration. Note that $\delta = 1$ means that no weighting  is applied.

      \begin{figure}
      \begin{minipage}[t]{0.5\textwidth}
      \includegraphics[width = \textwidth]{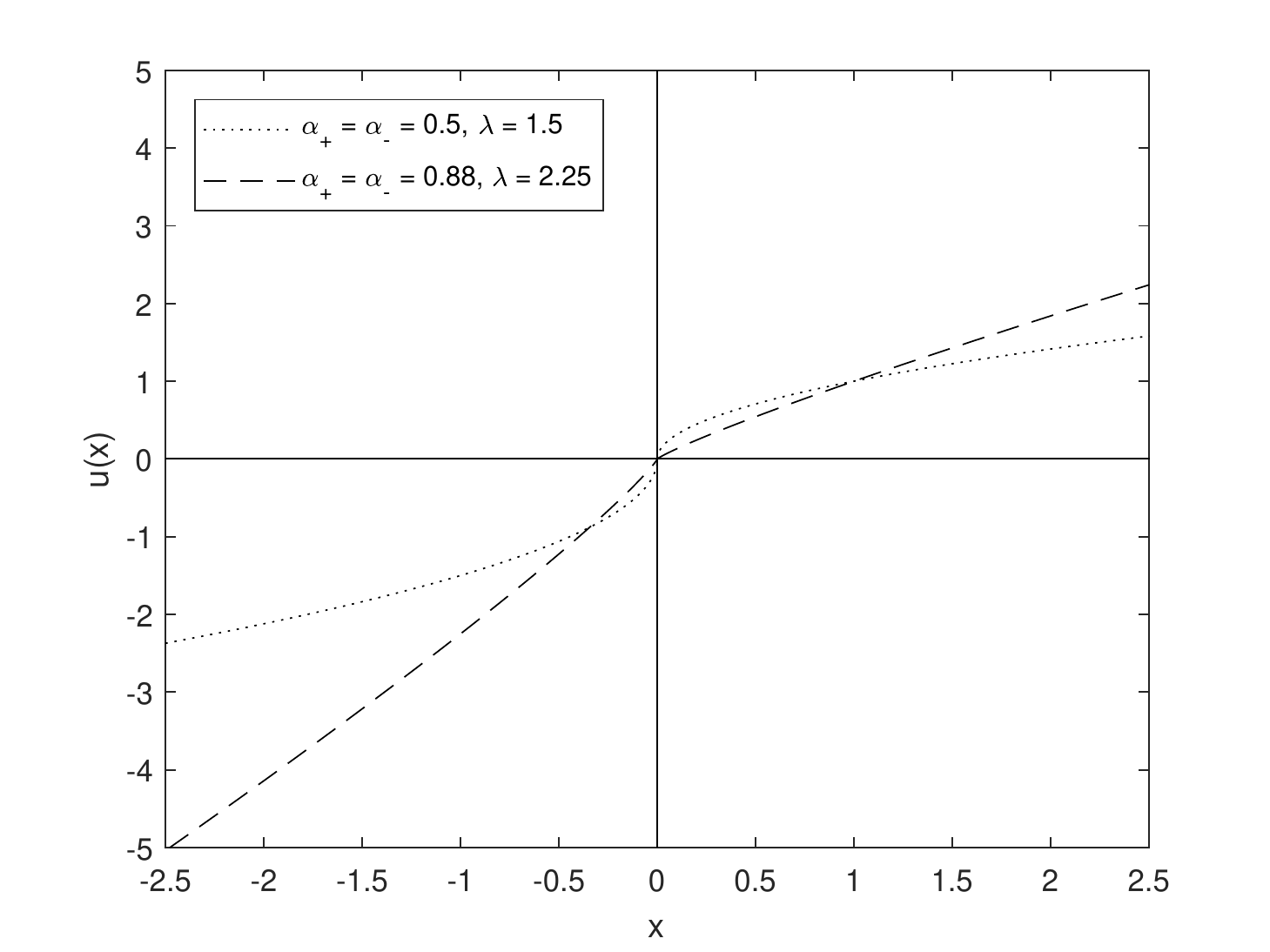}
      \end{minipage}
      \begin{minipage}[t]{0.5\textwidth}
      \includegraphics[width = \textwidth]{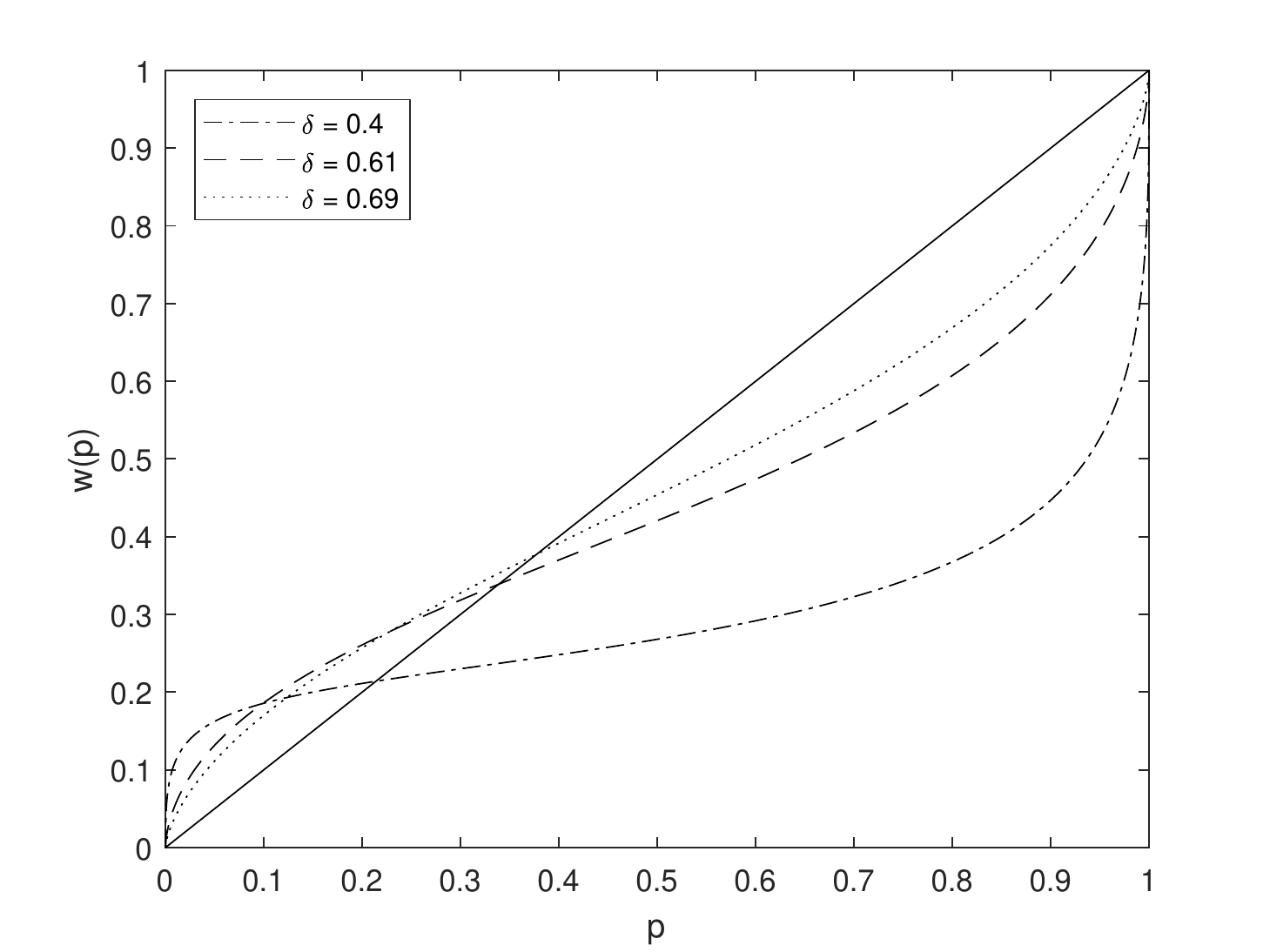}
      \end{minipage}
      \caption{The left panel graphs two S-shaped utility functions \eqref{eq:utility} with $\alpha_+ = \alpha_- = 0.5, \lambda = 1.5$ and $\alpha_+ = \alpha_- = 0.88, \lambda = 2.25$, respectively. The right panel depicts three inverse S-shaped probability weighting functions \eqref{eq:distortion} with $\delta = 0.4$, $\delta = 0.61$, and $\delta = 0.69$, respectively.}
      \label{fig:distortion_utility}
      \end{figure}

\subsection{Formulation of a casino gambling model}

We now reformulate \citet{Barberis2012:Casino}'s model of casino gambling in a finite time horizon $[0,T]$, where $T \in \mathbb{Z}^{+}:=\{1,2,3,...\}$ is given. The gambling process proceeds as follows. At time 0, the gambler is offered a fair bet, e.g., one with a roulette wheel: win or lose \$1 with equal probability.\footnote{As in \citet{Barberis2012:Casino}, we assume in this paper that the gamble is fair. It will not affect the main economic findings and implications of our  results.
A model of unfair games is more technical, and is left for a future study.}  If the gambler decides not to play the bet, then he will not even enter the casino. If the gambler enters and takes the bet, then the bet outcome is played out at time 1, leading to either a win or a loss of \$1 at time 1. At that time the gambler is offered the same bet again and he decides whether to play. If he declines the bet, then the game is over and the gambler leaves the casino with \$1 gain or loss. This process continues in the same fashion until time $T$: the bet is offered and played out repeatedly until either the first time the gambler declines the bet, or at time $T$ when the gambler {\it must} quit gambling and leave. %The gambler leaves the casino with all her prior gains and losses accumulated.
The {\it accumulated} gain/loss process can be represented as a binomial tree; see
Figure \ref{fig:gainlossprocess}.  Therein, each node is marked by a pair $(t,x)$, where $t\in\bN:=\{0,1,2, \ldots\}$ stands for the time and $x\in \mathbb{Z}:=\{0,\pm 1,\pm 2,\ldots\}$ the amount of cumulative gains or losses. For example, the node $(2,-2)$ signifies a cumulative loss of \$2 at  time 2. The process has a terminal time $T$ but the gambler may quit at some earlier time $\tau\leq T$.

\begin{figure}
\centering
\includegraphics[width = 100mm]{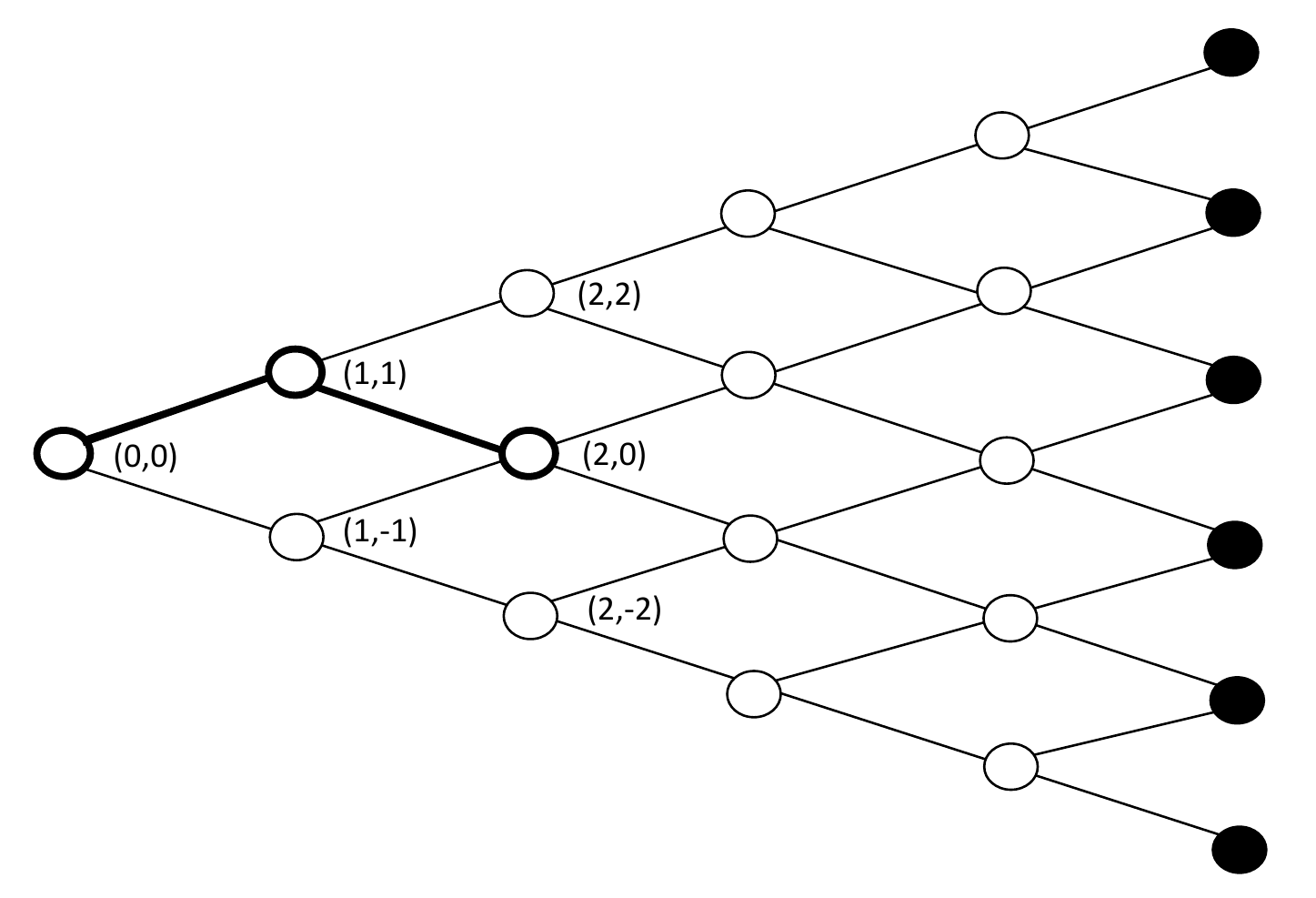}
\caption{\small
The gain/loss  binomial tree with $T = 5$. %Each node is marked by a pair $(t,x)$, where $t$ stands for the time and $x$ the amount of cumulative gains or losses. For example, the node $(2,-2)$ signifies a total loss of \$2 at time 2. In this example where $T = 5$,
The gambler must  leave the casino by time 5, which is represented by the black nodes.}\label{fig:gainlossprocess}
\end{figure}

%Suppose the reference point $k = 0$. Denote a discrete gain/loss random variable by
%\begin{align*}
%(-T,p_{-T}),...,(-1,p_{-1}),(0,p_0),...,(T,p_T)
%\end{align*}
%where $p_n$ is the probability of gaining $n$ dollars and $p_{-n}$ is the probability of losing $n$ dollars, $n \in \mathbb{N}$. Moreover, $\sum_{n=1}^{\infty} p_n + p_0 + \sum_{n=1}^{\infty} p_{-n}= 1$. The CPT value of this random gain/loss is
%      \begin{align}\label{eq:CPTvalue}
%      \sum_{n=1}^{\infty} u_+(n) \Big(w_+(\sum_{j=n}^{\infty} p_j)-w_+(\sum_{j=n+1}^{\infty} p_j)\Big) - \lambda  \sum_{n=1}^{\infty} u_-(n) \Big(w_-(\sum_{j=n}^{\infty} p_{-j})-w_-(\sum_{j=n+1}^{\infty} p_{-j})\Big).
%      \end{align}
%Note that we differentiate the gain part from the loss part by applying possibly different utility functions and distortion functions. A CPT gambler is risk-averse in the gain domain and risk-seeking in the loss one, so $u_+(x)$ and $u_-(x)$ are both concave. Moreover, he overweights the tails of a distribution, so $w_+(p)$ and $w_-(p)$ are concave for small values of $p$.

%\subsection{Optimal stopping}

%Throughout, $t\in \bN=\{0,1,2, \ldots\}$ is the discrete time index.
The gain/loss binomial tree $S=(S_t : t \in \bN)$  is a {\it standard symmetric random walk} (SSRW)  defined on a filtered probability space $(\Omega, \cF, \prob; (\cF_t)_{t\in \bN})$. We assume the probability space is rich enough to support an $\cF_0$-measurable random variable $\xi$ that is uniformly distributed on $[0,1]$ and {\it independent} of $S$. %The random walk describes the gain/loss process relative to the initial wealth level, i.e., the constant reference point $k$.

Suppose the gambler quits gambling at a random time $\tau\in[0,T]$. Then, with the reference point being his {\it initial} wealth before he enters the casino,
the CPT value of his wealth  upon leaving is
    \begin{equation}\label{prob:InfiniteDimProgramming:original}
    \begin{aligned}
    V(S_\tau) : = & \sum_{n=1}^T u_+(n)\left[w_+\left(\mathbb{P}(S_\tau \geq n)\right)-w_+\left(\mathbb{P}(S_\tau \geq n+1)\right)\right]\\
    & \quad - \lambda \sum_{n=1}^T u_-(n)\left[w_-\left(\mathbb{P}(S_\tau \leq -n)\right)-w_-\left(\mathbb{P}(S_\tau \leq -n-1)\right)\right].
    \end{aligned}
    \end{equation}
    Throughout this paper we assume that both $u_+(\cdot)$ and $u_-(\cdot)$ are concave and both $w_+(\cdot)$ and $w_-(\cdot)$ are inverse S-shaped.
%Note this definition implicitly assumes the reference point in calculating the CPT value is the {\it initial} wealth before he casts any bet.
The gambler needs to determine the optimal time to quit  and leave the casino: such a stopping (exit) strategy $\tau$ is made at $t=0$ to  maximize $V(S_\tau)$ among all admissible strategies. Note that, due to probability weighting, the problem is inherently {\it time-inconsistent}; so $\tau$ is optimal only at $t=0$ in the sense of a {\it precommitted} strategy; it may no longer be optimal from the vantage point of any later time $t>0$.
%Such a planned $\tau$ is a pre-commitment strategy determined at time $0$.

We now define precisely the set of {\it admissible} stopping strategies
\begin{align*}
{\mathcal{T}}_T := \left\{ \tau\in[0,T]: \tau \text{ is an } (\cF_t)_{t\in \bN}\text{-stopping time}\right\}.
\end{align*}
So a decision whether or not to quit at time $t\in[0,T]$ depends on all the information up to $t$.
%This is a standard formulation in optimal stopping theory (see, e.g., \citet{Shiryaev:78osr}).
In particular, path-dependent strategies are admissible.
 %\footnote{Path-dependence, or history-dependence, is inherent in our everyday decision makings. Path-independence or Markov property, in most cases, is an assumption that offers mathematical convenience to dramatically reduce the dimension of the underlying problem.}
 Moreover, $\cF_0$ -- and hence all $\cF_t$ -- contains the information about $\xi$, a uniform random variable independent of $S$. Using $\xi$, we can define countably many binary random variables which are mutually independent and also independent of $S$. In consequence, an admissible strategy may involve randomization by tossing a (generally biased) coin. In the next subsection we will outline the
rationale behind allowing randomized  strategies.

The gambler's problem is
\begin{equation}\label{originaloptimization}
\begin{array}{rl}
\underset{\tau \in {\mathcal{T}}_T}{\max} & V(S_\tau) \;.
\end{array}
\end{equation}

\subsection{Randomization}\label{randomize}

In our model (\ref{originaloptimization}), the filtration $(\cF_t)_{t\in \bN}$ includes the information based on a uniform random variable that is {\it independent} of the underlying random walk. This means that we allow the gambler to {\it assist} his decision by flipping an independent,   most likely biased, coin at each node.\footnote{Note any random variable taking a continuum of values, including the Bernoulli random variable, can be generated from a uniform random variable.}
We now discuss the rationale behind making  this randomization available in our model.

%First of all, preference for randomization is observed in daily life and in different cultures. Examples include last-minute deals by flight booking apps, ``sushi omakase" (you entrust yourself to a sushi chef to choose the ingredients and presentations of your sushi plate), and ``fukubukuro" (grab bags filled with unknown and random contents). The practice of ``drawing divination sticks", popular in Chinese
%culture even today, is a vivid example of deliberate randomization in addition to religious reasons. When people are reluctant or unable to make their own decisions
%on important matters (usually marriage, school, or even home move), they
%go
%to a temple, pray and draw divination sticks, and follow whatever words on the sticks tell them to do.

First of all, randomization  is related  to the accommodation of path-dependence.   It is practically more reasonable, and indeed necessary,  to consider   path-dependent strategies than mere Markovian ones. How has the gambler arrived at a current amount (say \$500) -- whether he has won big first and then lost most of them, or he has gradually accumulated this amount by many small wins -- clearly  might affect his decision.
As a matter of fact, almost all our decisions in life are made based on all the information, past and present, rather than just on the current state of affair.
After all, being Markovian is just a mathematical assumption and convenience that aims to dramatically reduce the dimension of the underlying problem or, in a continuous-time setting, turns an infinite dimensional problem into a finite dimensional one.
Path-dependence  is also a standard formulation in optimal stopping theory (see, e.g., \citealp*{Shiryaev:78osr}) and indeed in general stochastic control theory (see, e.g., \citealp*{YZ1999}). Now, \citet{HeHuOblojZhou2014:OptimalCasinoBettingwhyluckycoinandgoodmemoryareimportant}
shows that any non-randomized, path-dependent stopping time is equivalent to a randomized, {\it Markovian} stopping time in the sense that  both attain the same CPT value.\footnote{While this result has been obtained for the infinite horizon model, the underlying argument is exactly the same for the finite horizon case. The intuition is that, due to the independent increments of a random walk, considering all the past information can be achieved by
randomizing at the current state.} As a result, we can consider randomization {\it in lieu of} considering the past information.

However, \citet{HeHuOblojZhou2014:OptimalCasinoBettingwhyluckycoinandgoodmemoryareimportant} also show the converse is {\it not} true, namely, a randomized Markovian strategy may not be replicated by a non-randomized, path-dependent strategy, and the optimal CPT value among the former type may be {\it strictly} greater than that among the latter type. The authors attribute this to the lack of quasi-convexity of CPT preference (in contrast to the classical expected utility theory preference). This property was also exploited by \cite{HendersonHobsonTse14} to complement and counter the findings in \citet{EbertStrack2012:UntilTheBitterEnd}.

So, what are the  {\it other} reasons why  a gambler may want to randomize, beyond and independent of  replacing path-dependence and maximizing CPT preference?
Indeed, preference for randomization is observed in daily life and in different cultures, such as last-minute deals by flight booking apps, ``sushi omakase" (you entrust yourself to a sushi chef to choose the ingredients and presentations of your sushi plate), and ``fukubukuro" (grab bags filled with unknown and random contents). The practice of ``drawing divination sticks", popular in Chinese
culture even today, is a vivid example of seeking randomization in addition to religious reasons. When people are reluctant or unable to make their own decisions
on important matters (usually marriage, school, or even home move), they
go
to a temple, pray and draw divination sticks, and follow whatever words on the sticks tell them to do.

There are also rich literatures in experimental psychology and economics that document extensive experiments about individuals {\it deliberately} randomizing  when making decisions. \cite{AgranovOrtoleva2017:StochasticChoice}
report on experiments in which subjects
who face identical questions repeated  three times in a row often switch
between their answers, and a significant portion of them are even  willing to pay for a coin flip to choose answers for them. \cite{DwengerEtal2013:FlippingACoin}
study a clearing house data for university admissions
in Germany, where applicants
submit multiple rankings of the universities they wish to attend. The authors  find that a significant fraction of students report
contradictory  rankings without any rational reasons.

%This way they delegate the decision to a god, a benefit of which is to release themselves from making the decision on their own and hence relieve themselves from regret if the outcome is bad.

 % which {\it happens} to be favored by his risk preference? There are reasons beyond and independent of  replacing path-dependence and maximizing CPT preference.
 The  psychological literature has put forward  various theories to explain the preference for randomization, such as
responsibility aversion \citep{LKP2011}, decision avoidance \citep{Anderson2003}, and regret theory \citep{ZR2007}. In the aforementioned example of divination sticks, prayers delegate their decisions to a god, a benefit of which is to release themselves from making the decisions on their own and hence relieve themselves from regret should the choice turn out to be  bad.
In the economics literature,  \cite{DSW2004} use the general ``utility of gambling" to explain
randomization.\footnote{The theory of utility of gambling can also be used to explain why a gambler is willing to play a bet that has unfavorable average return; see   \citet[p. 38]{Barberis2012:Casino}. Here, the utility of gambling is applied to a different phenomenon, namely, the desire to randomize.} This type of utilities either violate some basic axioms underlining classical utility theory such as second-order stochastic dominance and betweenness \citep{CH1994,Blavatskyy2006:ViolationBetweenness} (precisely the same reason why randomization strictly improves CPT value), or prefer irrational diversification or hedging \citep{B2002}, or weight more on fairness than on outcomes \citep{KKT1986,BBO2005}.

In summary, a gambler may have various independent reasons to perform randomization while gambling, which may be only partly relevant, or completely irrelevant, to his CPT preference. That is why we introduce an {\it independent} binary random variable to capture such a desire for randomization. Finally, the easy availability for flipping a biased coin nowadays also makes the inclusion of randomization more plausible. For complex gambles such as stock trading, one can easily simulate the outcomes of any randomization in a computer. For literal casino gamble, the gambler can bring in a smartphone where apps are available to simulate coin flips with any user-defined probabilities.

\section{Characterization of Stopped State Distributions}\label{se:RandomizedRootstoppingtime}

As explained earlier, the main thrust of our approach to solving Problem \eqref{originaloptimization} is to change its decision variable from
the stopping time $\tau$ to the {\it distribution} of the stopped state $S_\tau$.
A key step is therefore to characterize the admissible set of these distributions.
Moreover, once an optimal distribution is obtained there needs to be a way to
recover the stopping time that generates this distribution. These two questions are intertwined and will actually be solved together. This section addresses
them.

%
%In this section we construct randomized Root stopping times. We study their optimality properties and use them to obtain necessary and sufficient conditions describing the distributions that can be embedded in the random walk $S$ in the finite time horizon $0\leq t\leq T$.

\iffalse

The Root stopping time, developed in \citet{Root1969} to solve the classical Skorokhod embedding problem for a Brownian motion $B$  in a finite time horizon, is the first hitting time on an explicitly constructed  barrier region in the time-space $\mathbb{R}_+ \times \mathbb{R}$. For any centred $\mu$ on $\mathbb{R}$ with a finite second moment $v$, such a time $\tau_R$ exists which embeds $\mu$ in $B$, i.e., $B_\tau \sim \mu$ and $\mathbb{E}[\tau] = v$. Further, it minimizes $\bE[(\tau-t)^+]$, for any $t\geq0$, among all $\tau$ which embed $\mu$; see \citet{Root1969,Rost1976} and \citet{Obloj2004:SkorokhodEmbedding}. It follows that if a measure $\mu$ can be embedded in Brownian motion on a finite time horizon $[0,T]$, then it can  be done using simply the Root embedding. This was further generalized to the case of partial marginal information in \citet{Mykland19} and \citet{SpoidaPhD}. We develop here the analogous ideas in discrete time for the SSRW.

\fi

%\subsection{Potential functions and first hitting times}

Denote by $\cP(\bR)$ the set of probability measures $\mu$ on $\bR$  and by $\cP_0(\bR)$ the subset of $\cP(\bR)$ whose elements have finite first moments and are centered: $\int |x| \mu(dx)<\infty$ and  $\int x \mu(dx)=0$. Denote by $\cP_0(\bZ)=\{\mu \in \cP_0(\bR): \mu(\bZ)=1\}$ the subset of $\cP_0(\bR)$ supported on integers.

For $\mu \in \cP_0(\bR)$,  define a function
\begin{align*}
U_\mu(x) := \int_{\bR} |x-y|\mu(dy),\;\; x\in \bR,
\end{align*}
which is called the \emph{potential of }$\mu$.\footnote{Note that our definition here is the negative of the usual definition of  potential.} For $\mu\in \cP_0(\bZ)$, $U_\mu$ is a linear interpolation of the points $\{U_\mu(k): k\in \bZ\}$. The following are evident:
\begin{equation}\label{eq:potential_rep}
 \mu(\{x\})=\frac{U_\mu(x+1)+U_\mu(x-1)}{2}-U_{\mu}(x),\quad U_\mu(x)=-2x\mu([x,\infty))+x+2\sum_{y\geq x} y\mu(\{y\}).
\end{equation}
%We let $\lfloor x \rfloor$ denote the largest integer no greater than $x$.
Potential function uniquely determine probability measure, namely,  two measures are identical if and only if their potential functions are identical; see \cite{Obloj2004:SkorokhodEmbedding}.
Finally, for any stopping time $\tau$, with a slight abuse of notation we simply write $U_{S_\tau}$ for the potential of the distribution of $S_\tau$, when well defined.

We can use a sequence of piecewise linear functions, called {\it evolutional functions},  to approach a potential function. Indeed, given $\mu\in \cP_0(\bZ)$, we define recursively the following sequence of functions:
\begin{equation}\label{eq:differenceequation}
\begin{split}
U_0^\mu (x) &:= |x|,\\
U_t^\mu(x) &:= \frac{U_{t-1}^\mu(x-1) + U_{t-1}^\mu(x+1)}{2} \wedge U_{\mu}(x), \quad t=1,2,...,\;\;x \in \mathbb{Z}.
\end{split}
\end{equation}
We then extend each $U_t^\mu$ to non-integers $x\in \bR$ by linear interpolation. When $\mu$ is fixed, we may drop the superscript $\mu$ and just write $U_t$ for simplicity.
Figure \ref{potentialfunction} illustrates how $U_t$ evolves to $U_\mu$ for an example of $\mu$.
% given by: $\mu(\{5\}) = \mu(\{-5\}) = 1/32$, $\mu(\{3\}) = \mu(\{-3\}) = 5/32$, $\mu(\{1\}) = \mu(\{-1\}) = 5/16$ and $\mu(\{x\}) = 0$ otherwise.

\begin{figure}
  \centering
  \includegraphics[width=150mm]{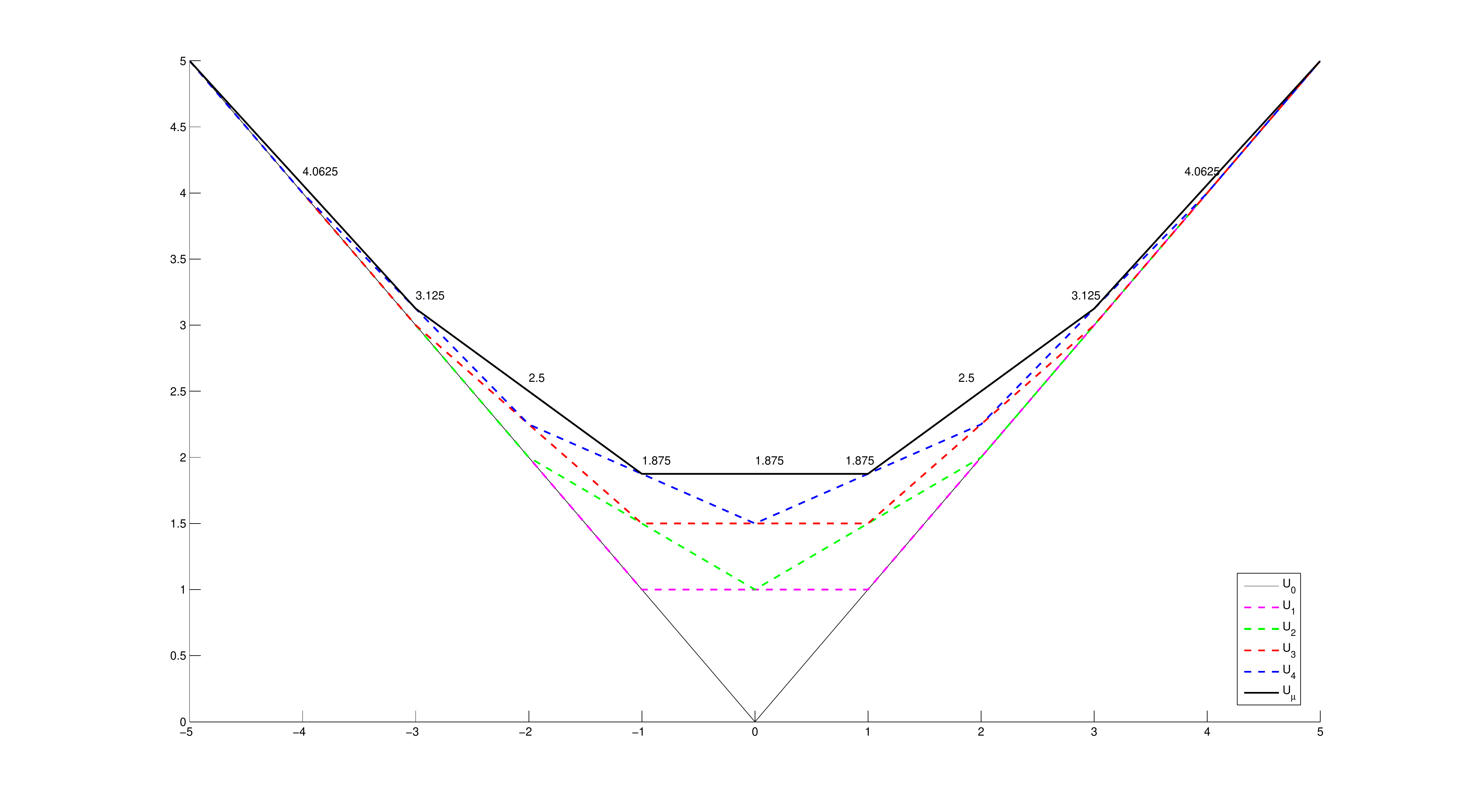}\\
  \caption{\small An illustration of how $U_t$ evolves to $U_\mu$, for $t=0,1,2,3,4$. Here $\mu$ has the distribution $\mu(\{5\}) = \mu(\{-5\}) = 1/32$, $\mu(\{3\}) = \mu(\{-3\}) = 5/32$, $\mu(\{1\}) = \mu(\{-1\}) = 5/16$; otherwise, $\mu(\{x\}) = 0$. }\label{potentialfunction}
  %$U_0$ is the lowest solid black line, $U_1$ the dashed  pink line, $U_2$ the dashed green line, $U_3$ the dashed red line, $U_4$ the dashed blue line, and $U_5 = U_\mu$ the solid bold black line.}\label{potentialfunction}
\end{figure}

%\subsection{Randomised Root stopping times and their properties}

The optimal stopping time we will derive belongs to a special class  of randomized, Markovian
stopping times called the {\it Root stopping times}. The
original version of the  Root stopping times was developed in \citet{Root1969} to solve the classical Skorokhod embedding problem for a Brownian motion $B$ on an infinite time horizon.\footnote{Precisely, an (original)  Root stopping time is the first hitting time of $B$ on an explicitly constructed  region with a barrier in the time-space $\mathbb{R}_+ \times \mathbb{R}$. For any centred $\mu$ on $\mathbb{R}$ with a finite second moment $v$, such a time $\tau_R$ exists which embeds $\mu$ in $B$, i.e., $B_\tau \sim \mu$ and $\mathbb{E}[\tau] = v$; %Further, it minimizes $\bE[(\tau-t)^+]$, for any $t\geq0$, among all $\tau$ which embed $\mu$;
see \citet{Root1969,Rost1976} and \citet{Obloj2004:SkorokhodEmbedding}.}
%Moreover, It follows that if a measure $\mu$ can be embedded in Brownian motion on a finite time horizon $[0,T]$, then it can  be done using simply the Root embedding. This was further generalized to the case of partial marginal information in \citet{Mykland19} and \citet{SpoidaPhD}.
We now develop  the analogous ideas in discrete time for the SSRW.

Consider an integer-valued vector ${\bf b} := (...,b({-1}),b(0),b(1),...)$ where, for any $x\in \bZ$, $b(x)= x + 2k\geq |x|$ with some $k \in \mathbb{Z}$, and  another vector ${\bf r} := (...,r({-1}),r(0),r(1),...)$, where $r(x) \in [0,1]$. %We write interchangeably  $b(x)=b_x$ and $r(x) = r_x$.
Given ${\bf b}$ and ${\bf r}$, define the probability distributions of a family of  Bernoulli random variables $\{\xi_{t,x}: t\in \bN,x\in \bZ\}$ as follows: %For each $x \in \mathbb{Z}$,
\begin{align*}
\begin{cases}
\mathbb{P}(\xi_{t,x} = 0) = 1- \prob(\xi_{t,x} = 1) = 0, & t < b(x), \\
\mathbb{P}(\xi_{t,x} = 0) = 1- \prob(\xi_{t,x} = 1) = r(x), & t = b(x), \\
\mathbb{P}(\xi_{t,x} = 0) = 1- \prob(\xi_{t,x} = 1) = 1, & t > b(x).
\end{cases}
\end{align*}
%The randomized Root stopping time is defined as
%\begin{equation}\label{eq:Rootdef}
%\tau_R({\bf b}, {\bf r}) := \inf\{t \geq 0 : S_t = x, t \geq b(x), \xi_{t,x} = 0 \} \;.
%\end{equation}
Graphically,  ${\bf b}$ is a barrier that defines a time-space stopping region
$$ \mathcal{R}_b := \big\{(t,x) : t\in \bN,\;x\in \bZ,\;t \geq b(x) \big\}, $$
 and the components of ${\bf r}$ are the probabilities to stop exactly on the boundary of this stopping region. The randomized Root stopping time is defined as
 \begin{equation}\label{eq:Rootdef}
  \tau_R({\bf b}, {\bf r}): = \inf\left\{t \in \bN: (t,S_t) \in \mathcal{R}_b \text{ and } \xi_{t,S_t} = 0 \right\}.
  \end{equation}
 This stopping time is Markovian, because it depends only on the current state of the random walk $S$. It is randomized because it depends on the outcome of the Bernoulli random variables $\xi_{t,x}$'s.

 Figure \ref{RandomizedRootstoppingtime} illustrates such a stopping time.
%, where the boundary $\bf b$ is given as follows: $b(4) = 4$, $b(3) = 3$, $b(2) = 4$, $b(1) = 3$, $b(0) = 2$, $b(-1) = 3$, $b(-2) = 2$, $b(-3) = 3$, $b(-4) = 4$.
The grey boundary divides the area into two subareas: the one on the left hand side has white nodes representing ``continue", and that  on the right hand side consists of black nodes indicating ``stop".
 Stopping at a grey node $(t,x)$  is randomized with $r(x)$ being the probability of stopping.

\begin{figure}
  \centering
  \includegraphics[width=100mm]{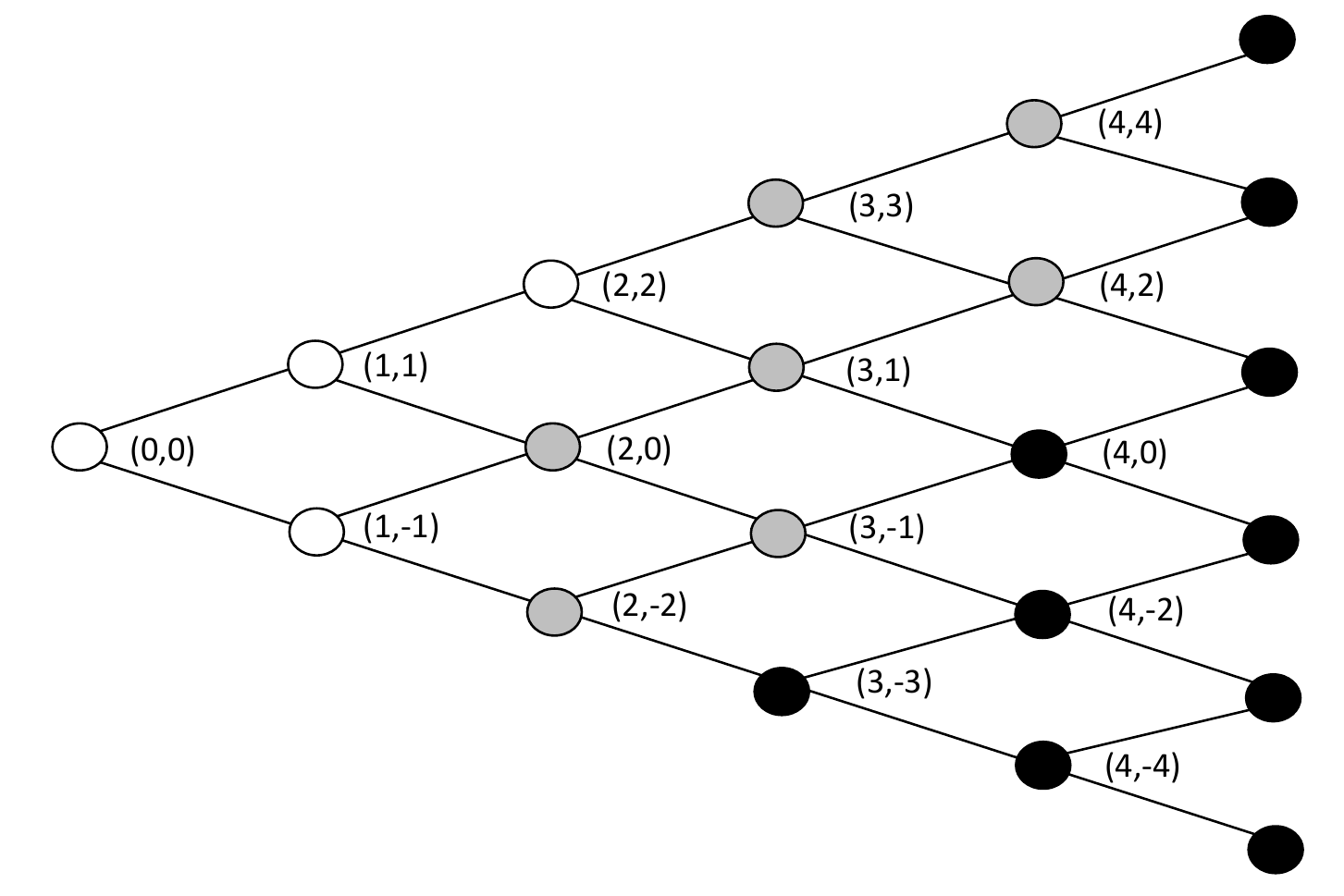}\\
  \caption{\small An example of the  Root stopping time with $T = 5$. Black nodes mean ``stop'', white nodes mean ``continue'', and grey nodes mean ``randomize''. The boundary $\bf b$ is given as follows: $b(4) = 4$, $b(3) = 3$, $b(2) = 4$, $b(1) = 3$, $b(0) = 2$, $b(-1) = 3$, $b(-2) = 2$, $b(-3) = 3$, $b(-4) = 4$.}\label{RandomizedRootstoppingtime}
  %The grey boundary divides the area into two subareas: the nodes on the left hand side are white, which is non-stop, and those on the right hand side are black, which is the stopping region.}\label{RandomizedRootstoppingtime}
  \end{figure}

The following theorem is one of the main results of the paper that provides a theoretical foundation for the numerical algorithm we are going to present to solve our  casino gambling model.
It characterizes the admissible set of stopped distributions under stopping times in ${\cal T}_T$, and reveals that the set is the same as that of stopped distributions using only randomized Root stopping times. As a consequence, any admissible stopping strategy is always dominated by a randomized Root stopping time.

%Theorem \ref{thm:rootstoppingtime} presents explicitly the necessary and sufficient conditions on $\mu$ for existence of a randomized Root stopping, and hence also {\it any} stopping time, which embeds $\mu$ by time $T$. As a result, any optimal value can be achieved by a randomized Root stopping. Moreover, the original casino gambling problem of optimal stopping can be reformulated as a finite-dimension program, by changing decision variables from stopping times to probability distributions of accumulated gains or losses at exit time.

\begin{theorem}\label{coro:Rootst}
  Let $T\geq 1$, $\mu \in \cP_0(\bZ)$ such that $\mu([-T,T]) = 1$. Then there exists a stopping time $\tau \in {\mathcal{T}}_T$ such that $S_\tau\sim \mu$ if and only if
%there exists a randomized Root stopping time $\tau_R({\bf b}, {\bf r}) $ such that $S_{\tau_R({\bf b}, {\bf r}) }\sim \mu$ if and only if
%\begin{enumerate}
%\item[(1)] the support of $\mu$ is contained by $[-T,T]$, i.e., $\mu([-T,T)] = 1$,
%\item[(2)] $\mu$ is a centered probability measure, i.e., $\sum_{x = -T}^T x \mu(\{x\}) = 0$,
\begin{equation}\label{3}
U_{\mu}(x) \leq \frac{U_{T-1}^\mu(x+1)+U_{T-1}^\mu(x-1)}{2},\;\; x = -(T-2), -(T-4), ..., T-4, T-2.
\end{equation}
Moreover, in this case there exists a randomized Root stopping time $\tau_R({\bf b}, {\bf r})\in {\mathcal{T}}_T$ such that $S_{\tau_R({\bf b}, {\bf r})} \sim \mu$.
\end{theorem}

\section{A Mathematical Program}\label{se:finiteprogram}

Theorem \ref{coro:Rootst}
hints that we can, instead of endeavoring to find the stopping time $\tau$
in Problem \eqref{originaloptimization}, try to find the probability distribution $\mu$ of the stopped state $S_\tau$. Namely we change decision variable from $\tau$ to $\mu$ for Problem \eqref{originaloptimization}. The resulting problem is a (nonlinear) {\it mathematical  program} (i.e., a constrained optimization problem) with the condition (\ref{3}) translating into
certain constraints.

 Moreover, once we solve this problem and find the optimal distribution $\mu$, then it follows from Theorem \ref{coro:Rootst} that there exists a randomized Root stopping time $\tau_R({\bf b}, {\bf r}) $ that achieves the same
stopped distribution and, hence, solves \eqref{originaloptimization}. Furthermore, based on the proof of Theorem \ref{coro:Rootst} (see Appendix A), we can devise an algorithm to find $({\bf b}, {\bf r})$ and, consequently, $\tau_R({\bf b}, {\bf r}) $.
 We now formulate  the mathematical program and provide its solution algorithm.

\subsection{The mathematical program formulation and solution}

Given $\tau\in {\cal T}_T$, let $\mu \sim S_\tau$. Define
%Given a centered discrete probability measure $\mu$ on $\mathbb{Z}$ with $\mu([-T,T]) = 1$, let
two $T$-dimensional vector variables, $\vecx := (x_1,x_2,...,x_T)$ and $\vecy := (y_1,y_2,...,y_T)$, where
$x_n = \mu([n,T])$, $y_n = \mu([-T,-n])$, $n=1,2,...,T$. Clearly, $\vecx$ and $\vecy$ are gambler's  decumulative gain distribution and cumulative loss distribution, respectively.
Then the original objective function \eqref{prob:InfiniteDimProgramming:original} is equivalent to, as a function of $(\vecx,\vecy)$,
\begin{align}\label{eq:mainvalueinxy}
\mathbb{U}(\vecx,\vecy):=& \sum_{n=1}^T \left[u_+(n)-u_+(n-1)\right]w_+(x_n) - \lambda\sum_{n=1}^T \left[u_-(n)-u_-(n-1)\right]w_-(y_n).
\end{align}
 Naturally,  we must have $1 \geq {x_1} \geq {x_2} \geq {...} \geq {x_T} \geq 0$, $1 \geq {y_1} \geq {y_2} \geq {...} \geq {y_T} \geq0$, $x_1 + y_1 \leq 1$. On the other hand, $\mu$ has zero expectation due to optional sampling theorem; so
\begin{align*}
0 & = \sum_{n = -T}^T n \mu(\{n\}) = \sum_{n = 1}^T n \mu(\{n\}) - \sum_{n = 1}^{T} n \mu(\{-n\}) = \sum_{n = 1}^{T} \mu([n,T]) - \sum_{n = 1}^{T} \mu([-T,-n]) \\
& = \sum_{n=1}^T {x_n} -\sum_{n=1}^T {y_n}.
\end{align*}
In summary, the following constraints are required for the probability distribution
of $S_\tau$ where $\tau\in{\cal T}_T$:
\begin{align}\label{eq:generalinequalities}
  \begin{cases}
1 \geq {x_1} \geq {x_2} \geq {...} \geq {x_T} \geq0,\\
1 \geq {y_1} \geq {y_2} \geq {...} \geq {y_T} \geq0,\\
x_1 + y_1 \leq 1,\\
\sum_{n=1}^T {x_n}=\sum_{n=1}^T {y_n} .
  \end{cases}
\end{align}

Moreover, Theorem \ref{coro:Rootst} necessitates % the following constraints:
%\begin{align}\label{eq:rootstinequalities}
%U_{\mu}(n) \leq \frac{U_{T-1}^\mu(n+1)+U_{T-1}^\mu(n-1)}{2},  \;\; n = -(T-2), -(T-4), ..., T-4, T-2.
%\end{align}
condition   (\ref{3}), which constitutes a family of inequalities on $\mu$'s potential function and the corresponding evolutional functions, which will later be translated into constraints on $S_\tau$'s distribution functions.
Here, let us illustrate (\ref{3}) for each of   $T=1,2,\ldots, 5$. To ease notation, we will suppress the superscript $\mu$ on the evolutional functions. For $T = 1$, the condition is
satisfied automatically for any $\mu \in \cP_0(\bZ)$ with $\mu([-T,T]) = 1$.
For $T = 2$, (\ref{3}) amounts to
$U_{\mu}(0) \leq \frac{U_1(1) + U_1(-1)}{2} = 1$.
For $T = 3$, (\ref{3}) reduces to
\begin{align*}
\max\left\{U_{\mu}(1),U_\mu(-1)\right\} \leq \frac{2 + \min\{U_{\mu}(0),1\}}{2}.
\end{align*}
For $T = 4$, (\ref{3}) is equivalent to
\begin{align*}
\begin{cases}
U_{\mu}(2) \leq \frac{3 + U_3(1)}{2}, & U_3(1) = \min\left\{U_{\mu}(1), \frac{2 + \min(U_{\mu}(0),1)}{2}\right\},\\
U_{\mu}(0) \leq \frac{U_3(1) + U_3(-1)}{2},  & U_3(-1) = \min\left\{U_{\mu}(-1), \frac{2 + \min(U_{\mu}(0),1)}{2}\right\},\\
U_{\mu}(-2) \leq \frac{3 + U_3(-1)}{2}.
\end{cases}
\end{align*}
For $T = 5$, (\ref{3}) specializes to
\begin{align*}
\begin{cases}
U_{\mu}(3) \leq \frac{4 + U_4(2)}{2}, & U_4(2) = \min\left\{U_{\mu}(2),\frac{3 + \min\left\{U_{\mu}(1), \frac{2 + \min\{U_{\mu}(0),1\}}{2}\right\}}{2}\right\},\\
U_{\mu}(1) \leq \frac{U_4(2) + U_4(0)}{2}, & U_4(0) = \min\left\{U_{\mu}(0), \frac{\min\left(U_{\mu}(1), \frac{2 + \min(U_{\mu}(0),1)}{2} \right) +\min\left(U_{\mu}(-1), \frac{2 + \min(U_{\mu}(0),1)}{2}\right)}{2}\right\}, \\
U_{\mu}(-1) \leq \frac{U_4(-2) + U_4(0)}{2}, & U_4(-2) = \min\left\{U_{\mu}(-2), \frac{3 + \min\left(U_{\mu}(-1), \frac{2 + \min(U_{\mu}(0),1)}{2} \right)}{2}\right\}, \\
U_{\mu}(-3) \leq \frac{4 + U_4(-2)}{2}.
\end{cases}
\end{align*}

The following lemma, which follows a direct, if somewhat lengthy, computation, expresses $U_\mu(n)$ and, consequently, the constraints (\ref{3}), in terms of $\vecx$ and $\vecy$.
\begin{lemma}
%Let $\mu\in \cP_0(\bZ)$ with $\mu([-T,T]) = 1$ and denote $x_n = \mu([n,T])$ and $y_n = \mu([-T,-n])$. Then, for
For $n \in \mathbb{Z} \cap [-T,T]$,
\begin{align*}
U_\mu(n) =
\begin{cases}
2\sum_{j = n+1}^{T} x_j + n, &  n \geq 0, \\
2 \sum_{j = |n|+1}^{T} y_j + |n|, &  n < 0 .
\end{cases}
\end{align*}
\end{lemma}
%
%\begin{proof}
%For $0 \leq n \leq T-1$,
%\begin{align*}
%U_\mu(n) & = \sum_{k \in \mathbb{Z}} |k-n| \mu(\{k\} ) \\
%& = \sum_{k = n+1}^{T} (k-n) \mu(\{k\}) + \sum_{k = -T}^{n-1} (n-k) \mu(\{k\})\\
% %& = \sum_{j = n+1}^{T} \sum_{k = j}^{T} \mu(\{k\}) + \sum_{j = -T}^{n-1} \sum_{k = -T}^{j} \mu(\{k\}) \\
% & = \sum_{j = n+1}^{T} \mu([j,T]) + \sum_{j = -T}^{n-1} \mu([-T,j]) \\
% & = \sum_{j = n+1}^{T} \mu([j,T]) + \sum_{j = 0}^{n-1} (1-\mu([j+1,T])) + \sum_{j = -T}^{-1} \mu([-T,j])\\
% & = \sum_{j = n+1}^{T} x_j + \sum_{j = 0}^{n-1} (1-x_{j+1}) + \sum_{j = 1}^{T} y_j \\
% & = \sum_{j = n+1}^{T} x_j + \sum_{j = 0}^{n-1} (1-x_{j+1}) + \sum_{j = 1}^{T} x_j \\
% & = \sum_{j = n+1}^{T} x_j + n - \sum_{j = 0}^{n-1} x_{j+1} + \sum_{j = 1}^{n} x_j + \sum_{j = n+1}^{T} x_j \\
% & = 2\sum_{j = n+1}^{T} x_j + n \;.
%\end{align*}
%Similarly, for $-(T-1) \leq n < 0$, $U_\mu(n) = 2 \sum_{j = |n|+1}^{T} y_j + |n|$.
%
%\end{proof}
%
To illustrate, take $T = 5$. Then
\begin{align*}
\begin{array}{lll}
&U_\mu(3) = 2 \sum_{n = 4}^{5} x_n + 3 \;,
&U_\mu(-3) = 2 \sum_{n = 4}^{5} y_n + 3\;, \\
&U_\mu(1) = 2 \sum_{n = 2}^{5} x_n + 1 \;,
&U_\mu(-1) = 2 \sum_{n = 2}^{5} y_n + 1 \;.
\end{array}
\end{align*}
%Moreover,
%\begin{align*}
%\begin{array}{cl}
%&U_4(2)  = \min\left(2 \sum_{n = 3}^{5} x_n + 2,\frac{3 + \min\left(2 \sum_{n = 2}^{5} x_n + 1, \frac{2 + \min(2 \sum_{n = 1}^{5} x_n,1)}{2}\right)}{2}\right) \;, \\
%&U_4(-2)  = \min\left(2 \sum_{n = 3}^{5} y_n + 2,\frac{3 + \min\left(2 \sum_{n = 2}^{5} y_n + 1, \frac{2 + \min(2 \sum_{n = 1}^{5} y_n ,1)}{2}\right)}{2}\right) \;, \\
%&U_4(0)  = \min\left(2 \sum_{n = 1}^{5} x_n, \frac{\min\left(2 \sum_{n = 2}^{5} x_n + 1, \frac{2 + \min(2 \sum_{n = 1}^{5} x_n,1)}{2} \right) +\min\left(2 \sum_{n = 2}^{5} y_n + 1, \frac{2 + \min(2 \sum_{n = 1}^{5} x_n,1)}{2}\right)}{2}\right) \;.
%\end{array}
%\end{align*}

%Combining above together, the original optimal stopping problem \eqref{originaloptimization} is turned into the finite-dimension program, maximizing the objective function in \eqref{eq:mainvalueinxy}, subject to constraints \eqref{eq:generalinequalities} and \eqref{eq:rootstinequalities}.

We are now ready to formulate the mathematical program that is equivalent to the original stopping problem \eqref{originaloptimization}.
Define
\begin{align*}
{\bf A} =
\begin{bmatrix}
1 & 0 & \dots & 0 \\
-1 & 1 & \ddots & \vdots \\
0 & -1 & \ddots & 0 \\
\vdots & \ddots & \ddots & 1 \\
0 & \dots & 0 & -1
\end{bmatrix}_{(T+1) \times T},
\quad
{\bf c} =
\begin{bmatrix}
1 \\
0 \\
0 \\
\vdots \\
0
\end{bmatrix}_{(T+1) \times 1},
\quad
{\bf 1} =
\begin{bmatrix}
1 \\
1 \\
\vdots \\
1
\end{bmatrix}_{T \times 1},
\quad
{\bf e_j} =
\begin{bmatrix}
\vdots \\
1 \\
0 \\
\vdots
\end{bmatrix}_{T \times 1},
\end{align*}
%Denote $T \times (2T+1)$ matrix by $[a_{m,n}]$,
along with a set of functions $f^m_n : \mathbb{R}^T \times \mathbb{R}^T \to \mathbb{R}$, $m = 1,...T$, $n = 1,...2T+1$, in the following way.  For $\vecx, \vecy \in \mathbb{R}^T$, let
$$f^m_1(\vecx,\vecy) = f^m_{2T+1}(\vecx,\vecy) \equiv T,\;\;m = 1,...T,$$ $$f^1_n(\vecx,\vecy) \equiv |n - (T+1)|,\;\;n = 2,...2T,$$
 and for $m = 2,3,...T$:
\begin{align*}
f^m_n(\vecx,\vecy) =
 \begin{cases}
  \min\left( \frac{f^{m-1}_{n-1}(\vecx,\vecy) + f^{m-1}_{n+1}(\vecx,\vecy)}{2}, 2 \sum_{j = T+2 - n}^T {\bf e_j}' \vecy + (T+1) - n \right), \; n = 2,3,...T,\\
 % \quad \text{for } n = 2,3,...T \\
  \min\left( \frac{f^{m-1}_{n-1}(\vecx,\vecy) + f^{m-1}_{n+1}(\vecx,\vecy)}{2}, 2 \sum_{j = n - T}^T {\bf e_j}' \vecx + n - (T+1) \right),\;n = T+1,...2T-1, 2T.
  %\quad \text{for } n = T+1,...2T-1, 2T
 \end{cases}
\end{align*}
Then, the mathematical  program is
\begin{equation}\label{finitedimensionprogram}
\begin{array}{cl}
  \underset{\vecx,\vecy}{\max} \quad & \mathbb{U}(\vecx,\vecy) \;, \\
 \text{subject to } & {\bf A} {\vecx} \le {\bf c}, \; {\bf A} {\vecy} \le {\bf c}, \; {\bf e_1}' {\vecx} + {\bf e_1}' {\vecy} \le 1, \; {\bf 1}' {\vecx} - {\bf 1}' {\vecy} = 0,\\
 & \frac{ f^T_{n-1}(\vecx,\vecy) + f^T_{n+1}(\vecx,\vecy)}{2} \ge 2 \sum_{j = T+2 - n}^T {\bf e_j}' \vecy + (T+1) - n \\
 & \quad\quad\quad \text{ for } n = 2k+1, k =1,2,..., n \le T, \\
 & \frac{ f^T_{n-1}(\vecx,\vecy) + f^T_{n+1}(\vecx,\vecy)}{2} \ge 2 \sum_{j = n - T}^T {\bf e_j}' \vecx + n - (T+1) \\
 & \quad\quad\quad \text{ for } n = 2T-2k+1, k = 1,2,..., n \ge T+1.
\end{array}
\end{equation}

%
%
%
%Moreover, we define a set of numbers $a_{m,n}$, $m = 1,2,...,T$, $n = 1,2,...,2T+1$ as follows: $a_{1,n} = |n - (T+1)|$ for $n = 1,...,2T+1$, and $a_{m,1} = a_{m,2T+1} = T$ for $m = 1,...T$.\footnote{{\bf You missed defining many numbers, e.g. $a_{2,2}$.}}
%Then the mathematical programm is
%\begin{equation}\label{finitedimensionprogram}
%\begin{array}{cl}
%  \underset{\vecx,\vecy}{\max} \quad & \mathbb{U}(\vecx,\vecy) \;, \\
% \text{subject to } & {\bf A} {\vecx} \le {\bf c}, \; {\bf A} {\vecy} \le {\bf c}, \; {\bf e_1}^T {\vecx} + {\bf e_1}^T {\vecy} \le 1, \; {\bf 1}^T {\vecx} - {\bf 1}^T {\vecy} = 0,\\
% & a_{m,n} =
% \begin{cases}
%  \min\left( \frac{a_{m-1,n-1} + a_{m-1,n+1}}{2}, 2 \sum_{j = T+2 - n}^T {\bf e_j}^T \vecy + (T+1) - n \right) \\
%  \quad \text{for } n = 2,3,...T \\
%  \min\left( \frac{a_{m-1,n-1} + a_{m-1,n+1}}{2}, 2 \sum_{j = n - T}^T {\bf e_j}^T \vecx + n - (T+1) \right) \\
%  \quad \text{for } n =2T,2T-1,... T+1\footnote{{\bf This equality requires
%  $a_{m,n}\leq \frac{a_{m-1,n-1} + a_{m-1,n+1}}{2}$, which has not been stated when
%  $a_{m,n}$ was defined.}}
% \end{cases} \\
% & \quad \text{ for } m = 2,3,...T, \\
% & \frac{ a_{T,n-1} + a_{T,n+1}}{2} \ge
% \begin{cases}
% 2 \sum_{j = T+2 - n}^T {\bf e_j}^T \vecy + (T+1) - n \\
% \quad \text{ for } n = 3,5,..., n \le T \\
% 2 \sum_{j = n - T}^T {\bf e_j}^T \vecx + n - (T+1) \\
% \quad \text{ for } n = 2T-1,2T-3,..., n \ge T+1.
% \end{cases}
%\end{array}
%\end{equation}

The number of decision variables ($\vecx$ and $\vecy$) and the number of constraints in (\ref{finitedimensionprogram}) are both {\it linear} in $T$; hence the complexity of the problem is manageable. Moreover, there are standard solvers to solve this type of mathematical program.\footnote{In the following numerical experiments, we employ nonlinear optimization solver `fmincon' from {\bf MATLAB} Optimization Toolbox, on a desktop with Intel Core i5-4590/CPU 3.30GHz/RAM 8.00GB. For the \citet{Barberis2012:Casino}'s parameters $\alpha_+ = \alpha_- = 0.95$, $\delta_+ = \delta_- = 0.5$, $\lambda = 1.5$ with $T=5,6,7,8$, MATLAB uses 205 seconds, 220 seconds, 280 seconds, 350 seconds respectively. (Compare with those of the brute force reported in Footnote 3.)
 The running times for $T=10, 20,30,40,50$ are 9.35 minutes, 29 minutes, 89 minutes, 3.5 hours, 7.17 hours, respectively.}

The running times for $T = 5, 6, 7 $ were 39 seconds, 771 seconds and 27 hours, respectively. We were unable to obtain the solution for $T=8$ due to out of storage, with the running time estimated to be 300 days.

Once we solve this problem to get optimal  $(\vecx^*,\vecy^*)$,  we then run the    following algorithm to find the optimal randomized Root stopping time:

{\noindent \bf Step 1 }
Given $\mu^* \equiv (\vecx^*,\vecy^*)$, compute the corresponding potential function:
$ U_{\mu^*}(n) = 2\sum_{j = n+1}^{T} x_j^* + n$ for $n \geq 0$, and $U_{\mu^*}(n) = 2 \sum_{j = |n|+1}^{T} y_j^* + |n|$ for $n < 0$.
Then, compute its evolutional functions $U_t^{\mu^*}$ by \eqref{eq:differenceequation}, $t = 0,1,...T$.

{\noindent \bf Step 2 }
Compute the boundary ${\bf b}$ that separates the ``continue'' region from the ``stop'' region:
$b(n) = \inf\{t \geq |n|, t \in \mathbb{Z}: U_{t+1}^{\mu^*}(n) = U_{\mu^*}(n)\}$,
$n \in [-T,T] \cap \mathbb{Z}$.
(The constraints in \eqref{finitedimensionprogram} guarantees that the set involved is non-empty and $b(n) \le T$ $\forall n \in [-T,T] \cap \mathbb{Z}$.)

{\noindent \bf Step 3 }
Compute the probability ${\bf r}$ to stop at the boundary: $$r(n) = \frac{U_{b(n)}^{\mu^*} (n-1) + U_{b(n)}^{\mu^*}(n+1) - 2 U_{\mu^*}(n)}{ U_{b(n)}^{\mu^*}(n-1) + U_{b(n)}^{\mu^*}(n+1) - 2 U_{b(n)}^{\mu^*}(n)},\;\;n \in [-T,T] \cap \mathbb{Z}.$$

{\noindent \bf Step 4 }
Construct $\tau_R({\bf b}, {\bf r})$ according to  \eqref{eq:Rootdef}. %, which is the optimal solution to \eqref{originaloptimization}.

\subsection{A numerical example}\label{numerical}

We present an example to illustrate the solution procedure, using the same
parameters as in \citet{Barberis2012:Casino} with  $T=5$, $\alpha_+ = \alpha_- = 0.95$, $\delta_+ = \delta_- = 0.5$, $\lambda = 1.5$.\footnote{More examples with much longer time horizons will be presented in the next section.} Solving the corresponding
mathematical program  for the optimal distribution $\mu^*$ yields
\begin{align*}
& x_1^* = 0.1875, \; x_2^* = 0.1273, \; x_3^* = 0.1227, \; x_4^* = 0.03152, \;  x_5^* = 0.03098, \\
& \quad y_1^* = 0.5, \; y_2^* = 0, \; y_3^* = 0, \; y_4^* = 0, \; y_5^* = 0.
\end{align*}
The corresponding potential function $U_{\mu^*}$ is
\begin{align*}
& U_{\mu^*}(0) = 1, \; U_{\mu^*}(1) = 1.625, \; U_{\mu^*}(2) = 2.3704, \; U_{\mu^*}(3) = 3.125, \; U_{\mu^*}(4) = 4.06196, \\
& \quad U_{\mu^*}(n) = |n| \text{ for } n \geq 5 \text{ and } n \leq -1.
\end{align*}
Figure \ref{potentialfunctionexample} illustrates how $U_{\mu^*}$ is achieved by the evolutional functions within five steps.
\begin{figure}
  \centering
  \includegraphics[width=150mm]{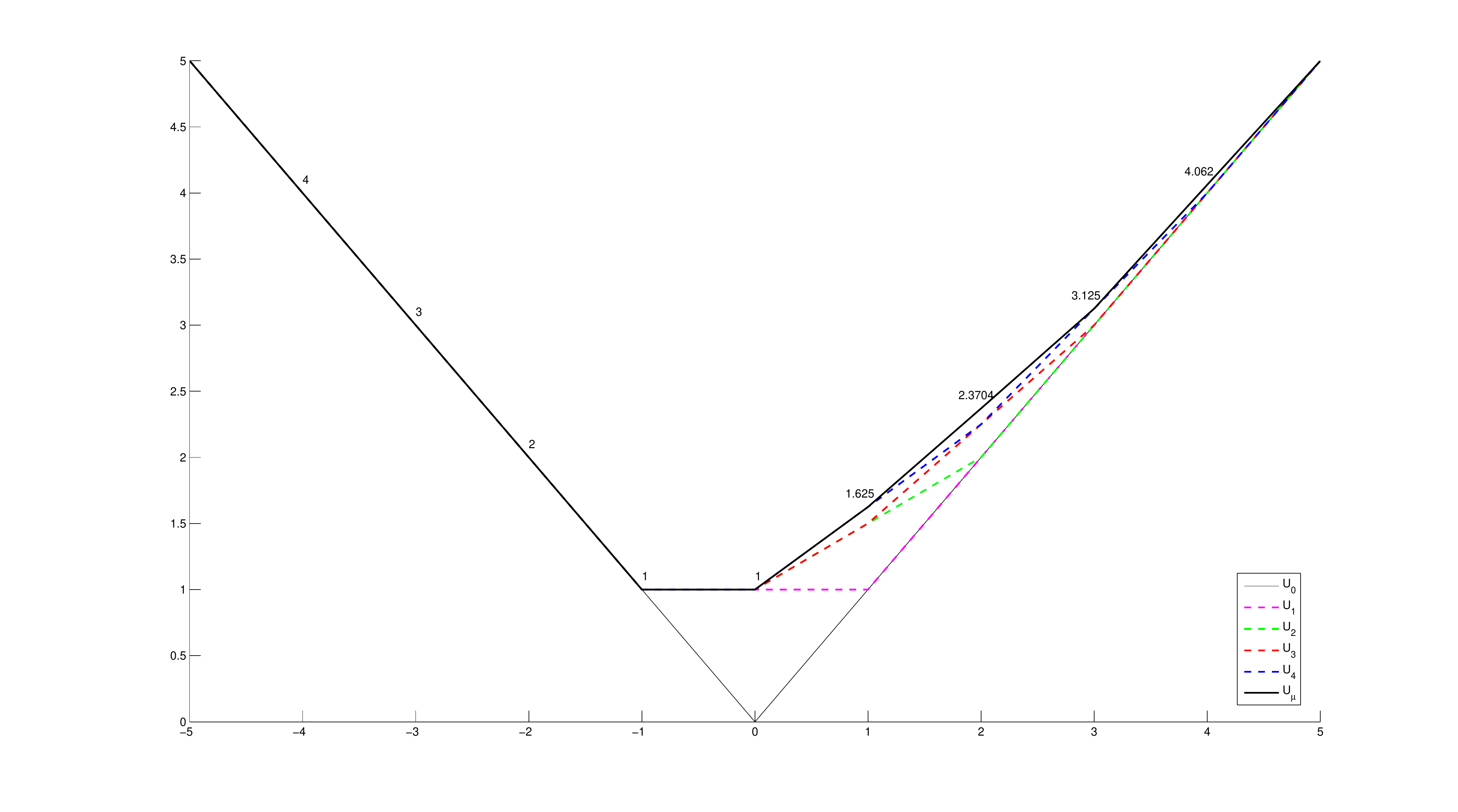}\\
  \caption{\small $U_{\mu^*}$ is achieved within five steps, where the optimal distribution $\mu^*$ is the solution to (\ref{finitedimensionprogram}) for $T = 5$, $\alpha_+ = \alpha_- = 0.95$, $\delta_+ = \delta_- = 0.5$, $\lambda = 1.5$.}\label{potentialfunctionexample}
\end{figure}

We then apply the algorithm previously presented to recover the optimal randomized Root stopping time $\tau^*$ from the optimal distribution $\mu^*$, with $S_{\tau^*} \sim \mu^*$.
The strategy, which is optimal at $t=0$ (only) and implemented by the precommitted gambler, is drawn in the left panel of Figure \ref{Optimalstrategy}. Note that black nodes mean ``stop'', white ones mean ``continue'', and grey ones mean ``randomization''. The number above a grey node is the probability to stop.

\begin{figure}
  \centering
  \begin{minipage}[t]{0.49\textwidth}
     \includegraphics[width=\textwidth]{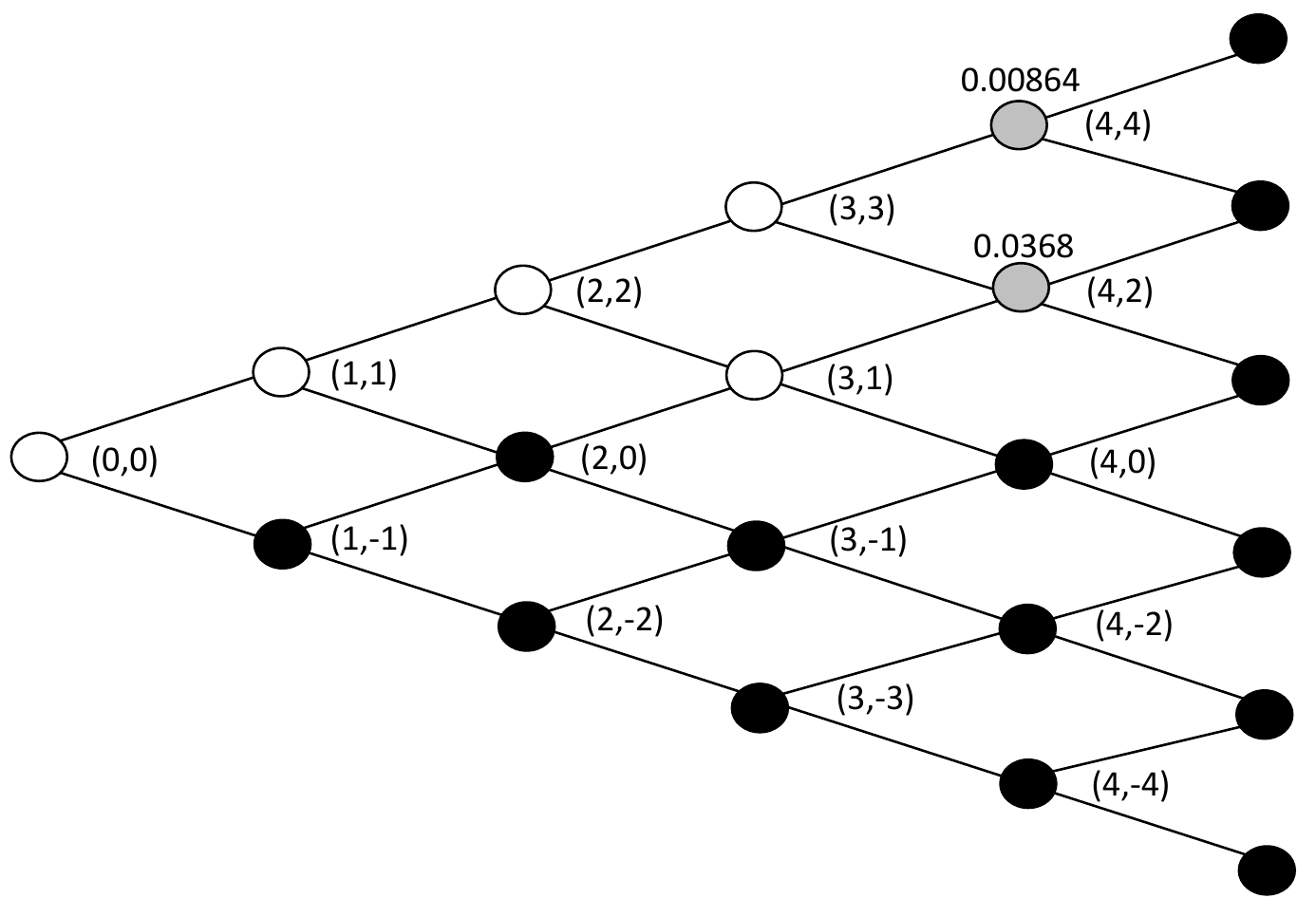}
  \end{minipage}
  \begin{minipage}[t]{0.49\textwidth}
     \includegraphics[width=\textwidth]{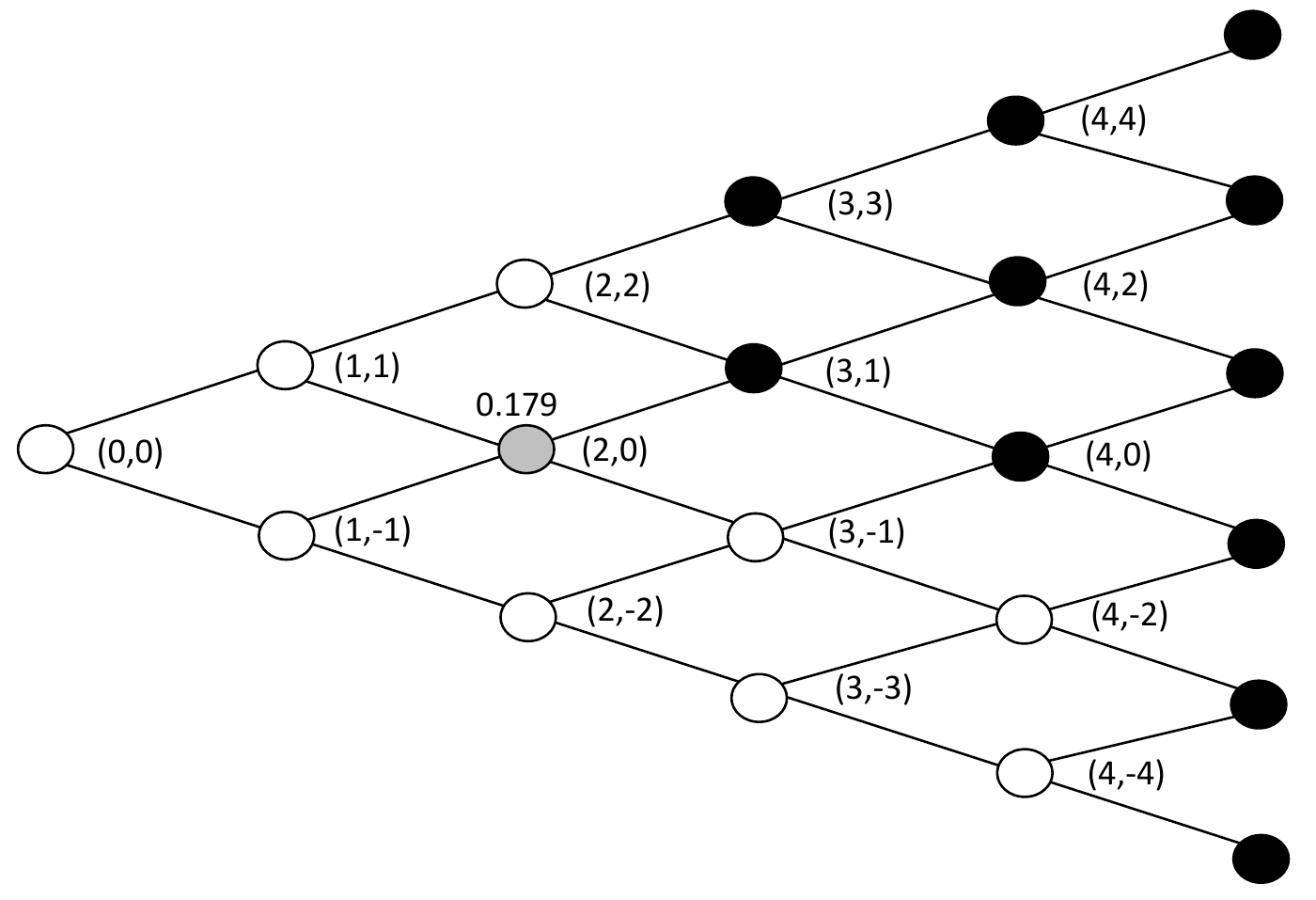}
  \end{minipage}
  \caption{\small The left panel shows the precommitter's strategy and the right panel shows the naivet\'e's strategy, for $T = 5$, $\alpha_+ = \alpha_- = 0.95$, $\delta_+ = \delta_- = 0.5$, $\lambda = 1.5$. Black nodes mean ``stop'', white nodes mean ``continue'', and grey nodes mean ``randomize''. The numbers above the grey node stand for the probability to stop. While the precommitter is mainly to continue in gains and stop in losses, the naivet\'e's behavior  is almost completely reversed.}     \label{Optimalstrategy}
\end{figure}

The main feature of this precommitted optimal strategy is to continue in the gain domain and to stop in the loss domain until $T=5$, except at time 4 where there are positive probabilities to stop in gains.
In particular, randomization takes place at nodes $(4,4)$ and  $(4,2)$, with the (very small) probabilities to stop equal to 0.00864 and 0.0368, respectively. The CPT value of this  randomized strategy is 0.3369592.
Compared with \citet{Barberis2012:Casino} where the CPT value is 0.3369398, the optimal non-randomized, Markovian  strategy has white nodes at $(4,4)$ and $(4,2)$, instead of grey nodes that involve randomization.\footnote{These are the only two nodes that are different between \citet{Barberis2012:Casino} and the present paper.
Note they occur at  $T-1$ and when there are sufficient gains. The intuition why
the gambler randomizes at these two nodes will be explained in Subsection \ref{omr} when we investigation the situation when the horizon is extended from $T$ to $T+1$.}    In summary,
allowing  path-dependent and randomized strategies does indeed improve optimal CPT values (albeit only slightly in this {\it particular} instance) over non-randomized, Markovian ones.\footnote{This improvement can be {\it significant} with other parameter specifications. For example, with $T=2$ and $(\alpha_\pm, \delta_\pm, \lambda) = (0.9, 0.5, 1.25)$,  the optimal CPT value among non-randomized, path-independent strategies  is 0.058069135, and that among randomized ones is 0.065696808, representing a 13\% increase. When $T=2$ and $(\alpha_\pm, \delta_\pm, \lambda) = (0.5, 0.5, 1)$, the corresponding figures are 0.0253839 and 0.0492624, representing a 94\% increase.
%{\color{red} Such a percentage, however, may decline rapidly as $T$ increases. For example, still consider $(\alpha_\pm, \delta_\pm, \lambda) = (0.5, 0.5, 1)$. When $T = 3$, the corresponding figures are 0.0701760 and 0.0819106, implying a 17\% increase; when $T = 4$, the corresponding figures are 0.101216 and 0.106440, a 5.2\% increase; when $T = 5$, the corresponding figures are 0.123364 and 0.124920, a 1.3\% increase; when $T = 6$, the corresponding figures are 0.145565 and 0.145780, only a 0.15\% increase.}
}

 Moreover, one can
achieve this improved  optimal value by implementing a Markovian randomization, with the overall strategy  very similar qualitatively to \citet{Barberis2012:Casino}'s -- both are of the loss-exit type.\footnote{We have also revisited the $T=6$ example considered also in
\citet{HeHuOblojZhou2014:OptimalCasinoBettingwhyluckycoinandgoodmemoryareimportant}. In that paper, a randomized strategy, found by trial and error, leads to the value function $V=0.250702$ compared with $V=0.250440$ for the best non-randomized strategy. Using our algorithm, we see that the best randomized strategy actually gives $V=0.257483$. It is still a loss-exit type and it stops at node $(0,0)$ with probability $0.201$, node $(5,1)$ with $0.436$, node $(5,3)$ with $0.0292$ and node $(5,5)$ with $0.0113$.}
% --, and ours does indeed improve (albeit only slightly in this particular instance) over \citet{Barberis2012:Casino}'s as the former allows path-dependent strategies whereas the latter permits only Markovian ones. Moreover, we can achieve our optimal CPT value by implementing randomization.

While the precommitted gambler follows through the optimal strategy originally determined  at time 0, a {\it na\"ive} gambler thought he would do the same but in actuality
constantly deviates from previously planned strategies. More precisely, at {\it any} time $t>0$, a naivet\'e {\it re-considers} the optimal stopping problem starting from $t$, devises a precommitted strategy but carries it out  for only {\it one} period (because he will re-optimize again at the next time instant). Here, we assume
this gambler keeps his initial wealth at time 0 as the reference point.\footnote{This is also the assumption made in \cite{Barberis2012:Casino} when analyzing a na\"ive gambler's behavior. It is both natural and plausible that a gambler remembers the initial amount of cash he brought into the casino and always compares wins and losses against that amount.}
The na\"ive gambler's strategy can be computed by deriving all the time-$t$ precommitted strategies, $t=0,1,...,T$, implementing each of them for just one period, and then ``pasting" them together.
As a result, his {\it actual} quitting strategy could be drastically  different from the precommitted one, the one he originally planned before he enters the casino;
 %By following the same procedure of solving time-0 problem \eqref{finitedimensionprogram}, we solve time-$t$ problem and therefore get to know the decision made by the naive gambler at each node after time 0 and therefore her actual gambling behavior. Assume the gambler does not change her reference point during her stay in casino, i.e., the gambler sticks to her initial wealth started at time 0 as the reference point.
 see the right panel of Figure \ref{Optimalstrategy}. There, the only
 node calling for randomization is now (2,0), with a probability of 0.179 to quit.\footnote{This is also the only node that makes our na\"ive strategy different from
 \cite{Barberis2012:Casino}'s in which the node (2,0) is white meaning ``continue";
 see the right panel of Figure 4 therein.}
 Comparing the two strategies depicted  in Figure \ref{Optimalstrategy}, we find that the na\"ive strategy is not only significantly different from the precommitted one, but indeed almost completely {\it opposite} in character: the latter is mainly to continue in the gain domain and to stop in the loss domain, while the former is reversed. For a discussion on experimental evidence on the dramatic departure of the actual gambling behaviors from the planned ones,
 see \citet[Section 4.3]{Barberis2012:Casino}.
 \cite{HIIW} presents strong evidence from lab and field that supports the inconsistent dynamic framework of \cite{Barberis2012:Casino}.
  In the context of stock trading, such a
 na\"ive behavior -- the tendency of selling winners too soon and keeping losers too long -- is widely observed especially for retail investors, and is termed the {\it disposition effect} by \cite{OdeanT:98de}. %This  implication agrees with that  in \citet{Barberis2012:Casino}.

\section{Discussions}\label{se:Discussion}
%\vspace{1ex}

\subsection{To enter or not to enter: the power of randomization}\label{enterornot}

One of the main takeaways of \cite{Barberis2012:Casino} is that CPT offers an explanation why a gambler would be willing to enter a casino even if the bets there have neither
skewness nor positive expected values. By implementing a loss-exit strategy, namely keep gambling when winning
but stop gambling when accumulating a sufficient loss, he envisions a  positively skewed probability distribution of the accumulated gain/loss at the exit time which
is favored by the CPT preference. However, he would need a sufficiently {\it long} time period to build such a skewed distribution in order to have a {\it positive}
CPT value to justify the entry (recall that the CPT value of not playing at all is zero). For the case of a piece-wise power utility function \eqref{eq:utility} and an inverse S-shaped weighting function \eqref{eq:distortion},
\citet[Proposition 1]{Barberis2012:Casino} provides a sufficient condition for this to happen.\footnote{\citet[Proposition 1]{Barberis2012:Casino} is stated for a na\"ive gambler. However, the result holds for a precommitter as well because both gamblers face the same problem at $t=0$.} Moreover, for the parameter values $\alpha_+ = \alpha_- = 0.88$, $\delta_+ = \delta_- = 0.65$, $\lambda = 2.25$, this sufficient condition translates into $T\geq 26$; see
\citet[Corollary 1]{Barberis2012:Casino}.\footnote{These  parameter values
are  close to those given by \cite{TverskyKahneman1992:CPT}, i.e., $\alpha_+ = \alpha_- = 0.88$, $\delta_+ = 0.61$, $\delta_- = 0.69$, $\lambda = 2.25$.
If we apply the exact \cite{TverskyKahneman1992:CPT} parameter values to \citet[Proposition 1]{Barberis2012:Casino}, then the corresponding  $T \geq 20$.
Such a shorter period is expected because the probability weighting in gains is stronger than that in losses with \cite{TverskyKahneman1992:CPT}'s parameters; thus
it takes less time to build the desired positively skewed distribution with a positive CPT value.
}

However, with randomization allowed, the gambler may be willing to enter the casino even if he is allowed to play only {\it once} (i.e., $T=1$).

%Consider piece-wise power utility function \eqref{eq:utility} and inverse S-shaped probability weighting function \eqref{eq:distortion}. Suppose parameters take estimated values given by \cite{TverskyKahneman1992:CPT}, i.e., $\alpha_+ = \alpha_- = 0.88$, $\delta_+ = 0.61$, $\delta_- = 0.69$, $\lambda = 2.25$. With randomized strategies allowed, the optimal CPT preference value is positive for $T \ge 1$.
%Note that the probability distortion in the gain region is slightly heavier than the probability distortion in the loss region. Then, even in relatively short horizons, e.g., one horizon, a randomized gambling strategy could be strictly preferred to just non-gamble. The proposition below gives a sufficient condition for participating in one horizon.
\begin{proposition}\label{prop:onehrand}
Suppose $T=1$. If $\lim_{p \to 0} [w_+'(p)/w_-'(p)] > \lambda[u_-(1)/u_+(1)]$, then the optimal CPT value is strictly
positive.
\end{proposition}

Recall that in our model, randomization is available; so the optimal CPT value being strictly positive means that the gambler will enter the casino, possibly tossing a coin
to decide whether to actually play (the only) one round of bet.

What if $T\geq 2$?
Naturally, as $T$ increases, the optimal CPT values increase. Figure \ref{pic3} graphs the optimal CPT values for $T = 1,2,...,20$ with the \cite{TverskyKahneman1992:CPT} estimates. Therefore, if the gambler will enter
the casino for $T=1$ with a given set of parameters, so will he for $T\ge2$ with the same parameters. As a consequence, Proposition \ref{prop:onehrand} holds for $T\geq 2$ as well.

\begin{figure}
  \centering
    \includegraphics[width=100mm]{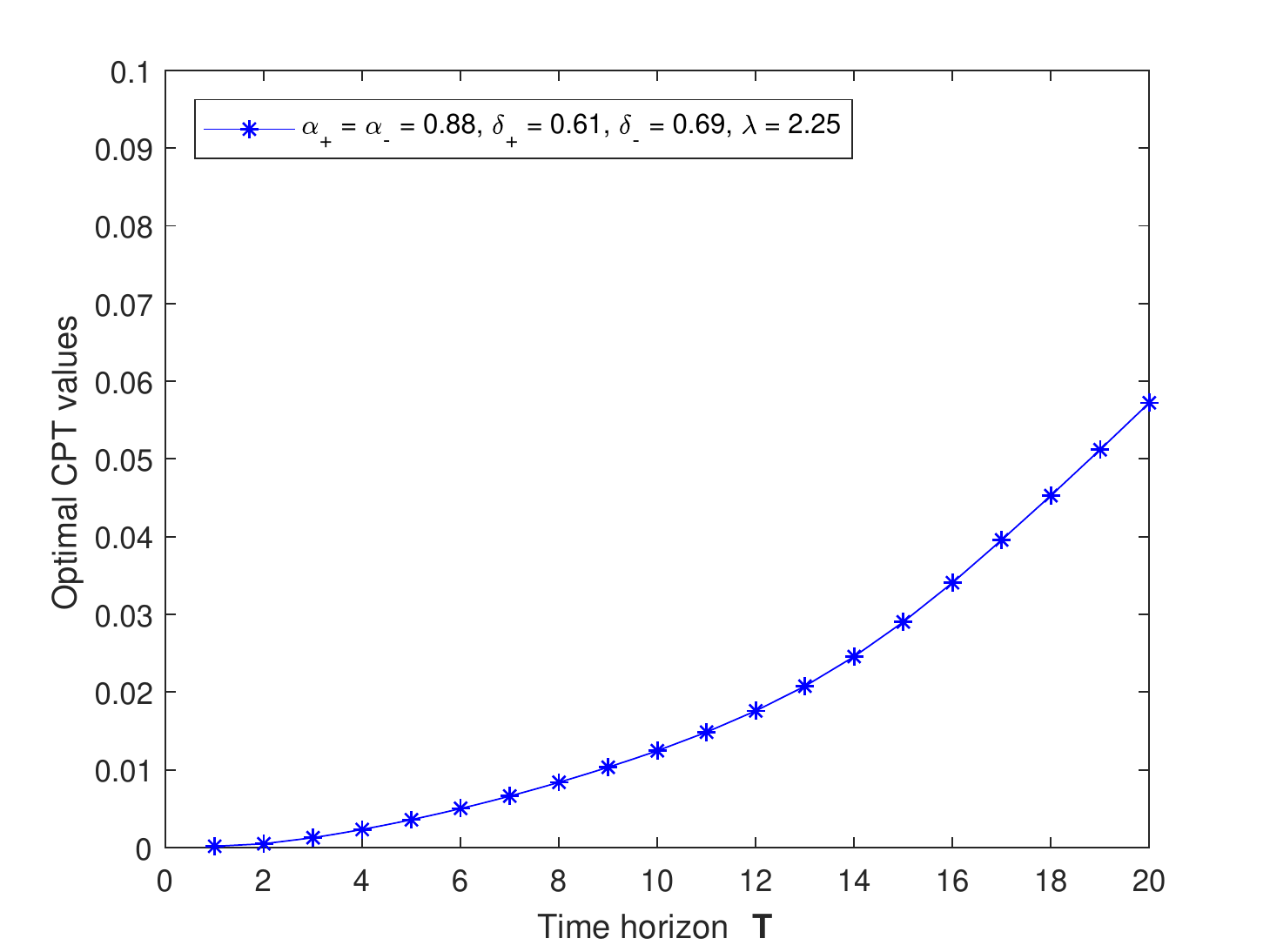}
  \caption{\small Optimal CPT values for $T = 1,2,...20$ under the parameter values of \cite{TverskyKahneman1992:CPT}, i.e., $\alpha_+ = \alpha_- = 0.88$, $\delta_+ = 0.61$, $\delta_- = 0.69$, $\lambda = 2.25$.}\label{pic3}
\end{figure}

It is straightforward to show that for the weighting function \eqref{eq:distortion}, when $\delta_+ < \delta_-$ (which is the case with
\cite{TverskyKahneman1992:CPT}'s estimates), we have $\lim_{p \to 0} [w_+'(p)/w_-'(p)] = +\infty$. Hence, Proposition \ref{prop:onehrand} yields that, as long as the loss-aversion degree $\lambda$ is finite, a randomized gambling strategy is always preferred to non-gamble, even when $T=1$.

The intuition of Proposition \ref{prop:onehrand} is as follows. The condition
$\lim_{p \to 0} [w_+'(p)/w_-'(p)] > \lambda[u_-(1)/u_+(1)]$ means that exaggeration  of big gains outweighs exaggeration  of big losses and loss aversion {\it combined}; so the gambler assigns a positive CPT value to a sufficiently positively skewed distribution. Without randomization, it would take time to build such a distribution. With randomization, however, the gambler can design a coin {\it right away} with such a distribution, saving him all the time otherwise needed. In other words, a coin toss can be used to supersede all the time-consuming (and perhaps clever) maneuvers  to reach the desired distribution. Note that even though randomization still gives rise to a {\it symmetric} distribution of gains and losses and hence the loss aversion seemingly would prevent the gambler from entering, the sufficiently unequal levels of probability weighting on gains and losses, as stipulated by the condition $\lim_{p \to 0} [w_+'(p)/w_-'(p)] > \lambda[u_-(1)/u_+(1)]$, yields the contrary.

On the other hand, the effectiveness of randomization crucially depends on the chosen parameters. If the degree of probability weighting in gains is equal to or less than that in losses, and the level of loss-aversion  is sufficiently large so that
$\lambda [u_-(1)/u_+(1)]>1$, then the above proposition does not apply because $\lim_{p \to 0} [w_+'(p)/w_-'(p)] \le 1 < \lambda [u_-(1)/u_+(1)]$. For example, for utility function \eqref{eq:utility} and probability weighting function \eqref{eq:distortion}, let $\alpha_+ = \alpha_- = 0.88$, $\delta_+ = \delta_- = 0.65$, $\lambda = 2.25$, the values used in \citet[Corollary 1]{Barberis2012:Casino}. In this case, we find a positive CPT value of {\it randomized} strategies only when the time horizon is at least   $T=25$, slightly shorter than that ($T = 26$) presented in \citet[Corollary 1]{Barberis2012:Casino} using a special {\it non-randomized} loss-exit strategy.
The corresponding optimal precommitted strategy is also of a loss-exit type: stop once losing \$1, and continue with possible randomization when wining.
If we further let the probability weighting  in gains be weaker than that in losses, e.g., $\delta_+ = 0.69$ and $\delta_- = 0.61$, then a positive preference value is found only at $T \ge 39$.

\subsection{Alpha, delta, and lambda}\label{adl}

There are three components in the risk/loss preferences under CPT: the utility function, the probability weighting and the loss aversion. They are intertwined and
compete with each other in determining the overall preference and dictating the final behavior. In this subsection, we study the roles they play in the case when
the utility function is \eqref{eq:utility} and the  weighting function is \eqref{eq:distortion}, with parameters $\alpha_\pm$, $\delta_\pm$ and $\lambda$.

Others being kept unchanged, the effect of each  of these parameters  is as follows: a smaller $\alpha_+$ implies a higher degree of risk-aversion in gains and a smaller $\alpha_-$ implies a higher degree of risk-seeking in losses, a smaller $\delta_\pm$ yields a higher level of probability weighting in gains/losses, and a smaller $\lambda$ indicates a smaller extent of
loss aversion.
To understand  the overall impact of these parameters on exit decisions, we first fix $\lambda$ and consider four sets of scenarios: large $\alpha_\pm$ and small $\delta_\pm$; small $\alpha_\pm$ and large $\delta_\pm$; small $\alpha_\pm$ and small $\delta_\pm$; and large $\alpha_\pm$ and large $\delta_\pm$. Then we examine the effect of $\lambda$.
In the following discussions we fix $T = 10$.

%The precommitted strategies for $T = 10$ under these two sets of parameters are depicted  in Figure \ref{tree_10}. In the left panel, the parameters are $\alpha_+ = \alpha_- = 0.95$, $\delta_+ = \delta_- = 0.5$, $\lambda = 1.5$. These are the same parameters as in the numerical example presented in Section 4.2, but now we have a much longer horizon. The precommitted strategy is still  a loss-exit one.
%The intuition is as follows. This is the case where $\alpha_{\pm}$ are relatively high (lower risk-aversion) and $\delta_{\pm}$ relatively low (heavier probability weighting). In the gain region, the stronger exaggeration of the small probability of winning a large amount outweighs the weaker risk aversion; hence the gambler is
%willing to take more risk and stay longer. In the loss region,  the stronger exaggeration of the small probability of losing a large amount, together with the loss aversion,  outweighs the weaker risk-seeking appetite and prompts the gambler to
%play safe and quit earlier. The argument is reversed, leading to
%a gain-exit type of strategy,  when $\alpha_{\pm}$ are relatively low and $\delta_{\pm}$ relatively high, such as the one  in the right panel where $\alpha_+ = \alpha_- = 0.5$, $\delta_+ = \delta_- = 0.95$, $\lambda = 1.5$.
%This analysis shows that whether the strategy is loss-exit or otherwise really depends on the interplay between three intertwining and competing forces represented by $\alpha_{\pm}$, $\delta_{\pm}$, and $\delta$.

The left panel of  Figure \ref{tree_10_l} draws the optimal precommitted strategy when $\alpha_\pm = 0.95$, $\delta_\pm = 0.5$ and $\lambda = 1.5$. These are the same parameters as in the numerical example presented in Subsection \ref{numerical}, except now we have a much longer horizon. Again,  black nodes mean ``stop'', white nodes ``continue'', grey nodes ``randomize'', and the number above the grey node is the probability to stop. This  strategy is mainly to continue or toss a coin in gains until the final time and to stop in losses, which is thus a loss-exit one. The intuition is as follows. This is the case where $\alpha_{\pm}$ are relatively large (lower risk-aversion/-seeking) and $\delta_{\pm}$ relatively small (heavier probability weighting). In the gain region, the stronger exaggeration of the small probability of winning a large amount outweighs the weaker risk aversion; hence the gambler is
willing to take more risk and stay longer. In the loss region,  the stronger exaggeration of the small probability of losing a large amount, together with the loss aversion,  outweighs the weaker risk-seeking appetite and prompts the gambler to
play safe and quit earlier.

\begin{figure}
  \centering
  \begin{minipage}[t]{0.49\textwidth}
    \includegraphics[width=\textwidth]{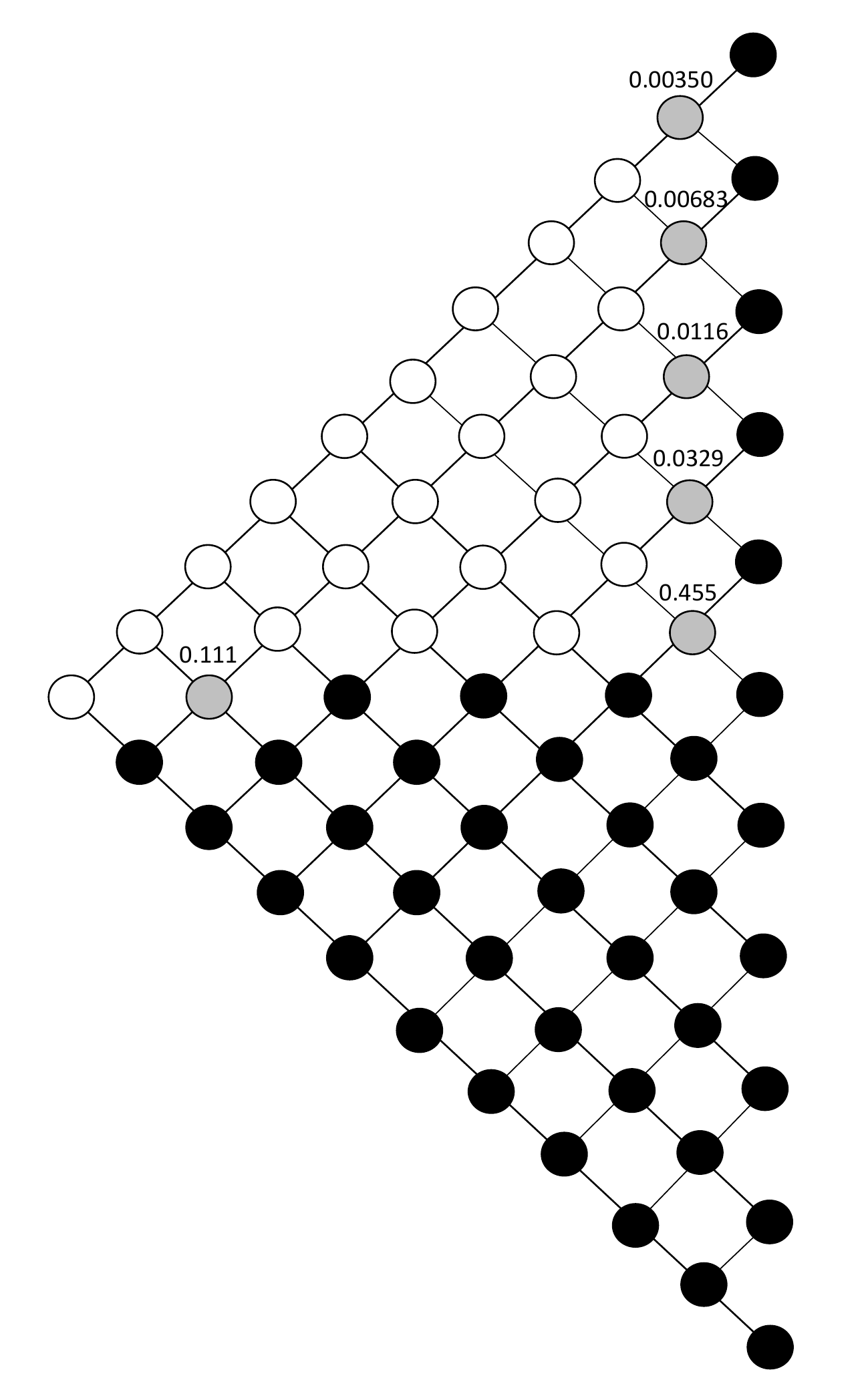}
  \end{minipage}
  \begin{minipage}[t]{0.49\textwidth}
    \includegraphics[width=\textwidth]{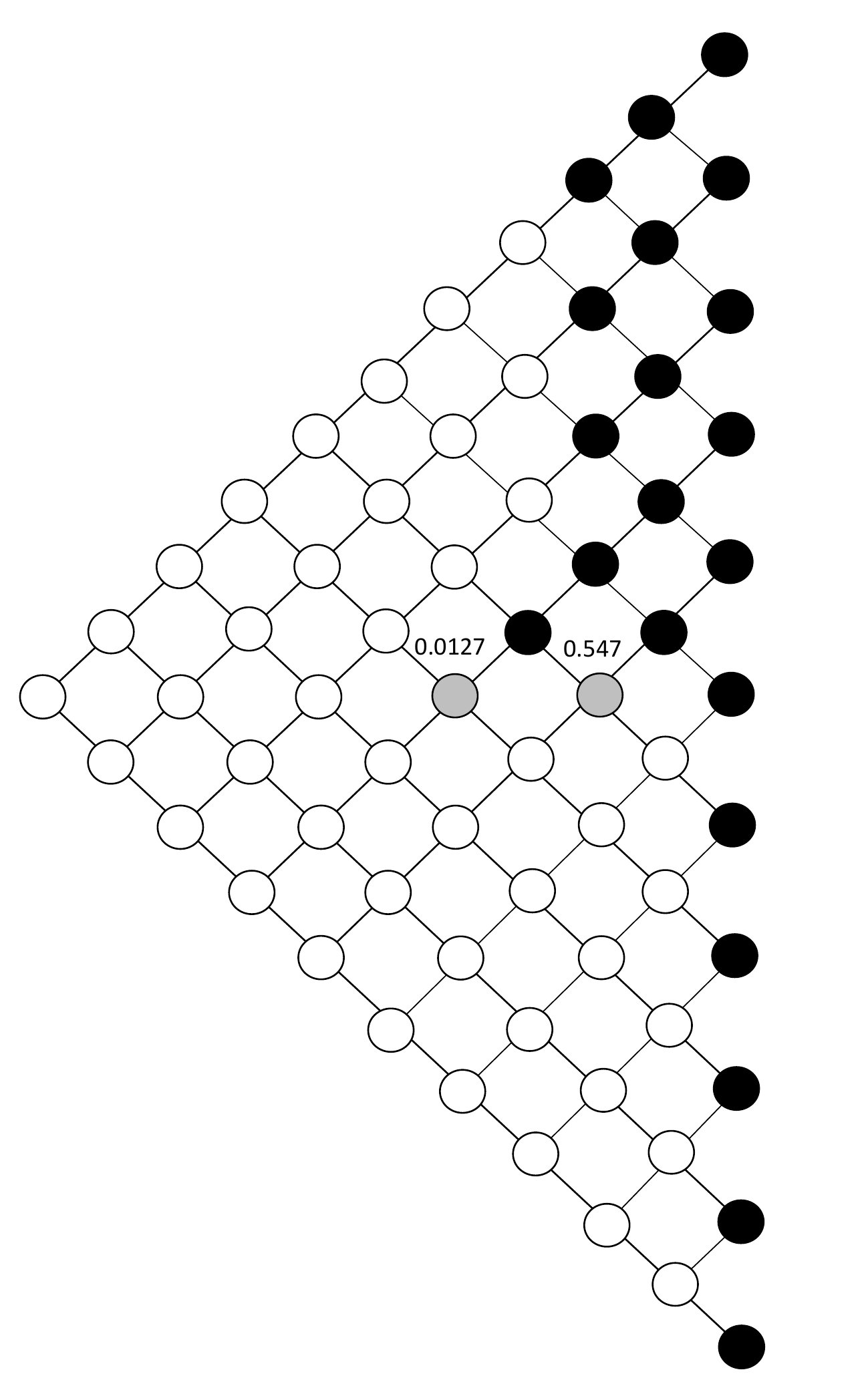}
  \end{minipage}
  \caption{\small The precommitted (left panel) and na\"ive (right panel) strategies for $T = 10$, $\alpha_\pm = 0.95$, $\delta_\pm = 0.5$, $\lambda = 1.5$. Black nodes are ``stop'', white nodes are ``continue'', and grey nodes are ``randomize''. The numbers above the grey nodes are the probabilities to stop.}\label{tree_10_l}
\end{figure}

The above argument is reversed, leading to
a gain-exit type of strategy,  when $\alpha_{\pm}$ are relatively small and $\delta_{\pm}$ relatively large, such as the one depicted in the left panel of Figure \ref{tree_10_g} where $\alpha_\pm = 0.5$, $\delta_\pm = 0.95$, $\lambda = 1.5$. An interesting small variation of this case is when
probability weighting is absent, i.e., $\alpha_\pm = 0.5$, $\delta_\pm = 1$, $\lambda = 1.5$, in which the optimal CPT value is positive  and the precommitted strategy is still a gain-exit one. Indeed, a positive preference value is found at a much shorter horizon $T=4$ under this group of parameters, and the optimal distribution of $S_\tau$ is left-skewed (which is favored by a strong risk-seeking preference in losses represented by $\alpha_-$).

\begin{figure}
  \centering
  \begin{minipage}[t]{0.49\textwidth}
    \includegraphics[width=\textwidth]{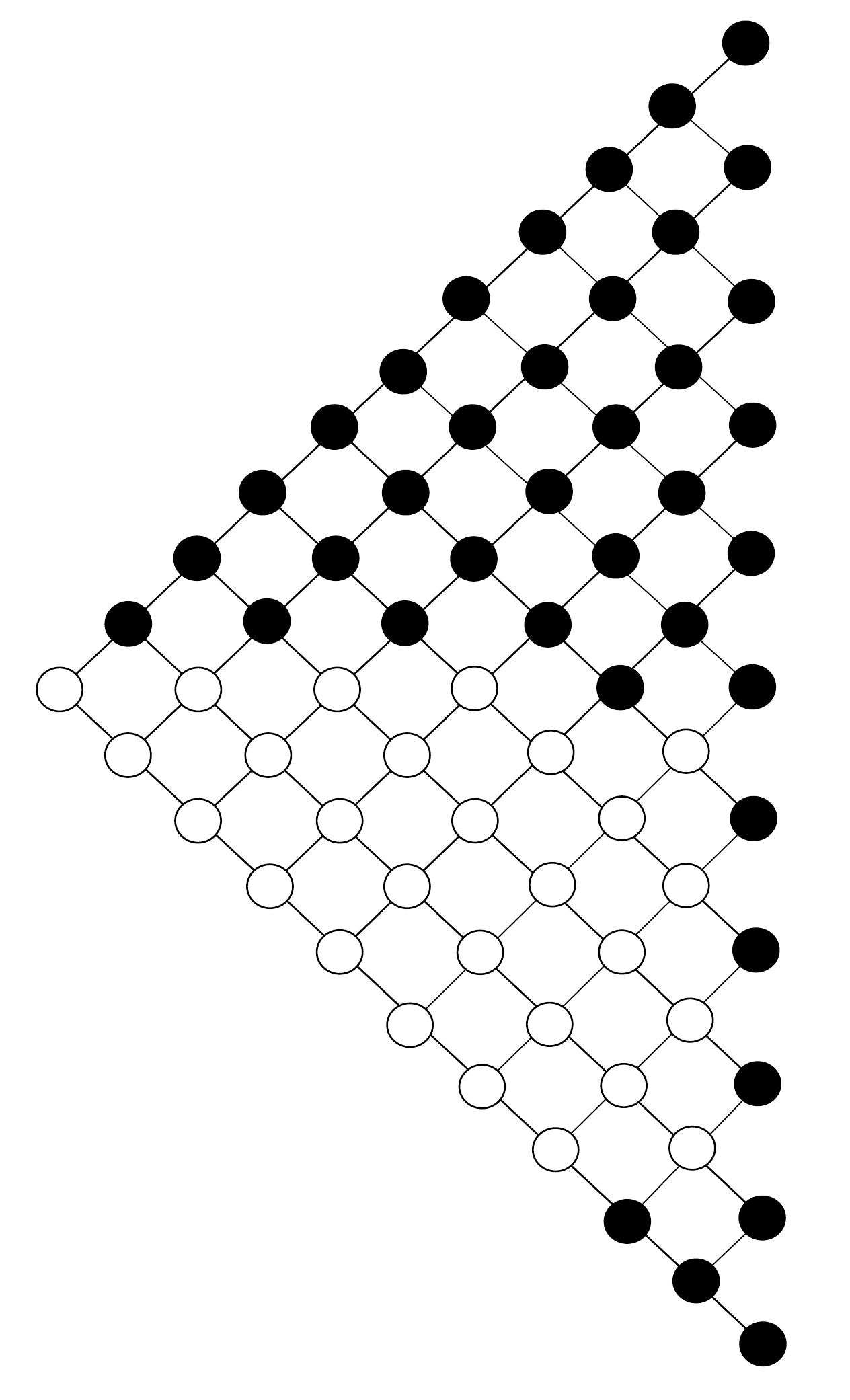}
  \end{minipage}
  \begin{minipage}[t]{0.49\textwidth}
    \includegraphics[width=\textwidth]{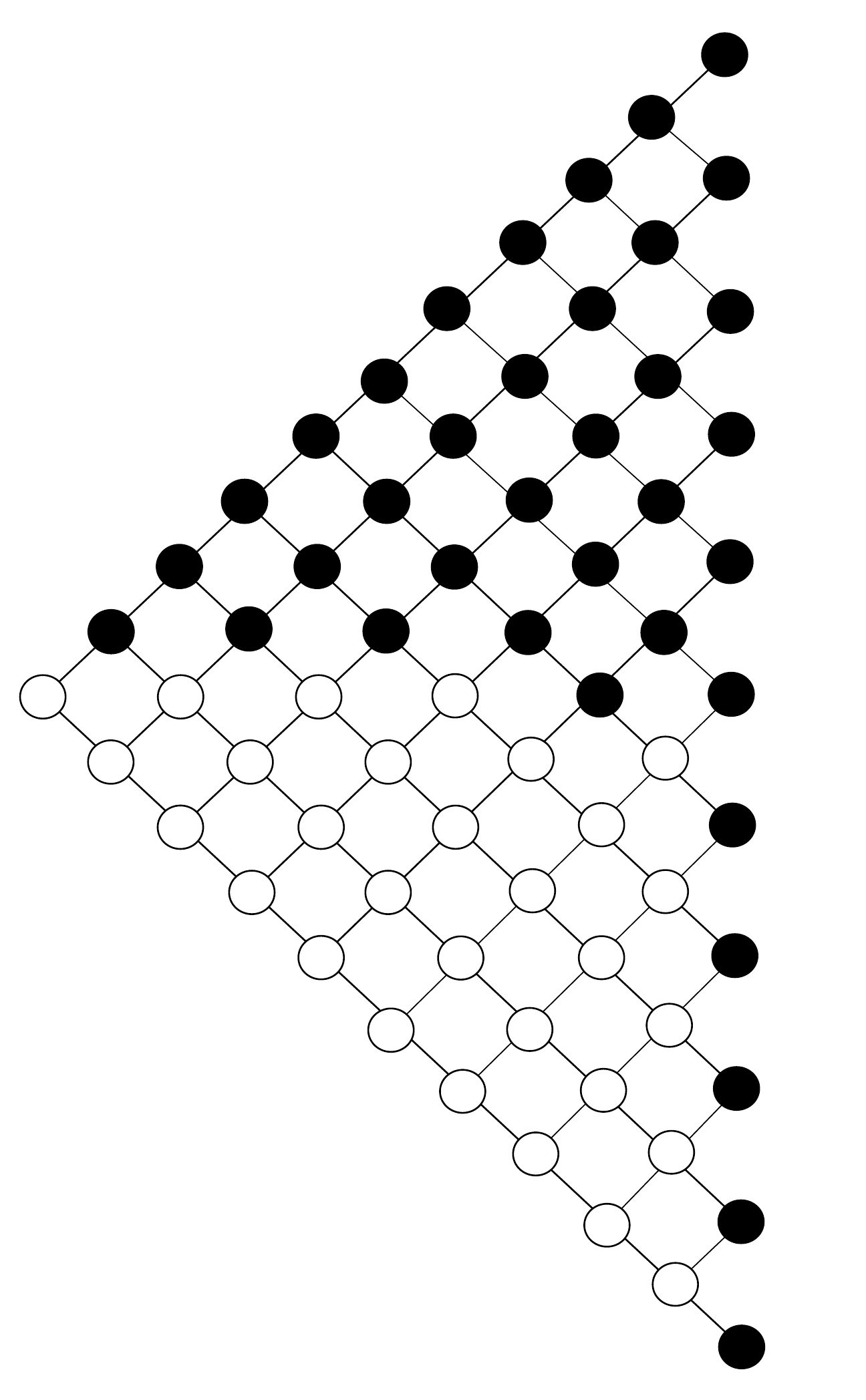}
  \end{minipage}
  \caption{\small The precommitted (left panel) and na\"ive (right panel) strategies for $T = 10$, $\alpha_\pm = 0.5$, $\delta_\pm = 0.95$, $\lambda = 1.5$. Black nodes are ``stop'' and white nodes are ``continue''. There is no grey node.}\label{tree_10_g}
\end{figure}

%This analysis shows that whether the strategy is loss-exit or otherwise really depends on the interplay between three intertwining and competing forces represented by $\alpha_{\pm}$, $\delta_{\pm}$, and $\delta$.

%Cases are more complicated when the effects from $\alpha$ and $\delta$ on gambler's risk attitude are not in the same direction. In
The left panel of Figure \ref{tree_10_l2} shows the precommitted strategy for the parameter values  $\alpha_\pm =  \delta_\pm = 0.5$ and $\lambda = 1.5$, which is the case of small $\alpha_\pm$ and small $\delta_\pm$. This is still a loss-exit strategy, but the main differences from that visualized by the left panel of Figure \ref{tree_10_l} are that, in the gain region, there are now more black nodes  and the numbers above
the grey nodes are larger, implying a higher likelihood of stop even when the gambler has accumulated a gain. The reason is that with a smaller $\alpha_+$, the exaggeration of the small probability of winning a large amount still outweighs the risk aversion in gains, but with a lesser degree than the previous case.

\begin{figure}
  \centering
  \begin{minipage}[t]{0.49\textwidth}
    \includegraphics[width=\textwidth]{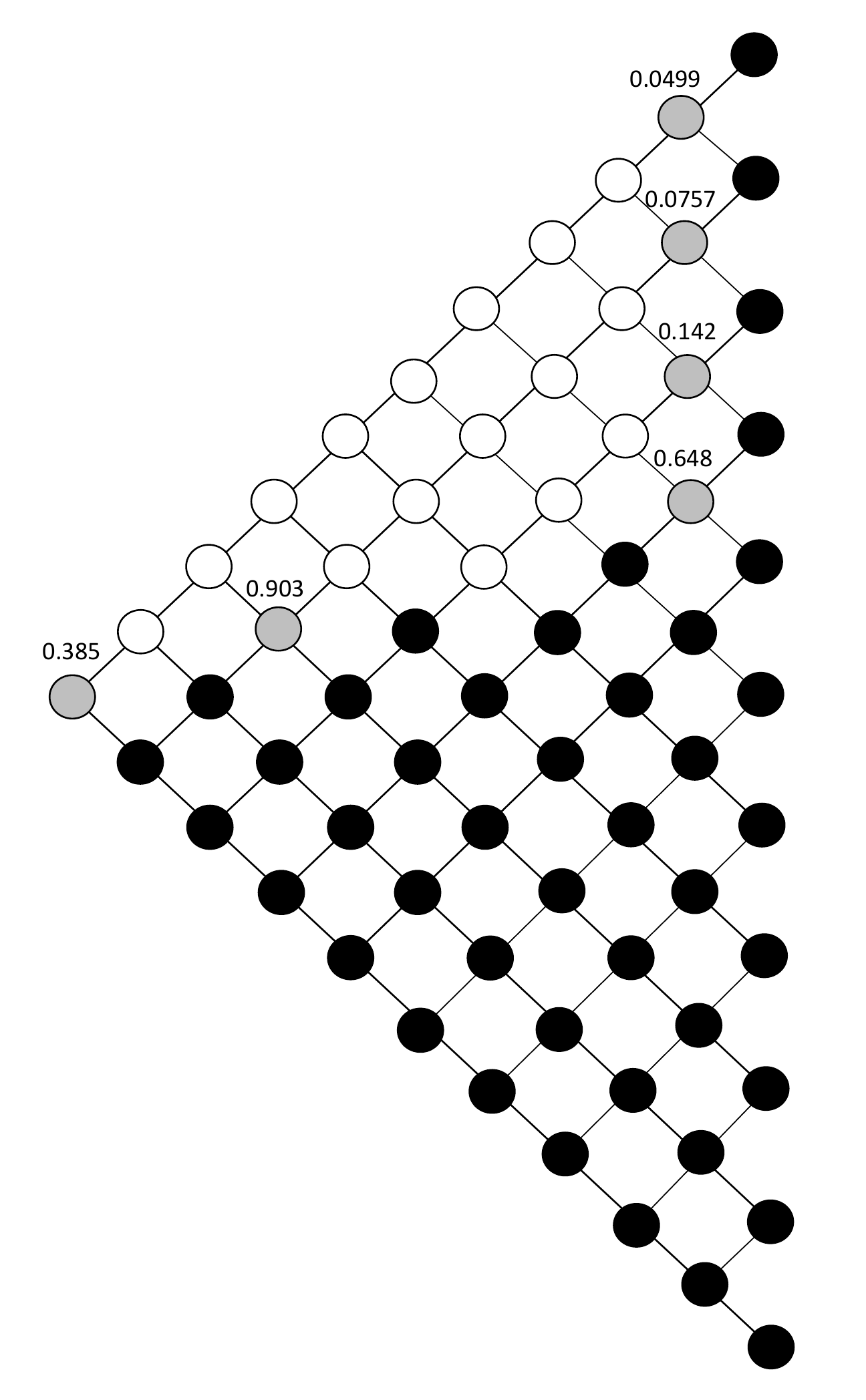}
  \end{minipage}
  \begin{minipage}[t]{0.49\textwidth}
    \includegraphics[width=\textwidth]{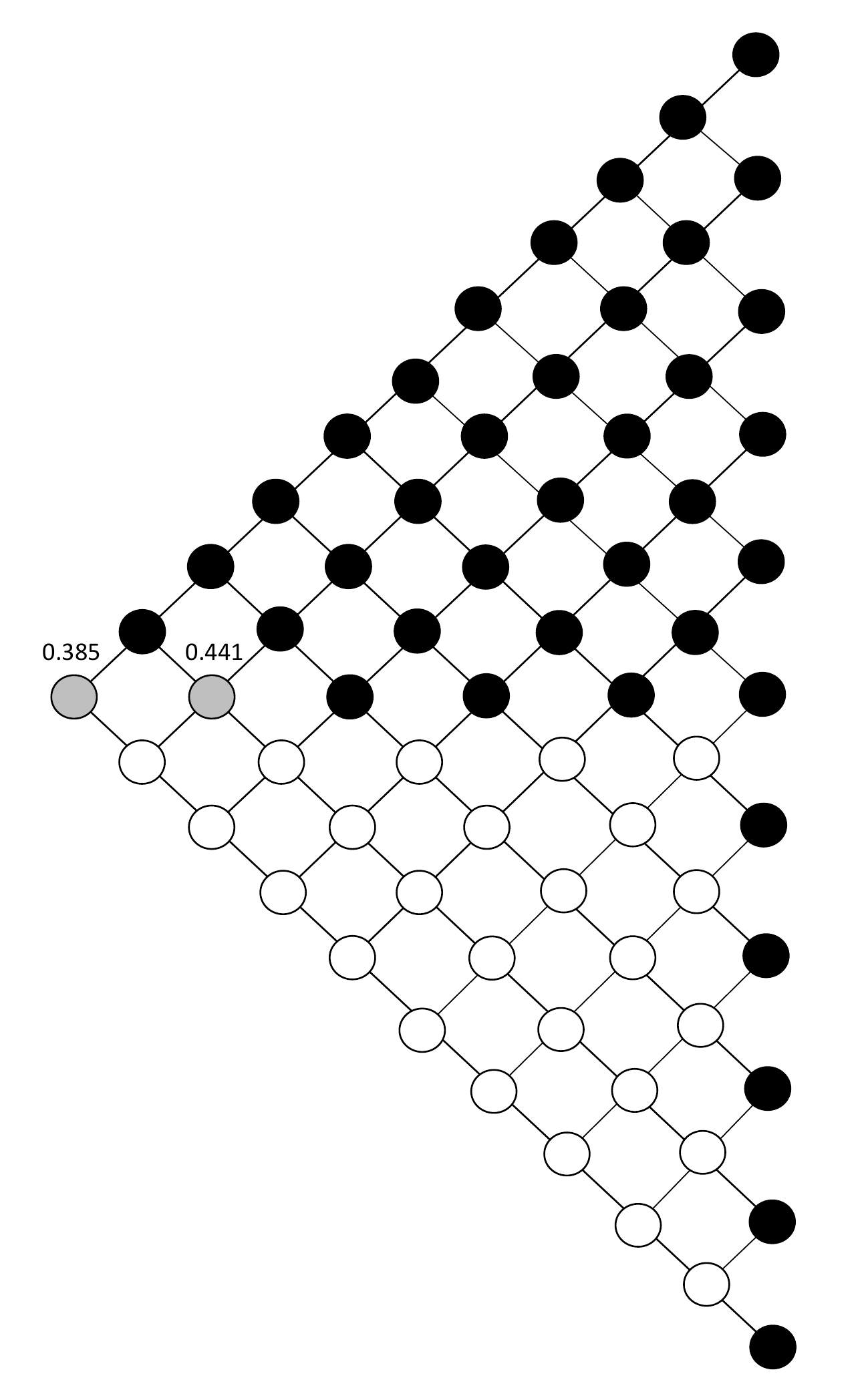}
  \end{minipage}
  \caption{\small The precommitted (left panel) and na\"ive (right panel) strategy for $T = 10$, $\alpha_\pm = \delta_\pm = 0.5$, $\lambda = 1.5$. Black nodes are ``stop'', white nodes are ``continue'', and grey nodes are ``randomize''. The numbers above the grey nodes are the probabilities to stop. }\label{tree_10_l2}
\end{figure}

The last set of parameters are  $\alpha_\pm =\delta_\pm = 0.95$, $\lambda = 1.5$, with which the optimal CPT value is  zero and the gambler will simply not enter the casino. This is because
these parameter values render a risk preference close to both risk-neutral and probability--weighting--free, while  a zero-mean bet and a loss-aversion degree $\lambda>1$ prevent the gambler from  playing the game at all.

The impact of $\lambda$ is more straightforward, which we now examine. For each group of $\alpha_\pm$ and $\delta_\pm$ considered above, we obtain the optimal CPT value by varying $\lambda$ from 1 to 3; see Figure \ref{pic}, the left panel. Quite naturally, each of the optimal CPT  values decreases as $\lambda$ increases,  and three of them hit  zero before $\lambda$ reaches 3. As a result, the gambler will be increasingly reluctant to stay in or even enter the casino as his level of loss aversion increases.

\begin{figure}
  \centering
  \begin{minipage}[t]{0.49\textwidth}
    \includegraphics[width=\textwidth]{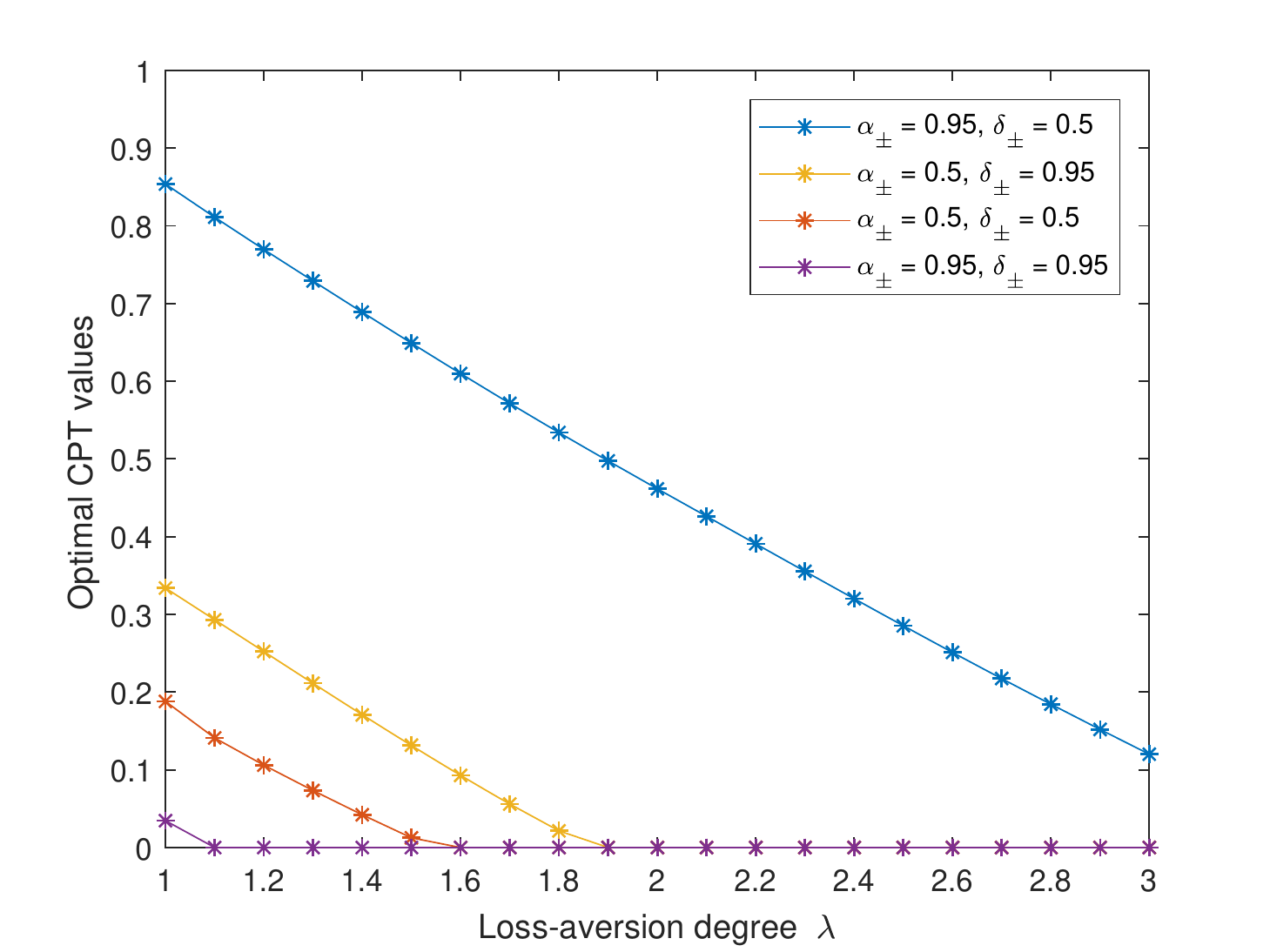}
  \end{minipage}
  \begin{minipage}[t]{0.49\textwidth}
    \includegraphics[width=\textwidth]{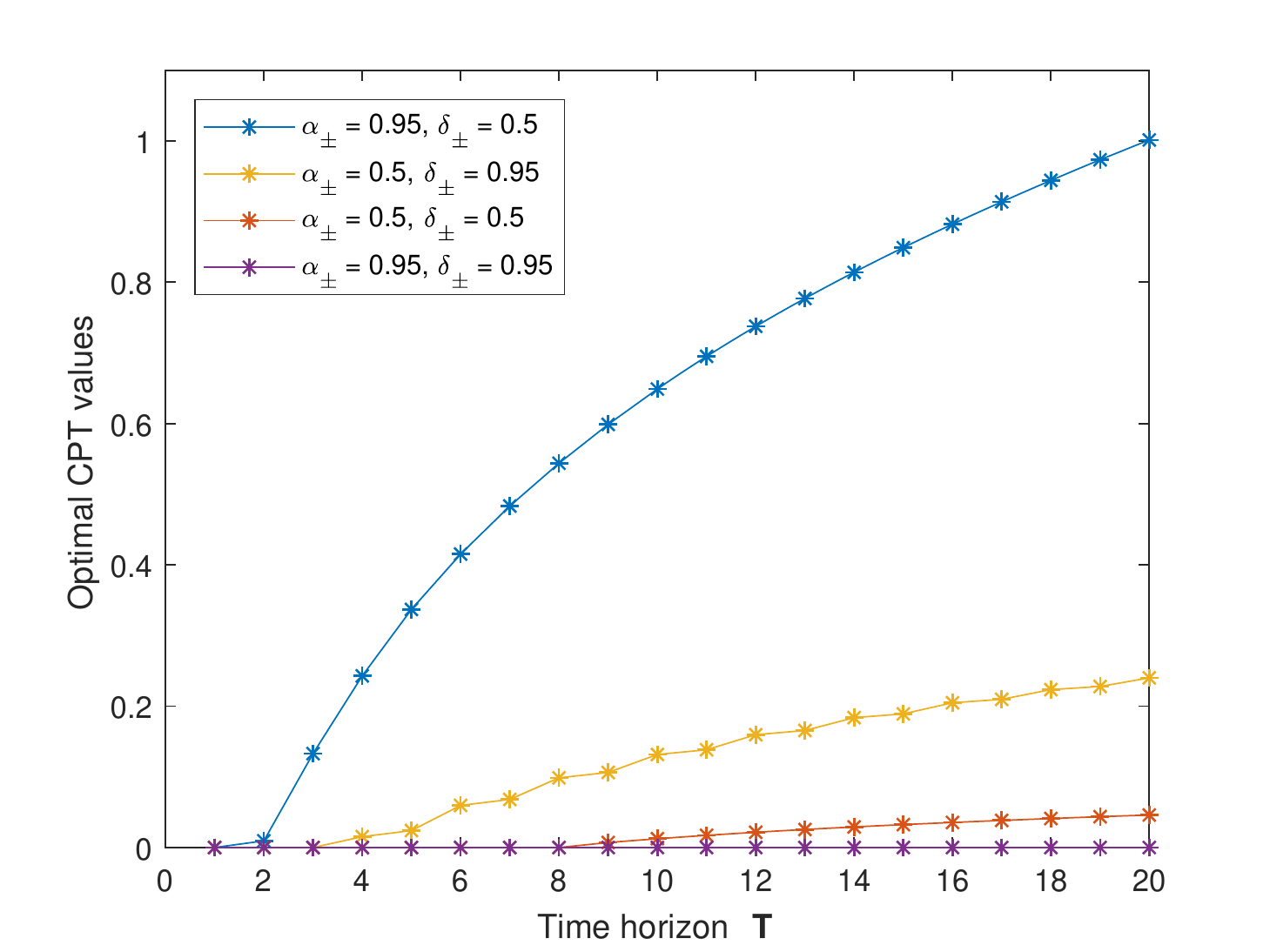}
  \end{minipage}
  \caption{\small Optimal CPT values for $\lambda$ from 1 to 3, while $T = 10$, are shown in the left panel. Optimal CPT values for $T = 1,...20$, while $\lambda = 1.5$, are shown in the right panel. In both panels, $(\alpha_\pm, \delta_\pm) \in \{ (0.95,0.5), (0.5,0.95),(0.5,0.5),(0.95,0.95)\}$. }\label{pic}
\end{figure}

The analysis in this subsection shows that the CPT casino modeling with various
constellations of parameter specifications can predict and explain
a rich array of gambler behaviors. In particular,  whether the strategy is loss-exit or otherwise depends on the interplay between the three intertwining and competing forces represented by $\alpha_{\pm}$, $\delta_{\pm}$, and $\lambda$.

%On the other hand, in the right panel while $\alpha_+ = \alpha_- = 0.5$, $\delta_+ = \delta_- = 0.95$, $\lambda = 1.5$, the strategy is essentially a gain-exit type. In fact, when $\alpha_+$ are relatively high and $\delta_+$ relatively low, the gambler is attracted by the exaggerated small probability of winning a large amount and is not much discouraged by the risk-aversion component due to the concavity of the utility in the gain region. In other words, risk-seeking effect dominates in the gain region. On the other hand, when $\alpha_-$ are relatively low and $\delta_-$ relatively high, the gambler is discouraged by the exaggerated small probability of losing a large amount and is not much attracted by the risk-seeking feature due to convexity of the utility in the loss region. In this case risk-aversion effect dominates in the loss region. Put them together, the optimal strategy turns out to be a loss-exit type. Results are reversed when $\alpha_{\pm}$ are relatively low and $\delta_{\pm}$ relatively high in which case the optimal strategy turns out to be a gain-exit type.

\subsection{One more round?}\label{omr}

With a longer time horizon a precommitter is more likely to obtain a positive CPT preference value and hence more likely  to enter the casino because, trivially, the optimal CPT value for $T$ is no less than that for $T-1$. On the other hand, with a longer time horizon and a loss-exit strategy one can possibly construct a more positively skewed probability distribution of the accumulated gain/loss at the exit time which, under CPT preference, is preferred by the precommitter. Hence, the optimal preference value may {\it strictly} increase as $T$ increases, which is demonstrated in the right panel of Figure \ref{pic}  where the optimal CPT values for $T = 1,2,...20$ under different groups of parameters are plotted.

So, the overall CPT value will be heightened  {\it if} the gambler is told to be granted an {\it additional} round of bet than previously agreed. But would he {\it always} take advantage of this extended time horizon and {\it actually} play the additional round?
%From Figure \ref{tree_10} we also note that randomization is more likely to happen in the gain region than in the loss region for a precommitter.
It turns out that the answer can be totally different depending on whether the gambler is   in the gain region or in the loss region.

Let the original problem have a horizon $T >0$ and $\tau\in\cT_T$ be a given exiting strategy. Assume $p_T:= \prob(S_{\tau} = T) > 0$ and consider the scenario in which the gambler has reached the upper most node
$(T,T)$ under $\tau$, namely $\tau=T$ and $S_{\tau} = T$.
%Assume $p_n := \prob(S_{\tau} = n) > 0$,
%$n=0,1,...,T$, the probability of leaving with an accumulated win of  $n$ under the strategy $\tau$.
Now suppose the time horizon is expanded to $T+1$ so the
gambler is allowed to play one more round. Firstly, we are interested in knowing,
{\it given} $\tau=T$ and $S_{\tau} = T$, namely the gambler has already played the originally final bet with the maximal possible accumulated win of $T$,  whether the gambler would {\it actually} take this opportunity and play one more time to possibly achieve a final accumulated gain of $T+1$ or $T-1$.\footnote{Bear in mind all the decisions are made at $t=0$ as we are considering precommitted strategies. So we are studying this problem from the vantage point of $t=0$.} The situation  is illustrated in the left panel of Figure \ref{tree_5}. Recall that randomization is allowed at any time; so let us denote by $r_T \in [0,1]$ the probability to stop at $S_{\tau} = T$, and
by $\tau'$ the strategy appending the original $\tau$ by, given $\tau=T$ and $S_{\tau} = T$, playing one more round with probability $1-r_T$ at time $T$ and finally stopping at time $T+1$.
 Let $q_T = \frac{1-r_T}{2}\in[0,\frac{1}{2}]$.
 The decumulative distribution of $S_{\tau'}$ differs from that of $S_\tau$ only at $\prob(S_{\tau'} \geq T+1) = q_T p_T$ and $\prob(S_{\tau'}\geq T) = (1-q_T)p_T$. The problem now is to choose $q_T$ to maximize $V(S_{\tau'})$ or, equivalently, to maximize
  %the {\it extra} CPT value
% due to the possible additional round of play, as a function of $q_n$, is %Then it is straightforward to see that the optimal $q*$ should maximize
\begin{align*}
w_+(q_T p_T) [u_+(T+1) - u_+(T)] + w_+\big((1-q_T)p_T \big) [u_+(T) - u_+(T-1)].
\end{align*}
%We would like to choose $q^*_n$ to {\it maximize} the above additional value.
For $T$ large enough, both $q_Tp_T$ and $(1-q_T) p_T$ are small enough to fall into the {\it concave} region of the probability weighting function $w_+(\cdot)$. Hence the above is a concave maximization and the following first-order condition is necessary and sufficient for a maximum $q^*_T$: \begin{align*}
0 = \left\{ w_+'(q_T^*p_T) [u_+(T+1) - u_+(T)] - w_+'\big((1-q_T^*)p_T \big) [u_+(T) - u_+(T-1)] \right \}p_T,
\end{align*}
or equivalently,
\begin{equation}\label{qstar}
\frac{w_+'(q^*_T p_T)}{ w_+'\big((1-q^*_T)p_T \big)} = \frac{u_+(T) - u_+(T-1)}{u_+(T+1) - u_+(T)} \;.
\end{equation}
Assuming $u_+(\cdot)$ is {\it strictly}  concave  (e.g.
that given by \eqref{eq:utility}), the right hand side of (\ref{qstar}) is strictly greater than one. Hence, the equation is satisfied by some $q^*_T \in (0, \frac{1}{2})$, but {\it not} $q^*_T=0$ (noting $w_+'(0)=+\infty$) or $q^*_T=\frac{1}{2}$. Recall that $q^*_T=0$ and  $q^*_T=\frac{1}{2}$ correspond to $r_T=1$ and $r_T=0$ respectively.
So, given the gambler has already played until the end  with a sufficiently accumulated gain (so that $p_T$
is sufficiently small), once he is allowed to play (only) one more time he will {\it not} have a black-and-white decision of either ``continue" or ``stop"; rather he will {\it always} engage in randomization to make his decision.\footnote{This also explains  why randomization happens at $T-1$ in the gain region when he has one final bet to play, as the left panels of Figures 6, 8, 10 indicate.}
Moreover, as $T$ increases the right hand side of (\ref{qstar}) decreases; hence
$q^*_T$ increases or $r_T$ decreases. In other words, the more gains accumulated, the more likely the gambler will continue. %This also reconciles with the values of the stopping probabilities indicated in the left panels of Figures 6, 8, 10.

\begin{figure}
  \centering
  \begin{minipage}[t]{0.33\textwidth}
    \includegraphics[width=\textwidth]{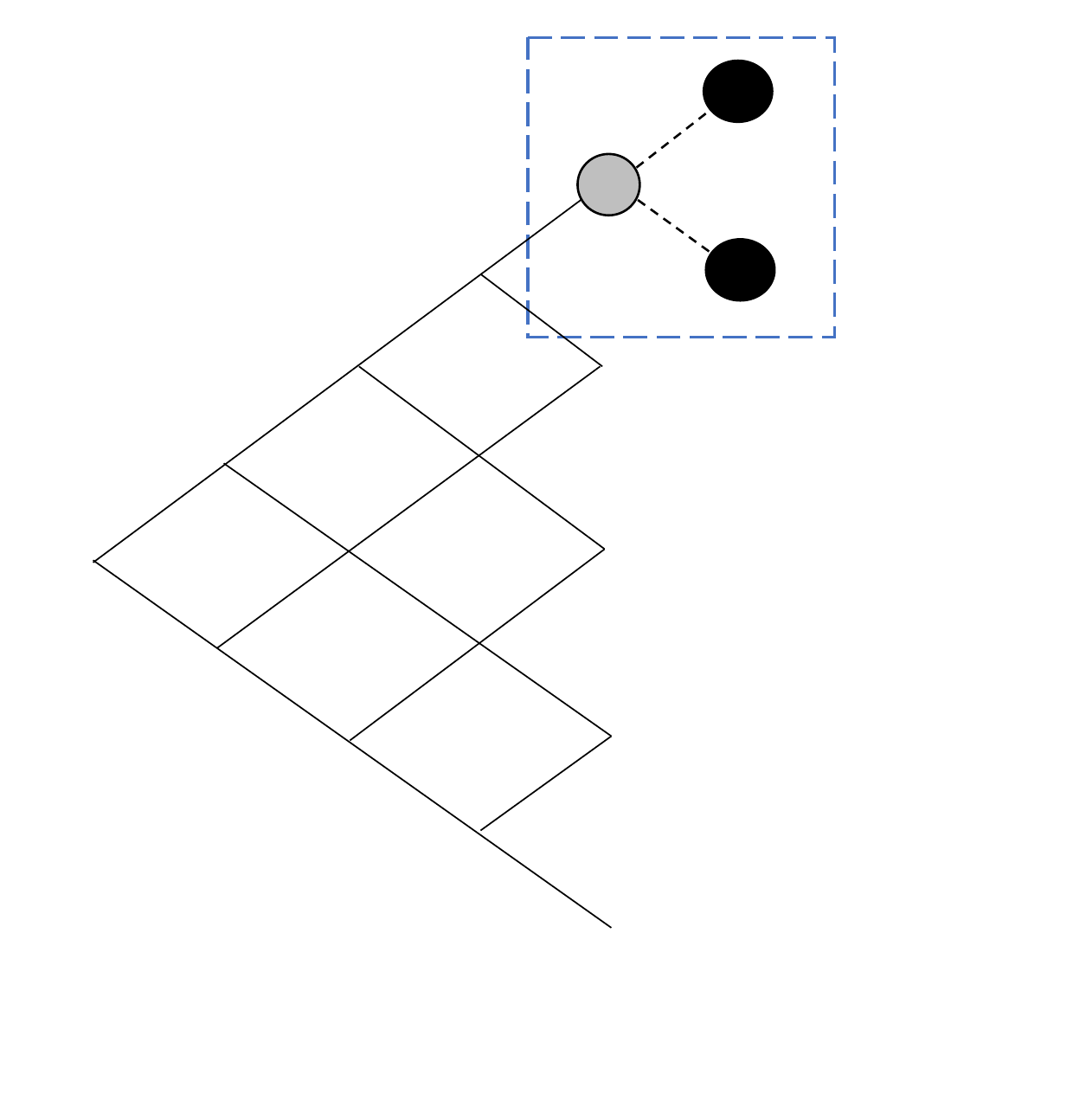}
  \end{minipage}
  \begin{minipage}[t]{0.33\textwidth}
    \includegraphics[width=\textwidth]{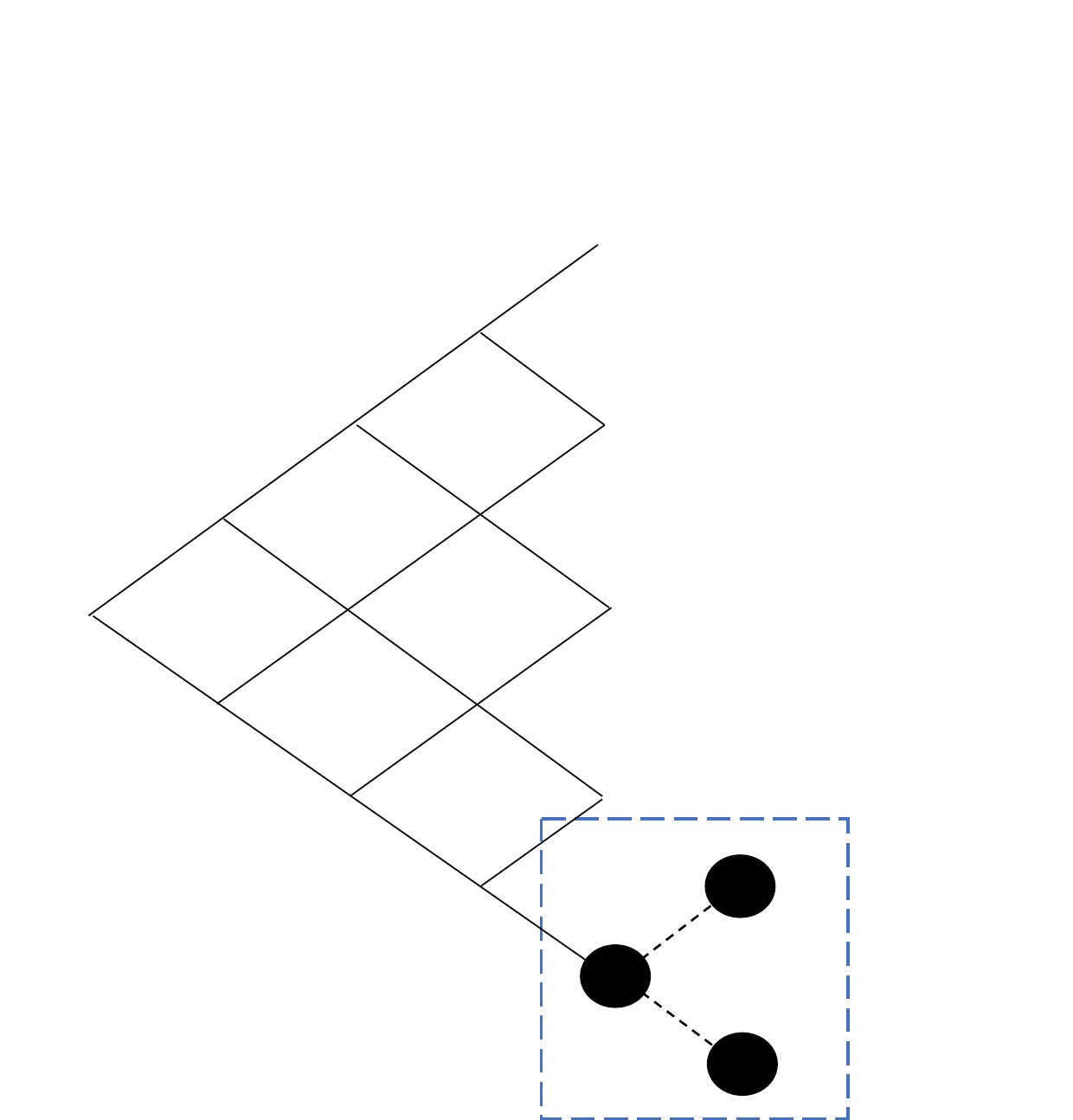}
  \end{minipage}
  \caption{\small If one more bet is allowed  given that the precommitted gambler could have played until the end with a sufficiently accumulated gain (loss), she would randomize (stop) with sufficiently large gain (loss) as shown in left panel (right panel).}\label{tree_5}
\end{figure}

What is the intuition behind these results? Standing at $t=0$, the probability of reaching the top most node and winning sufficiently large is very small; hence the effect of exaggeration of this small probability kicks in. Then, given the opportunity of an extra play, tossing a coin to decide is better than not playing at all, for the same reason as entering the casino even if one is allowed to play only once (see Subsection \ref{enterornot}). Moreover, the more gains the stronger probability weighting, and hence more likely to play.  On the other hand,
playing this additional bet without tossing a coin (i.e., {\it definitely} continuing) is not optimal either because of the {\it strict} risk aversion -- randomization helps trigger   probability weighting in large gains which in turn offsets the risk aversion level.
%Suppose the problem has a horizon $T >0$ and $\tau$ is a given exiting strategy. Assume $p_n := \prob(S_{\tau} = n) > 0$,
%$n=0,1,...,T$, the probability of leaving with an accumulated win of  $n$ under strategy $\tau$. Now suppose the time horizon is expanded to $T+1$ so
%gambler is allowed to play at least one more round.

Next, let us examine whether the gambler would like to take one more step
in the {\it loss} region if the horizon is expanded.
Again, suppose $\tau$ is a given exit strategy for the horizon $T>0$,
%
%at $S_n = -n$ to possibly achieve $S_{n+1} = -(n+1)$ or $S_{n+1} = -(n-1)$ , shown in the right panel in Figure \ref{tree_5}.
and denote $ p_{-T} = \prob(S_{\tau} = -T) > 0$. %$n=0,1,...,T$, the probability of having an accumulated loss of $n$ upon leaving  under $\tau$.
%The right panel of Figure \ref{tree_5} shows an example where
%$\tau=n=T$.
Let $q_{-T} = \frac{1-r_{-T}}{2}\in[0,\frac{1}{2}]$, where $r_{-T}$ is the probability to stop at $S_{\tau} = -T$, and
by $\tau''$ the strategy extending the original $\tau$ by, given $\tau=T$ and $S_{\tau} = -T$ (see the right panel of Figure \ref{tree_5} for an illustration), playing one more round with probability $1-r_{-T}$ at time $T$ and stopping at time $T+1$, assuming the horizon is now $T+1$.
Then a similar analysis to the gain case shows that the optimal $q^*_{-T}$ minimizes
\begin{equation}\label{-qstar}
w_-(q_{-T} p_{-T}) [u_-(T+1) - u_-(T)] + w_-\big((1-q_{-T})p_{-T} \big) [u_-(T) - u_-(T-1)].
\end{equation}
Different from the gain region, in the loss region the optimality is achieved by {\it minimizing} a {\it concave} function when $p_{-T}$ is sufficiently small. Hence, the optimal $q^*_{-T}$ is either $0$ or $ \frac{1}{2}$, corresponding to ``stop" or ``continue" respectively. This means that the gambler will {\it not} flip a coin this time. To investigate which is better between ``stop" and ``continue", we calculate the difference between the objective values (\ref{-qstar}) at $q^*_{-T} = 0$ and at $q^*_{-T} = \frac{1}{2}$:
\begin{align*}
& \quad w_-(p_{-T}) [u_-(T) - u_-(T-1)] - w_-(p_{-T}/2) [u_-(T+1) - u_-(T-1)] \\
& = [u_-(T) - u_-(T-1)]  w_-(p_{-T}/2) \left[ \frac{ w_-(p_{-T})}{w_-(p_{-T}/2)} - \frac{u_-(T+1) - u_-(T-1)}{u_-(T) - u_-(T-1)} \right].
\end{align*}
As $T\to \infty$, we have $p_{-T}\to 0$ and, hence,
\begin{align*}
\frac{ w_-(p_{-T})}{w_-(p_{-T}/2)} \to 2^{\delta_-}, \quad \frac{u_-(T+1) - u_-(T-1)}{u_-(T) - u_-(T-1)} = 1 + \frac{u_-(T+1) - u_-(T)}{u_-(T) - u_-(T-1)} \to 2,
\end{align*}
assuming $w_-(\cdot)$ is given by \eqref{eq:distortion} with $0<\delta_-<1$ and $u_-(\cdot)$ has diminishing marginal (dis)utility, namely,
$u_-'(x)\to 0$ as $x\to \infty$ (which holds for \eqref{eq:utility}).
This implies that the value (\ref{-qstar}) at $q^*_{-T} = 0$ is smaller than that at $q^*_{-T} = \frac{1}{2}$, when $T$ is sufficiently large. Consequently, the gambler will choose to stop even if he is offered to play one more round. The intuition is clear:
from the perspective at  $t=0$, the probability of losing  sufficiently big is very small, which is inflated by probability weighting. This inflation outweighs
the risk-seeking in losses because of the diminishing marginal disutility. As a result, the action of stop, which generates zero additional CPT value, is the best because any other action  will only add {\it negative} CPT values.

We have proved the following result.\footnote{We have put the proof of this result here instead of in the appendix, not only because it is relatively elementary, but also because the proof discloses why there are essential differences between the gain and loss regions.}

\begin{theorem}\label{prop:Tone}
%Suppose utility function \eqref{eq:utility} and probability weighting function \eqref{eq:distortion}.
Let $\tau \in \cT_{T}$ be a given strategy.
\begin{enumerate}
\item[{\rm (a)}] Assume that $u_+(\cdot)$ is strictly  concave and  $P(S_{\tau} = T) > 0$. Construct a new strategy $\tau' = \tau'(r_T) := \tau + {\bf 1}_{\tau=T,S_{\tau} = T} \xi_{T,T}$, where $r_T\in[0,1]$, $\xi_{T,T}$ is a Bernoulli random variable that is independent of $S=(S_t : t \in \bN)$ and $\tau$, and $\prob(\xi_{T,T} = 0) = r_T = 1-\prob(\xi_{T,T} = 1)$. Then $\tau' \in  {\cT}_{T+1}$ and,  for sufficiently large $T$, there exists $r_T \in (0,1)$ such that $V(S_{\tau'}) > V(S_{\tau})$ .
\item[{\rm (b)}] Assume that $w_-(\cdot)$ is given by \eqref{eq:distortion} with $0<\delta_-<1$, $u_-'(x)\to 0$ as $x\to \infty$, and $P(S_{\tau} = -T) > 0$. Construct a new strategy $\tau'' = \tau''(r_{-T}) := \tau + {\bf 1}_{\tau=T,S_{\tau} = -T} \xi_{T,-T}$, where $r_{-T}\in[0,1]$, $\xi_{T,-T}$ is a Bernoulli random variable that is independent of $S=(S_t : t \in \bN)$ and $\tau$, and $\prob(\xi_{T,-T} = 0) = r_{-T} = 1-\prob(\xi_{T,-T} = 1)$. Then $\tau'' \in  {\cT}_{T+1}$ and,  for sufficiently large $T$,  $V(S_{\tau''}) < V(S_{\tau})$ for all $r_{-T} \in [0,1)$.
\end{enumerate}
\end{theorem}

For general utility and weighting functions, the above results are valid for sufficiently large $T$; %(and hence also sufficiently large $T$),
but for the utility function \eqref{eq:utility} and probability weighting function \eqref{eq:distortion}
with \cite{TverskyKahneman1992:CPT}'s estimates, $T$ does not need to be excessively large. For example, it follows from the proof of (a) that all we need is to ensure
$p_T=P(S_{\tau} = T)$ falls into the concave domain of $w_+(\cdot)$. For $\delta_+=0.61$, this requires $p_T<0.3$ (refer to the right panel of Figure \ref{fig:distortion_utility}) which is satisfied when $T=2$. Similarly, by the proof of (b), for $\alpha_- = 0.88$ and $\delta_- = 0.69$, a straightforward calculation yields that when $T=2$, $p_{-T}$ falls into the concave domain of $w_-(\cdot)$ and $r_{-T} = 1$ dominates the other choices.
%It is therefore uncertain whether the gambler is willing to take one more step in the loss region when $T$ increases.
%{\color{red}JO: What follows is really interesting, not sure if we want to dig deeper? Expose more?}
%{\color{orange} This example shows that decision-making in gain region and loss region are different by fixing previous $T-1$ nodes. However, it is hard to compare the optimal distribution of $S_\tau$ in general.}

In the preceding discussions we assume that an original  (i.e., before the horizon is extended) strategy has resulted in the maximum possible gain or loss. %$\tau=T$ and $S_{\tau} = T$
We now investigate the situations when the strategy ends up with an intermediate state with a mild accumulated gain or loss.
Specifically, let $\tau \in T_{T}$ be a given exiting strategy and $n < T$ be a gain state.
Assume $p_n: = \prob(S_{\tau} = n) > 0$ and $\bar p_{n+1}: = \prob(S_{\tau} \ge n+1) > 0$ and consider the scenario in which the gambler has reached
the node $(T, n)$ under $\tau$, namely $\tau = T$ and $S_\tau = n$. Now, with an additional round of play granted, we
denote by $\tau'$
the strategy modifying the original $\tau$ by, given
$\tau = T$ and $S_\tau = n$, playing one more round with probability $1 - r_n$ at time $T$, where $r_n \in  [0, 1]$. Let $q_n =\frac{1-r_n}{2}
\in [0,\frac{1}{2}]$.
An argument similar to the case of $n=T$ yields that the extra CPT value due to the possible additional round of play, as a function of $q_n$, is
\begin{equation}\label{der}
w_+(q_n p_n + \bar p_{n+1}) [u_+(n+1) - u_+(n)] + w_+ \big((1-q_n)p_n + \bar p_{n+1} \big) [u_+(n) - u_+(n-1)],
\end{equation}
%The first-order condition stipulates that the optimal $q_n^*$ satisfies
whose first-order derivative is
\begin{equation}\label{inequality}
\begin{array}{ll}
& \quad \left\{ w_+'(q_n p_n + \bar p_{n+1}) [u_+(n+1) - u_+(n)] - w_+' \big((1-q_n)p_n + \bar p_{n+1} \big) [u_+(n) - u_+(n-1)] \right\} p_n \\
& = [u_+(n+1) - u_+(n)]  w_+'((1-q_n)p_n + \bar p_{n+1}) p_n \left[ \frac{w_+'(q_n p_n + \bar p_{n+1})}{w_+'((1-q_n)p_n + \bar p_{n+1})} - \frac{u_+(n) - u_+(n-1)}{u_+(n+1) - u_+(n)} \right].
\end{array}
\end{equation}
The necessary condition for a maximum $q_n^*$ is thus
\begin{equation}\label{qnstar}
\frac{w_+'(q_n^* p_n + \bar p_{n+1})}{w_+' \big((1-q_n^*)p_n + \bar p_{n+1} \big)} =\frac{u_+(n) - u_+(n-1)}{u_+(n+1) - u_+(n)} \;.
\end{equation}
%Assuming $u_+(\cdot)$ is {\it strictly}  concave  (e.g.
%that given by \eqref{eq:utility}), the right hand side of (\ref{qstar}) is strictly greater than one. Hence, the equation is satisfied by some $q^*_T \in (0, \frac{1}{2})$, but {\it not} $q^*_T=0$ (noting $w_+'(0)=+\infty$) or $q^*_T=\frac{1}{2}$. Recall that $q^*_T=0$ and  $q^*_T=\frac{1}{2}$ correspond to $r_T=1$ and $r_T=0$ respectively.
%
%
%derivative is given by
%\begin{align*}
%& \quad \left( w_+'(q_n p_n + \bar p_{n+1}) [u_+(n+1) - u_+(n)] - w_+'((1-q_n)p_n + \bar p_{n+1}) [u_+(n) - u_+(n-1)] \right) p_n \\
%& = [u_+(n+1) - u_+(n)]  w_+'((1-q_n)p_n + \bar p_{n+1}) p_n \left( \frac{w_+'(q_n p_n + \bar p_{n+1})}{w_+'((1-q_n)p_n + \bar p_{n+1})} - \frac{u_+(n) - u_+(n-1)}{u_+(n+1) - u_+(n)} \right) .
%\end{align*}

Assume $n$ is sufficiently large so that $p_n + \bar p_{n+1}\equiv \prob(S_{\tau} \ge n)$ falls into the concave region of $w_+(\cdot)$. Because $\frac{u_+(n) - u_+(n-1)}{u_+(n+1) - u_+(n)} > 1$ due to the strict concavity of $u_+(\cdot)$, $q_n^* = \frac{1}{2}$ will never satisfy (\ref{qnstar}); hence $r_n \ne 0$ or the gambler will not  continue decisively.  %Then $q_n^* \in [0, \frac{1}{2})$.
Moreover, if $\frac{w_+'(\bar p_{n+1})}{w_+'(p_n + \bar p_{n+1})} >\frac{u_+(n) - u_+(n-1)}{u_+(n+1) - u_+(n)}$, then there is $q_n^* \in (0, \frac{1}{2})$ such that $(\ref{qnstar})$ holds, in which case $r_n\in(0,1)$ indicating that the gambler will randomize. On the other hand, if $\frac{w_+'(\bar p_{n+1})}{w_+'(p_n + \bar p_{n+1})} \leq \frac{u_+(n) - u_+(n-1)}{u_+(n+1) - u_+(n)}$
then it follows from (\ref{inequality}) that (\ref{der}) is a non-increasing function of $q_n$; so its maximal value achieves at $q_n=0$ (and hence
$r_n = 1$). This is in stark contrast to the case when $n=T$: at some intermediate gain state $n$, the gambler may indeed choose to stop even if the time horizon is extended.\footnote{This is examplified by the black node (9,1) in the left panel of Figure \ref{tree_10_l2}.}

Finally, at an intermediate loss state $-n > -T$, a similar analysis  yields that randomization with $q_n \in (0,\frac{1}{2})$ is again being dominated. It is possible that $q^*_{-n} = \frac{1}{2}$ (resp, $r_{-n}^* = 0$), in which case the gambler will continue for sure  if the time horizon is extended. This is different from the case of maximal loss state $-n=-T$.
%Note white nodes in loss region close to state 0 before terminal time in Figure \ref{tree_10_g} left panel.

\subsection{Na\"ive gamblers}

While a precommitted gambler follows the optimal strategy determined at time 0, a na\"ive gambler constantly deviates from it. We have shown in Subsection \ref{numerical} that, under the parameter specification therein,  the naivet\'e's actual behavior changes from the originally planned loss-exit strategy to an eventual gain-exit one.

Numerically, the naivit\'e's strategy can be obtained by computing each time-$t$
precommitted strategy, carrying it out for just one period, and then pasting them together; see Subsection \ref{numerical} for details. We apply this scheme to the first three groups of parameters studied in Subsection \ref{adl}, and draw the na\"ive strategies  in the right panels of Figure \ref{tree_10_l} -- \ref{tree_10_l2}.

The problem in Figure \ref{tree_10_l} has the same parameter values as that in Figure \ref{Optimalstrategy} but a longer horizon. The changes from the left panels to the right ones in the two figures are qualitatively the same, namely the naivit\'e
turns a loss-exit strategy to a gain-exit one eventually. %, reminiscent of the disposition effect in security trading.
The same happens to Figure \ref{tree_10_l2}.\footnote{In the right panel of Figure \ref{tree_10_l2}, all the nodes with state $x=1$  are black, which ``block" the gambler from accessing the nodes beyond  state 1. This is why the  nodes above state 1 are also all black. }
In Figure \ref{tree_10_g}, the two panels are almost identical -- both are gain-exit -- except the two lowest nodes at $t=8,9$.
This is because the difference in behaviors of the precommitter and the naivit\'e emanates from time-inconsistency, which in turn stems from probability weighting. In this case, the strength of probability weighting is very low with $\delta_\pm=0.95$,
leading to a low level of time-inconsistency than the other two cases and hence the high similarity between the precommitted and na\"ive  strategies.

It is very interesting to note that, in {\it all} the cases, the na\"ive gambler's behavior is {\it consistent}, irrespective of the underlying parameter specifications: once he enters the casino he {\it always} takes gain-exit strategies, reminiscent of the disposition effect in security trading.
In particular, he never stops loss and gambles ``until the bitter end" \citep{EbertStrack2012:UntilTheBitterEnd}.\footnote{We reiterate that the result of
\citet{EbertStrack2012:UntilTheBitterEnd} depends critically on the assumption that  the gambler can  construct arbitrarily small random payoffs, which is possibly valid only in a continuous-time model. The finding that  ``gamble-until-bitter-end" is also present in the discrete-time casino model suggests that the behavior  is probably  more prevalent characterizing broadly a naivit\'e (be it a gambler or an investor).}

We now provide a theory that explains such a phenomenon.
Suppose a na\"ive gambler has accumulated a gain equal to $x > 0$ at time $T-1$, the date just before the terminal one. Then his decision problem regarding whether he should quit at $T-1$ can be formulated as
\begin{align*}
\max_{q \in [0,1/2] } g(q) := \big( u_+(x+1) - u_+(x) \big) w_+(q) - \big( u_+(x) - u_+(x-1) \big) (1- w_+(1-q)),
\end{align*}
where, as before, $q=\frac{1-r}{2}$ and $r\in[0,1]$ is the probability to stop.
Suppose $w_+$ satisfies the so-called subcertainty, i.e., $1-w_+(1-p) \ge w_+(p)$ for $p \in [0,1/2]$, a property that is proposed by \citet{KahnemanDTverskyA:79pt} and shared by many probability weighting functions including \eqref{eq:distortion}. Then
\begin{align*}
g(q) \le \Big( \big( u_+(x+1) - u_+(x) \big) - \big( u_+(x) - u_+(x-1) \big) \Big) w_+(q) \le 0,
\end{align*}
where the second inequality follows from the concavity of $u_+$, while the equality is achieved when $q = 0$, corresponding to the decision of ``stop''. We have established the following result.
\begin{proposition}\label{prop:naiveg}
Assume that $w_+$ satisfies subcertainty. Then it is optimal for a na\"ive gambler to stop in gain at $T-1$.
\end{proposition}

Next, suppose the naivit\'e' has accumulated a loss $-x < 0$ at $T-1$. His decision problem to continue or stop at $T-1$ is
\begin{align*}
\min_{q \in [0,1/2] } l(q) := \big( u_-(x+1) - u_-(x) \big) w_-(q) - \big( u_-(x) - u_-(x-1) \big) (1- w_-(1-q)).
\end{align*}
Suppose probability weighting function $w_-$ is differentiable and $w_-'(1-p)/w_-'(p) \ge 1$ for $p \in [0,1/2]$, with the left hand side in the sense of limit  for $p = 0$. A straightforward calculation verifies  that this condition is satisfied by the Tversky--Kahneman weighting function \eqref{eq:distortion}. Then
\begin{align*}
l'(q) & = \big( u_-(x+1) - u_-(x) \big) w_-'(q) - \big( u_-(x) - u_-(x-1) \big) w_-'(1-q) \\
& = \big( u_-(x) - u_-(x-1) \big) w_-'(q) \left( \frac{ u_-(x+1) - u_-(x) }{ u_-(x) - u_-(x-1) } - \frac{ w_-'(1-q) }{w_-'(q)} \right) \\
& \le \big( u_-(x) - u_-(x-1) \big) w_-'(q) \left( \frac{ u_-(x+1) - u_-(x) }{ u_-(x) - u_-(x-1) } - 1 \right) \leq  0,
\end{align*}
where the last inequality comes from the concavity of $u_-$. As a result, $l(q)$ is non-increasing in $q\in [0,1/2]$ and the minimum is achieved when $q = 1/2$, corresponding to the ``continue'' decision.
\begin{proposition}\label{prop:naivel}
Assume that $w_-$ is differentiable and $w_-'(1-p)/w_-'(p) \ge 1$ for $p \in [0,1/2]$. Then it is optimal for a na\"ive gambler to continue in loss at $T-1$.
\end{proposition}

A corollary of
Proposition \ref{prop:naivel} is that the naivit\'e will definitely continue even if there is only one round of play left as long as he is in loss, let alone when a longer horizon is allowed. As a consequence, he will not stop loss in any case, until the bitter end.

\subsection{Sophisticated gamblers}

A sophisticated gambler is unable to precommit
and realizes that her future selves will deviate from whatever plans she makes now. Her resolution
is to compromise and choose {\it consistent planning} in the sense that she optimizes taking the
future disobedience as a constraint. Consequently,
strategies of sophisticated gamblers  can be obtained  using backward deduction as in
dynamic programming.

To start, we note that at $T-1$, a sophisticated gambler and a na\"ive one face the same problem; hence we have the following immediate result.

\begin{proposition}\label{prop:soph}
Propositions \ref{prop:naiveg} and \ref{prop:naivel} hold true for a sophisticated gambler.
\end{proposition}

Next, we derive a sophisticated gambler's stopping strategies for the four cases studied in Subsection \ref{adl}, where $T=10$. It turns out that, of the four cases,  she will enter the casino {\it only} in the case when $(\alpha_\pm, \delta_\pm ,\lambda) = (0.5, 0.95, 1.5)$, corresponding to Figure \ref{tree_10_g}. Moreover, her strategy is identical to the one depicted in the right panel of Figure \ref{tree_10_g}, which is  the actual strategy of the na\"ive gambler and close to the precommitted strategy. This is because when $\delta_\pm$ is close to 1, the level of probability weighting is low, hence so is that of time-inconsistency, leading to similar strategies of all the three types of gamblers.

Note that in the case above, the sophisticated gambler takes the gain-exit type of strategy.
Indeed, so long as she enters the casino,  she essentially stops in gain under some mild conditions. This follows from the following argument: by Proposition \ref{prop:soph}, the sophisticated gambler will stop in gain at $T-1$. Knowing this, she will also stop in gain at $T-2$ by virtue of exactly the same reason. Inductively, this leads to an overall gain-exit type of strategy.
%\footnote{\hl An argument same to Proposition \ref{prop:naiveg} is applied inductively if subsequent nodes in ``strict'' gain in the next step are black. If $(t+1,0)$ is white, such an argument does not directly conclude $(t,1)$ is black. }

On the other hand, the sophisticated gambler always stops no later than her na\"ive counterpart does. This is because while the latter solves an optimal stopping problem at every node, the former  solves the {\it same}  problem but with constraints from her future selves' decisions. Hence, if the latter finds that stopping immediately is optimal at a current node, so will the former because the strategy of an immediate stop automatically satisfies the aforementioned constraints.

\begin{proposition}
Under any specification of parameters, a sophisticated gambler stops no later than a na\"ive gambler does.
\end{proposition}

An implication  of this result is that the naivit\'e is at least as risk-taking as the sophisticated, if not more.

\subsection{Finite horizon versus infinite horizon}

This section explores connection between the finite horizon and infinite horizon casino models.

Define
\begin{align*}
{\mathcal{T}}_\infty := \left\{ \tau\in[0,\infty): \tau \text{ is an } (\cF_t)_{t\in \bN}\text{-stopping time}\right\},
\end{align*}
which is the set of admissible stopping strategies (allowing randomization) in the infinite time horizon.
Suppose  $\tau
\in \cT_\infty$ is optimal for the infinite horizon model and achieves a finite CPT value.  Then we have
$$ V(S_{\tau\wedge T})\leq \underset{\sigma \in {\mathcal{T}}_T}{\sup} V(S_\sigma) \leq V(S_\tau).$$
We see immediately that the value of the finite horizon model converges to that  of the infinite horizon one as the horizon approaches infinity. The following makes this
formal.

%We could also ask if the reverse is possible: if taking limits as $T\to \infty$ we may recover the infinite horizon case from the finite horizon case? While we know the value function increases, we do not know if $\tau^*+T$ can be taken to increase a.s. - so that a precommitted gambler with $T+1$ always gambles at least as long as the one with $T$ horizon. This would be equivalent to showing that the optimal potentials are ordered - is it possible to deduce this from the optimisation problem in Sec 4? Maybe this links with the discussion in Sec 5?
%}
%{\color{orange} The numerical examples happen to show that optimal potentials are ordered. To prove it in general, not sure if it is possible to do without solving the finite dimensional programming explicitly.}
%
%The following theorem shows the convergence of optimal value from finite horizon problem to infinite horizon problem under certain conditions.
\begin{theorem}\label{thm:convergence}
Assume $\tau^*$ achieves the optimal value of the gambling model in the infinite time horizon with $\tau^*<\infty$ a.s., $V(S_{\tau^*}) = v^* < \infty$, and $S_{\tau^*}$ is lower-bounded a.s. Then
\[
\lim_{T\to\infty} \underset{\tau \in {\mathcal{T}}_T}{\sup} V(S_\tau) =\lim_{T\to\infty} V(S_{\tau^* \wedge T})=v^*.
\]
\end{theorem}

We stress that this result only reveals the relationship between the two models in terms of the optimal values. It does not offer a solution to the finite horizon problem (which is harder) from a solution to the infinite one (which is comparatively  easier), nor does it tell the error in the optimal values when $T$ is given and fixed. That said, the result suggests that the optimal value of the infinite horizon model is an upper bound of that of the finite horizon one, and it is a {\it tight} upper bound if $T$ is sufficiently large. Moreover, while the truncation method mentioned earlier does not provide an {\it exact} optimal solution to the finite horizon model, it does  nevertheless offer a {\it good} solution when $T$ is large enough.

\section{Conclusion}\label{se:concludingremark}

In this paper we develop a systematic approach to studying the stopping behaviors of CPT gamblers in a finite time horizon. We hope that this work opens an avenue of
thoroughly understanding \citet{Barberis2012:Casino}'s model and beyond.
Indeed, as \citet{Barberis2012:Casino} points out, casino gambling is not an isolated model requiring a unique treatment; rather it is just one of the many examples, including ones in financial markets, that share a common feature of the probability weighting.

%We tackle the problem by first devising the randomized Root stopping time to solve the accompanied Skorokhod embedding problem in a finite time setting so that the decision variable is changed from the stopping time to the probability distribution function of the cumulative gains or losses at the stopped time. Given a distribution, we find the necessary and sufficient conditions under which there exists the stopping time in the finite time random walk, and show the optimality properties satisfied by the randomized Root stopping time, which is the fastest stopping rule to complete embedding all the probability mass.
%
%After establishing the embedding theorem, we reformulate the original optimal gambling problem as a finite-dimension program.
%With the optimal probability obtained from the finite-dimension program, we recover the optimal stopping strategy in the form of randomized Root stopping time.
%We also carry out some numerical studies of the optimal randomized strategies taken by a precommitted gambler and a naive gambler who always changes decision, and compare it with the non-randomized strategies, showing that the latter strategies are overwhelmed by considering randomization.
%In summary, this paper gives an analytical treatment on the optimal casino betting problem in the finite time horizon.

\bibliography{LongTitles,BibFile}

\bigskip

\appendix
\noindent{\Large \bf Appendix}

\section{Proof of Theorem \ref{coro:Rootst}}

We prove this theorem through a series of results. We start by recalling some properties of the potential and its link to the first exit times.

\begin{proposition}\label{prop:st}
Let $\tau$ be an $(\cF_t)_{t\in \bN}$-stopping time such that $\{S_{\tau\wedge t}:t\in \bN\}$ is uniformly integrable. Then
\begin{enumerate}
\item[{\rm (i)}] For any $t \in \mathbb{N}$, $U_{S_{\tau \wedge t}}$ is a convex function, $U_{S_\tau}(x)\geq U_{S_{\tau \wedge t}}(x) \geq |x|$ $\forall x\in\bR$, with $U_{S_{\tau \wedge t}}(x) = |x|$ $\forall x \notin (-t,t)$.
\item[{\rm (ii)}] For any two integers $a<b$ and $\rho:= \inf\{u\geq \tau : S_u \notin (a,b)\}$, $U_{S_\rho}(x)=U_{S_\tau}(x)$ $\forall x\notin (a,b)$, and $U_{S_\rho}$ is linear on $[a,b]$.
\item[{\rm (iii)}] Fix $t\geq 1$ and  let $\mathcal{K} := \{k \in \mathbb{Z} | k = t - 1 + 2j, \; j \in \mathbb{Z} \}$. Then
\begin{align}\label{eq:stevolution}
U_{S_{\tau \wedge t}}(x) = U_{S_{\tau \wedge (t-1)}}(x) + \mathbb{P}(S_{t-1} = x, \tau \geq t) {\bf 1}_{x \in \mathcal{K}} \;\; \forall x \in \mathbb{Z}.
\end{align}
In particular, if $t$ is odd, then $U_{S_{\tau \wedge t}}(x) = U_{S_{\tau \wedge (t-1)}}(x)$ for any odd $x$; and  if $t$ is even, then $U_{S_{\tau \wedge t}}(x) = U_{S_{\tau \wedge (t-1)}}(x)$ for any even $x$.
%if $t$ is odd, $U_{S_{\tau \wedge t}}(x) = U_{S_{\tau \wedge (t+1)}}(x)$ for any even $x$; if $t$ is even, $U_{S_{\tau \wedge t}}(x) = U_{S_{\tau \wedge (t+1)}}(x)$ for any odd $x$.
\end{enumerate}
\end{proposition}

\begin{proof}%[Proof of Proposition \ref{prop:st}]
The first two properties are standard; see \citet[Section 2]{Obloj2004:SkorokhodEmbedding}. So we only establish (iii). Note that $S_{t-1}$ is supported on $\mathcal{K}$. We have $|S_{\tau\wedge t}-S_{\tau\wedge (t-1)}|\leq 1$ so that $\{S_{\tau\wedge t}\geq x\}=\{S_{\tau\wedge (t-1)}\geq x\}$  $\forall x\notin \mathcal{K}$. In particular, since $S$ is a martingale, we have $U_{S_{\tau \wedge t}}(x) = U_{S_{\tau \wedge (t-1)}}(x)$  $\forall x\notin \mathcal{K}$. Now take $x\in \mathcal{K}$. Since $x+1,x-1\notin \mathcal{K}$, using \eqref{eq:potential_rep}, we have
\begin{equation}
 \begin{split}
 \prob(S_{\tau\wedge t}=x) &= \frac{U_{S_{\tau\wedge t}}(x+1)+ U_{S_{\tau\wedge t}}(x-1)}{2}- U_{S_{\tau\wedge t}}(x)\\
 & =\frac{U_{S_{\tau\wedge (t-1)}}(x+1)+ U_{S_{\tau\wedge (t-1)}}(x-1)}{2}- U_{S_{\tau\wedge t}}(x)\\
 &=  \prob(S_{\tau\wedge (t-1)}=x) + U_{S_{\tau\wedge (t-1)}}(x) - U_{S_{\tau\wedge t}}(x).
\end{split}
\end{equation}
Rearranging and observing that $\prob(S_t=x,\tau\geq t)=0$ the thesis follows.
\end{proof}

The following proposition provides some useful properties of $U_t$.
%\footnote{{\bf How is it useful? There is no explanation below. Also, do we know $U_t$ becomes $U_\mu$ when $t$ is large enough?}}

\begin{proposition}\label{prop:Rootst}
Let $\mu\in \cP_0(\bZ)$ and $U_t=U_t^\mu$. Then
\begin{enumerate}
\item[{\rm (i)}] $U_0(x)\leq U_t(x) \leq U_{\mu}(x)\wedge U_{S_t}(x)$ $\forall  x \in \mathbb{Z}$, $t \in \bN$.
\item[{\rm (ii)}] $U_t(x) = U_{t+1}(x)$ when $t$ is odd and $x$ is even, or when $t$ is even and $x$ is odd.
\item[{\rm (iii)}] $U_t(x)$ is convex in $x\in \bR$ and non-decreasing in $t \in \bN$.
%\item[(iv)] Given $n$, $U_n(x) = |x|$ for all $x \geq n$ and $x \leq -n$.
\end{enumerate}
\end{proposition}

\begin{proof}%[Proof of Proposition \ref{prop:Rootst}]
(i) By the construction of $U_t$ we have $U_0(x)\leq U_t(x) \leq U_{\mu}(x)$.
On other hand, by \eqref{eq:potential_rep} and the structure of SSRW that $\prob(S_t = x) = \big( \prob(S_{t-1} = x-1) + \prob(S_{t-1} = x+1) \big)/2$, one can show easily that $U_{S_t}(x) = \big( U_{S_{t-1}}(x-1) + U_{S_{t+1}}(x+1) \big)/2$, $x \in \mathbb{Z}$.
Then by induction, we have $U_t(x) \leq U_{S_t}(x)$.

(ii) Again,  by construction we have $U_0(x) = U_1(x)$ for all odd $x$ and $U_1(x) = U_2(x)$ for all even $x$. The conclusions follow immediately from induction.
%Suppose for $t \leq n-1$, we have proved that when $t$ is odd, $U_t(x) = U_{t+1}(x)$ for all even $x$, and when $t$ is even, $U_t(x) = U_{t+1}(x)$ for all odd $x$.  If $n$ is odd and $x$ is even, $U_n(x) = \frac{U_{n-1}(x-1) + U_{n-1}(x+1)}{2} \wedge U_{\mu}(x) = \frac{U_n(x-1) + U_n(x+1)}{2} \wedge U_{\mu}(x) = U_{n+1}(x)$; similarly, $U_n(x) = U_{n+1}(x)$ if $n$ is even and $x$ is odd.

(iii) Clearly $U_0$ is convex. Suppose $U_t$ is convex and fix $m\in \bZ$. If we put $\tilde U(x)=U_t(x)$ for $x\in \bZ\setminus \{m\}$, pick any
$$\tilde U(m)\in \left[U_t(m), \frac{1}{2}(U_t(m-1)+U_t(m+1))\right],$$
and finally define $\tilde U$ by a linear interpolation for $x\in \bR$, then $\tilde U$ is convex.
Observe that $U_{t+1}$ is obtained exactly by repeating this procedure for all $m\in \bZ$ and, hence, is also convex. Moreover, it now follows, by its definition, that $U_t(x)$ is non-decreasing in $t$.
\end{proof}

\begin{proposition}\label{thm:rootstoppingtime}
Let $T\geq 1$, $\mu \in \cP_0(\bZ)$ such that $\mu([-T,T]) = 1$ and $U_t=U_t^\mu$ be  defined in \eqref{eq:differenceequation}. Then, the following are equivalent:
\begin{enumerate}
\item[{\rm (i)}] $U_T(x) = U_{\mu}(x)$ $\forall x \in \mathbb{Z} $.
\item[{\rm (ii)}] There exists a randomized Root stopping time $\tau_R({\bf b}, {\bf r})$ such that $\tau_R({\bf b}, {\bf r}) \leq T$ and $U_{S_{\tau_R({\bf b}, {\bf r})\wedge t}}=U_t$ $\forall t\leq T$; in particular $S_{\tau_R({\bf b}, {\bf r})} \sim \mu$.
\item[{\rm (iii)}] There exists $\tau\in {\cal T}_T$ such that $S_\tau \sim \mu$.
%\item[(ii)] $\inf\{n \geq 0 :U_n(x) \geq U_{\mu}(x), \forall x \in \mathbb{Z} \} \leq T$.
%\item[(iv)] $U_{\mu}(x) \leq \frac{U_{T-1}(x-1)+U_{T-1}(x+1)}{2}$, $\forall x \in \mathbb{Z} $.
\end{enumerate}
Furthermore, for any $\tau\in {\cal T}_T$ such that $S_\tau \sim \mu$ we have $U_{S_{\tau \wedge t}}(x) \leq U_{t}(x)$ $\forall x\in \bR$, $t\leq T$.
\end{proposition}

\begin{proof}%[Proof of Theorem \ref{thm:rootstoppingtime}]
Proof of (i) $\to$ (ii).
To show the existence of a randomized Root stopping time embedding $\mu$ we first construct its stopping barrier ${\bf b}$.
% Let $r_\mu=\inf\{x\in \bZ: \mu((x,\infty))=0\}\leq T$ and $l_\mu=\sup\{x\in \bZ: \mu(-\infty,x)=0\}\geq -T$. For $x\notin \bZ\cap (l_\mu,r_\mu)$, set $b(x)=|x|$ and $\xi_{t,x}
For $x \in \mathbb{Z}$, define
\begin{align}\label{eq:Rootstboundary}
b(x) := \inf\{t \geq |x|: U_{t+1}(x) = U_{\mu}(x)\}.
\end{align}
It follows from Proposition \ref{prop:Rootst} that $b(x)= x + 2k$ for some $k \in \mathbb{Z}$.
Next define the probabilities of the binary random variables $\{\xi_{t,x}\}$, $\mathbb{P}(\xi_{t,x}  = 0) = 1- \mathbb{P}(\xi_{t,x} = 1)$. For each $x \in \mathbb{Z}$,
\begin{align}\label{eq:Rootstrandomization}
\begin{cases}
\mathbb{P}(\xi_{t,x}  = 0) = 0 & \text{ for } t < b(x) \;, \\
\mathbb{P}(\xi_{t,x}  = 0) = r(x) := \frac{ U_{t}(x-1) + U_{t}(x+1) - 2 U_{\mu}(x)}{ U_{t}(x-1) + U_{t}(x+1) - 2 U_{t}(x)} & \text{ for } t = b(x) \;, \\
\mathbb{P}(\xi_{t,x}  = 0) = 1 & \text{ for } t > b(x) \;.
\end{cases}
\end{align}
Note that $r(x)=1$ is only possible if $U_\mu(x)=|x|$ which happens for $x$ outside of the support of $\mu$. For other $x$ we have $U_\mu(x)> |x|$ and a randomization, i.e., $0<r(x)<1$, happens at a node $(t,x)$ when $t = b(x)$ and
\begin{align*}
\frac{U_{t}(x-1) + U_{t}(x+1)}{2} > U_{\mu}(x)> U_t(x) \;.
\end{align*}
%In particular, by Proposition \ref{prop:Rootst}, $x$ and $b(x)$ have the same parity.\footnote{{\bf We haven't defined ``parity".}}
Let $\tau=\tau_R({\bf b}, {\bf r})$ be the randomized Root stopping time in \eqref{eq:Rootdef}. By (i), $U_T \geq U_{\mu}$ and hence $b(x) \leq T-1$ $\forall x \in \mathbb{Z} \cap (-T,T)$. It follows that $\tau \leq T$ as required.

To show $S_\tau \sim \mu$, we need only to establish $U_{S_{\tau}} (x) = U_{\mu}(x)$ $\forall x \in \mathbb{Z}$. Note $U_{S_0}(x) = U_{S_{\tau \wedge 0}}(x) = U_0(x)$. Suppose we have $U_{S_{\tau \wedge t}}(x) = U_{t}(x)$ for $t \leq n-1$. It follows from  \eqref{eq:potential_rep} that
\begin{align*}
\mathbb{P}(S_{\tau \wedge t} = x) = \frac{U_{t}(x+1) + U_{t}(x-1)}{2} - U_{t}(x),\;\;t \leq n-1.
\end{align*}
On the other hand, by Proposition \ref{prop:st}, we have
\begin{align*}
U_{S_{\tau \wedge n}}(x) = U_{S_{\tau \wedge (n-1)}}(x) + \mathbb{P}(S_{n-1} = x, \tau \geq n) {\bf 1}_{x \in \mathcal{K}} \;\forall x \in \mathbb{Z},
\end{align*}
where $\mathcal{K} = \{k \in \mathbb{Z} | k = n - 1 + 2j, j \in \mathbb{Z} \}$.

If $U_{n-1}(x) = U_{\mu}(x) = U_{n}(x)$, then $b(x) < n-1$ and $\mathbb{P}(S_{n-1} = x, \tau \geq n) = 0$; hence
\begin{align*}
U_{S_{\tau \wedge n}}(x) = U_{S_{\tau \wedge (n-1)}}(x) = U_{n-1}(x) = U_{\mu}(x) = U_{n}(x) \;.
\end{align*}
If $U_{n-1}(x) < U_{\mu}(x) = U_{n}(x)$, then $b(x) = n-1$ and necessarily $x \in \mathcal{K}$. We have, by definition,
\begin{align*}
\mathbb{P}(S_{n-1} = x, \tau \geq n) & = \mathbb{P}(S_{\tau \wedge (n-1)} = x) \mathbb{P}(\xi_{n-1,x} = 1) \\
& = \left(\frac{U_{n-1}(x+1) + U_{n-1}(x-1)}{2} - U_{n-1}(x) \right) \mathbb{P}(\xi_{n-1,x} = 1).
\end{align*}
It then follows that
\begin{align*}
U_{S_{\tau \wedge n}}(x) & = U_{S_{\tau \wedge (n-1)}}(x) + \mathbb{P}(S_{n-1} = x, \tau \geq n) {\bf 1}_{x \in \mathcal{K}} \\
& = U_{n-1}(x) \mathbb{P}(\xi_{n-1,x} = 0) + \frac{U_{n-1}(x+1) + U_{n-1}(x-1)}{2} \mathbb{P}(\xi_{n-1,x} = 1) \\
& = U_{n-1}(x) \frac{ U_{n-1}(x-1) + U_{n-1}(x+1) - 2 U_{\mu}(x)}{ U_{n-1}(x-1) + U_{n-1}(x+1) - 2 U_{n-1}(x)} \\
& \quad + \frac{U_{n-1}(x+1) + U_{n-1}(x-1)}{2} \frac{2 U_{\mu}(x) - 2 U_{n-1}(x)}{ U_{n-1}(x-1) + U_{n-1}(x+1) - 2 U_{n-1}(x)} \\
& = U_{\mu}(x) = U_n(x) \;.
\end{align*}
Finally, if $U_n(x) < U_{\mu}(x)$, then $b(x) > n-1$. By definition, we have $U_n(x) = \frac{U_{n-1}(x+1) + U_{n-1}(x-1)}{2}$ and $\prob(S_\tau=x)=\prob(S_\tau=x, \tau\geq n)$. Consequently,
\begin{align*}
\mathbb{P}(S_{n-1} = x, \tau \geq n) & = \mathbb{P}(S_{\tau \wedge (n-1)} = x) = \frac{U_{n-1}(x+1) + U_{n-1}(x-1)}{2} - U_{n-1}(x) \;.
\end{align*}
Thus, if $x \in \mathcal{K}$, then
\begin{align*}
U_{S_{\tau \wedge n}}(x)  & = U_{S_{\tau \wedge (n-1)}}(x) + \mathbb{P}(S_{n-1} = x, \tau \geq n) {\bf 1}_{x \in \mathcal{K}} \\
& = U_{n-1}(x) + \frac{U_{n-1}(x+1) + U_{n-1}(x-1)}{2} - U_{n-1}(x) \\
& = \frac{U_{n-1}(x+1) + U_{n-1}(x-1)}{2} = U_n(x).
\end{align*}
If $x \notin \mathcal{K}$, then, noting that $\prob(S_\tau=x,\tau<n)=0$, we have $\mathbb{P}(S_{\tau \wedge (n-1)} = x) = 0$. As a result, $\frac{U_{n-1}(x+1) + U_{n-1}(x-1)}{2} = U_{n-1}(x)$ and
\begin{align*}
U_{S_{\tau \wedge n}}(x)  & = U_{S_{\tau \wedge (n-1)}}(x) + \mathbb{P}(S_{n-1} = x, \tau \geq n) {\bf 1}_{x \in \mathcal{K}} \\
& = U_{n-1}(x) = \frac{U_{n-1}(x+1) + U_{n-1}(x-1)}{2} = U_n(x) \;.
\end{align*}

In summary, $U_{S_{\tau \wedge n}} (x) = U_{n}(x)$ $\forall n \in \mathbb{Z}^{+}$. As a result, $U_{S_{\tau}}(x) = U_{S_{\tau \wedge T}}(x) = U_{T}(x)  = U_{\mu}(x)$ $\forall x \in \mathbb{Z}$, namely,  $S_\tau \sim \mu$.

Proof of  (ii) $\to$ (iii). This is trivial.

Proofs of  (iii) $\to$ (i)  and the last assertion of the theorem. We start with the latter assuming (iii) holds.
Let $\tau\in {\cT}_T$ such that $S_\tau\sim \mu$. Note that $U_{S_{\tau \wedge 0}} \equiv U_{0}$. Suppose $U_{S_{\tau \wedge t}}(x) \leq U_{t}(x)$ $\forall x$, for some $t < T$. Let $\tilde S_t = |S_{\tau \wedge t} - x|$, then $(\tilde S_t: t \geq 0)$ is a submartingale. Hence, $U_{S_{\tau \wedge 0}}(x) \leq ... \leq U_{S_{\tau \wedge (t-1)}}(x) \leq U_{S_{\tau \wedge t}}(x) \leq ... \leq U_{S_{\tau \wedge T}}(x) = U_{\mu}(x)$ $\forall x$. By \eqref{eq:stevolution}, if $x \notin \mathcal{K}$, then $U_{S_{\tau \wedge t}}(x) = U_{S_{\tau \wedge (t-1)}}(x) \leq U_{t-1}(x) \leq U_t(x)$; if $x \in \mathcal{K}$ and $U_t(x) = U_{\mu}(x)$, then $U_{S_{\tau \wedge t}}(x) \leq U_{\mu}(x) = U_t(x)$; and if $x \in \mathcal{K}$ and $U_t(x) < U_{\mu}(x)$, then
\begin{align*}
U_{S_{\tau \wedge t}}(x) & \leq \frac{U_{S_{\tau \wedge t}}(x-1) + U_{S_{\tau \wedge t}}(x+1)}{2} \\
& = \frac{U_{S_{\tau \wedge (t-1)}}(x-1) + U_{S_{\tau \wedge (t-1)}}(x+1)}{2} \leq \frac{U_{t-1}(x-1) + U_{t-1}(x+1)}{2} = U_t(x) \;,
\end{align*}
where the first inequality is due to the convexity of $U_{S_{\tau \wedge t}}(\cdot)$, and the second equality is due to $x-1, x+1 \notin \mathcal{K}$. This proves the last assertion of the theorem. Next,
taking $t=T$ and noting that $\tau\leq T$ we have $U_\mu=U_{S_{\tau\wedge T}}\leq U_T$ which shows (iii) $\to$ (i).
\end{proof}

We are now ready to prove Theorem \ref{coro:Rootst}.
%\begin{proof}%[Proof of Corollary \ref{coro:Rootst}]
The ``only if" part follows immediately from Proposition \ref{thm:rootstoppingtime}-(i) and the construction of $U_T(x)$. To prove the ``if" part, supposed (\ref{3}) holds. First, we have $U_T(x)=U_\mu(x)=|x|$ for $|x|\geq T$. For $x = -(T-2),-(T-4),...,T-4,T-2$, it follows from (\ref{3}) that $U_T(x) = \frac{U_{T-1}(x+1)+U_{T-1}(x-1)}{2} \wedge U_\mu(x) = U_\mu(x)$. Next, by Proposition \ref{prop:Rootst}, $U_T(x) = U_{T-1}(x)$ for all $x$ with $x = T+2j$ for some $j \in \mathbb{Z}$. As a result, for $x = -(T-1),-(T-3),...,T-3,T-1$, we have $U_\mu(x+1) = U_T(x+1) = U_{T-1}(x+1)$, $U_\mu(x-1) = U_T(x-1) = U_{T-1}(x-1)$, and, hence, $U_{\mu}(x) \leq \frac{U_\mu(x+1)+U_\mu(x-1)}{2} = \frac{U_{T-1}(x+1)+U_{T-1}(x-1)}{2}$, where the first inequality is due to the convexity of $U_\mu$, and it follows that there exists the randomized Root stopping time that embeds $\mu$ in the random walk with finite time $T$. We conclude that $U_T(x)\geq U_\mu(x)$  $\forall x\in \bZ$ and, hence, Proposition  \ref{thm:rootstoppingtime} yields the desired result.
%\end{proof}

\vspace{2ex}
\section{Proof of Proposition \ref{prop:onehrand}}

Suppose at time 0, the gambler takes a randomized strategy with probability $r$ of ``stop" and probability $1-r$ of ``continue",  where $r \in [0,1]$. Let $q = (1-r)/2 \in [0,1/2]$. %Then $q = 0$ just corresponds to non-gamble.
With utility function $u(x) = u_+(x) {\bf 1}_{x \ge 0} - \lambda u_-(-x) {\bf 1}_{x < 0}$, the CPT value of this strategy is given by $u_+(1) w_+(q) - \lambda u_-(1) w_-(q)$, whose derivative in $q$ is $u_+(1) w_+'(q) - \lambda u_-(1) w_-'(q)$.
If follows from the assumption  $\lim_{p \to 0} [w_+'(p)/w_-'(p)] > \lambda [u_-(1)/u_+(1)]$ that $u_+(1) w_+(q) - \lambda u_-(1) w_-(q)$ is {\it strictly} increasing in $q \in [0,\tilde q]$ for some $\tilde q \in (0,1/2]$. Hence, there exists  $\bar q > 0$ such that $u_+(1) w_+(\bar q) - \lambda u_-(1) w_-(\bar q) > 0$.
%Note that $\bar q > 0$ yields $r<1$, which implies that a randomized strategy is strictly preferred to non-gamble.

\vspace{2ex}
\section{Proof of Theorem \ref{thm:convergence}}

For any $T$, we have
\begin{align*}
V(S_{\tau^* \wedge T})\leq \underset{\tau \in {\mathcal{T}}_T}{\sup} V(S_\tau) \leq V(S_{\tau^*}) = v^*.
\end{align*}
%where $\tau^*$ is optimal for the infinite horizon problem.
Since $S_{\tau^*}$ is lower-bounded a.s., there exists $N>0$ such that $S_{\tau^*} > -N$ a.s. For any $\epsilon > 0$, we can choose $M$ large enough such that
\begin{align*}
& \sum_{n=1}^M u_+(n)\left(w_+\left(\mathbb{P}(S_{\tau^*} \geq n)\right)-w_+\left(\mathbb{P}(S_{\tau^*} \geq n+1)\right)\right)\\
    & \quad - \lambda \sum_{n=1}^N u_-(n)\left(w_-\left(\mathbb{P}(S_{\tau^*} \leq -n)\right)-w_-\left(\mathbb{P}(S_{\tau^*} \leq -n-1)\right)\right) =: \tilde v > v^* - \epsilon/2.
\end{align*}
On the other hand, since $\tau^*$ is finite a.s., the distribution of $S_{\tau^* \wedge T}$ converges to that of $S_{\tau^*}$. Then there is sufficiently large $T$ such that
\begin{align*}
V(S_{\tau^* \wedge T}) \ge & \sum_{n=1}^M u_+(n)\left(w_+\left(\mathbb{P}(S_{\tau^* \wedge T} \geq n)\right)-w_+\left(\mathbb{P}(S_{\tau^* \wedge T} \geq n+1)\right)\right)\\
    & \quad - \lambda \sum_{n=1}^N u_-(n)\left(w_-\left(\mathbb{P}(S_{\tau^* \wedge T} \leq -n)\right)-w_-\left(\mathbb{P}(S_{\tau^* \wedge T} \leq -n-1)\right)\right) \\
    & > \tilde v - \epsilon/2
    > v^* - \epsilon.
\end{align*}
This establishes the desired result.
%Hence, both $V(S_{\tau^* \wedge T})$ as well as $\underset{\tau \in \widetilde{\mathcal{T}}_T}{\sup} V(S_\tau)$ converge to $V(S_{\tau^*})$ as $T$ increases to infinity.

\end{document}